\documentclass[reqno]{conm-p-l}
\usepackage{amssymb,graphics,graphicx,bbm,hyperref,color,longtable}

\def\be{\begin{equation}}
\def\ee{\end{equation}}
\def\bm{\begin{multline}}
\def\bfig{\begin{figure}[htb]}
\def\efig{\end{figure}}
\newcommand{\dd}{{\rm d}}
\newcommand{\e}[1]{\,{\rm e}^{#1}\,}
\newcommand{\ii}{{\rm i}}
\def\Tr{{\operatorname{Tr\,}}}
\newcommand{\sumtwo}[2]{\sum_{\substack{#1 \\ #2}}}

\numberwithin{equation}{section}
\newtheorem{theorem}{Theorem}[section]
\newtheorem{proposition}[theorem]{Proposition}
\newtheorem{lemma}[theorem]{Lemma}
\newtheorem{corollary}[theorem]{Corollary}
\newtheorem{conjecture}{Conjecture}[section]

\newcommand{\caC}{{\mathcal C}}
\newcommand{\caE}{{\mathcal E}}
\newcommand{\caH}{{\mathcal H}}
\newcommand{\caL}{{\mathcal L}}
\newcommand{\caV}{{\mathcal V}}
\newcommand{\bbC}{{\mathbb C}}
\newcommand{\bbE}{{\mathbb E}}
\newcommand{\bbP}{{\mathbb P}}
\newcommand{\bbR}{{\mathbb R}}
\newcommand{\bbZ}{{\mathbb Z}}
\newcommand{\bss}{{\boldsymbol s}}

\newcommand{\N}{\mathbb{N}}
\newcommand{\Z}{\mathbb{Z}}

\newcommand{\R}{\mathbb{R}}
\newcommand{\Prob}{\mathop{\mathbb{P}}\nolimits}
\newcommand{\E}{\mathop{\mathbb{E}}}
\newcommand{\BorelSets}{\mathcal{B}}
\newcommand{\cF}{\mathcal{F}}

\newcommand{\Indi}[1]{\mathop{\mathbbm{1}_{#1}}\nolimits}

\newcommand{\Leb}{\mathrm{Leb}}

\newcommand{\eqdef}{:=}

\newcommand{\PSs}{\Delta_{1}}
\newcommand{\Measures}{\mathcal{M}}
\newcommand{\PDlaw}{\mathrm{PD}}

\newcommand{\SwapConfigs}{\Omega}

\newcommand{\Cycles}{\mathcal{C}}
\newcommand{\Loops}{\mathcal{L}}

\newcommand{\addremswapproc}{\mathfrak{X}}

\newcommand{\contactswaps}{\mathfrak{B}}
\newcommand{\contactzone}{\mathfrak{C}}

\newtheorem{definition}{Definition}[section]

\newcommand{\IdentityOperator}{\mathrm{\texttt{Id}}}

\newcommand{\magicword}[1]{\emph{#1}}

\makeatletter
  \def\tagform@#1{\maketag@@@{\tiny{(#1)}\@@italiccorr}}
\makeatother

\renewcommand{\eqref}[1]{(\ref{#1})}

\begin{document}

\title[Heisenberg models and their probabilistic representations]{Quantum Heisenberg models and their probabilistic representations}

\author[C. Goldschmidt]{Christina Goldschmidt}
\address{Departments of Mathematics and Statistics, University of Warwick,
Coventry, CV4 7AL, United Kingdom}
\email{c.a.goldschmidt@warwick.ac.uk}
\email{daniel@ueltschi.org}
\email{peter.windridge@warwick.ac.uk}

\author[D. Ueltschi]{Daniel Ueltschi}

\author[P. Windridge]{Peter Windridge}



\begin{abstract}
These notes give a mathematical introduction to two seemingly unrelated topics: (i) quantum spin systems and their cycle and loop representations, due to T\'oth and Aizenman-Nachtergaele; (ii) coagulation-fragmentation stochastic processes. These topics are nonetheless related, as we argue that the lengths of cycles and loops effectively perform a coagulation-fragmentation process. This suggests that their joint distribution is Poisson-Dirichlet. These ideas are far from being proved, but they are backed by several rigorous results, notably of Dyson-Lieb-Simon and Schramm.
\end{abstract}

\thanks{Work partially supported by EPSRC grant EP/G056390/1.}
\thanks{\copyright 2011 by the authors. This paper may be reproduced, in its
entirety, for non-commercial purposes.}

\maketitle

\setcounter{tocdepth}{2}
\tableofcontents

\section{Introduction}
\label{sec intro}

We review cycle and loop models that arise from quantum Heisenberg spin systems.  The loops
and cycles are geometric objects defined on graphs.  The main goal is to understand 
properties such as their length in large graphs. 

The cycle model was introduced by T\'oth as a probabilistic representation of the Heisenberg ferromagnet \cite{MR1224836}, while the loop model is due to Aizenman and Nachtergaele and is related to the Heisenberg antiferromagnet \cite{MR1288152}. Both models are built on the random stirring process of Harris \cite{MR0307392} 
and have an additional geometric weight of the form $\vartheta^{\#\rm cycles}$ or $\vartheta^{\#\rm loops}$ with parameter $\vartheta=2$. Recently, Schramm studied the cycle model on the complete graph and with $\vartheta=1$ (that is, without this factor) \cite{MR2166362}. He showed in particular that cycle lengths are generated by a split-merge process (or ``coagulation-fragmentation'' process), and that the cycle lengths have Poisson-Dirichlet distribution with parameter 1.

The graphs of physical relevance are regular lattices such as $\bbZ^{d}$ (or large finite boxes in $\bbZ^{d}$), and the factor $2^{\#\rm objects}$ needs to be present. What should we expect in this case? A few hints come from the models of spatial random permutations, which also involve one-dimensional objects living in higher dimensional spaces. The average length of the longest cycle in lattice permutations was computed numerically in \cite{MR2360227}. In retrospect, this suggests that the cycle lengths have the Poisson-Dirichlet distribution. In the ``annealed'' model where positions are averaged, this was proved in \cite{BU}; the mechanisms at work there (i.e., Bose-Einstein condensation and non-spatial random permutations with Ewens distribution), however, seem very specific.

We study the cycle and loop models in $\bbZ^{d}$ with the help of a stochastic process whose invariant measure is identical to the original measure with weight $\vartheta^{\#\rm cycles}$ or $\vartheta^{\#\rm loops}$, and which leads to an effective split-merge process for the cycle (or loop) lengths. The rates at which the splits and the merges take place depends on $\vartheta$. This allows us to identify the invariant measure, which turns out to be Poisson-Dirichlet with parameter $\vartheta$. While we cannot make these ideas mathematically rigorous, they are compatible with existing results.

As mentioned above, cycle and loop models are closely related to Heisenberg models.  In particular,
 the cycle and loop geometry is reflected in some important quantum observables.  These observables 
 have been the focus of intense study by mathematical and condensed matter physicists,  who have 
 used imagination and clever observations to obtain remarkable results in the last few decades.  Most relevant 
 to us is the theorem of Mermin and Wagner about the absence of magnetic order in one and two dimensions \cite{Mermin:1966lr}, and the theorem of Dyson, Lieb, and Simon, about the existence of magnetic order in the antiferromagnetic model in dimensions 3 and more \cite{MR0496246}. We review these results and explain
 their implications for cycle and loop models.

Many a mathematician is disoriented when wandering in the realm of quantum spin systems. The landscape of $2 \times 2$ matrices and finite-dimensional Hilbert spaces looks safe and easy. Yet, the proofs of many innocent statements are elusive, and one feels quickly lost. It has seemed to us a useful task to provide a detailed introduction to the Heisenberg models in both their quantum and statistical mechanical aspects. We require various concepts from stochastic process theory, and will need to describe carefully the split-merge mechanisms and the Poisson-Dirichlet distribution. The last two are little known outside of probability and are not readily accessible to mathematical physicists and analysts, since the language and the perspective of those domains are quite different (see e.g.\ the dictionary of \cite{MR1681462}, p.\ 314, between analysts' language and probabilists' ``dialect''). In these notes, we have attempted to introduce these different notions in a self-contained fashion.

\subsection{Guide to notation}

The following objects play a central r\^ole.  


\begin{longtable}{l p{0.75\textwidth} }
$\Lambda = (\caV,\caE)$ & A finite graph with undirected edges. \\
$S_{x}^{(j)}$, $\vec S_{x}$ & Spin operators (\S\ref{sec spins}). \\
$\langle \cdot \rangle_{\Lambda,\beta,h}$ & Gibbs state (\S\ref{sec Gibbs states}). \\
$Z_{\Lambda}$, $F_{\Lambda}$ & Partition function and free energy (\S\ref{sec Gibbs states}). \\
$\rho_{\Lambda,\beta}(d \omega)$  & Probability measure for Poisson point processes on $[0,\beta]$ ($\beta > 0$) 
attached to each edge of $\Lambda$ (defined in \S\ref{sec Poisson conf}). \\
$\Cycles(\omega), \Loops(\omega)$ & Cycle and loop configurations constructed from edges in $\omega$ (\S\ref{sec Poisson conf}). \\
$\gamma$  &  Cycle in $\Cycles(\omega)$ or loop in $\Loops(\omega)$. \\
$\vartheta > 0$ & Geometric weight involving the number of cycles and loops.\\
$\Lambda_{n} = (\caV_{n},\caE_{n})$ & Box $\{1,\dots,n\}^{d}$ in $\bbZ^{d}$ with nearest-neigbor edges (\S\ref{sec thermolim}). \\
$m^{*}_{\rm th}$, $m^{*}_{\rm res}$, $m^{*}_{\rm sp}$ & Various definitions of the magnetization (\S\ref{sec thermo limit ferro}). \\
$\overline\sigma(\beta)$ & Antiferromagnetic long-range order (\S\ref{sec thermo limit anti}). \\
$\eta_{\infty}, \eta_{\rm macro}$ & Fractions of vertices in infinite or macroscopic cycles/loops (\S\ref{sec prob phase transitions}). \\
$\PSs$ & Countable partitions of $[0,1]$ with parts in decreasing order, i.e. $\{p_{1} \geq p_{2} \geq \ldots \geq 0: \sum_{i}p_{i} = 1\}$. \\
$\PDlaw_{\theta}$ & Poisson-Dirichlet distribution with parameter $\theta$ on $\PSs$ (\S \ref{s:PoissDiricDefn}).\\
$(\addremswapproc_{t}, t \geq 0)$ & Stochastic process with invariant measure given by our cycle and loop models (\S\ref{s:bridgedynamics}).  \\
\end{longtable}

\section{Hilbert space, spin operators, Heisenberg Hamiltonian}
\label{sec spin basics}

We review the setting for quantum lattice spin systems described by Heisenberg models. Spin systems are relevant for the study of electronic properties of condensed matter. Atoms form a regular lattice and they host localized electrons, which are characterized only by their spin. Interactions are restricted to neighboring spins. One is interested in equilibrium properties of large systems. There are two closely related quantum Heisenberg models, which describe ferromagnets and antiferromagnets, respectively. The material is standard and the interested reader is encouraged to look in the references \cite{MR0289084,MR1239893,MR2310198,MR2681767} for further information.

\subsection{Graphs and Hilbert space}

Let $\Lambda = (\caV,\caE)$ be a graph, where $\caV$ is a finite set of vertices and $\caE$ is the set of ``edges", i.e.\ unordered pairs in $\caV \times \caV$. From a physical perspective, relevant graphs are regular graphs such as $\bbZ^d$ (or a finite box in $\bbZ^d$) with nearest-neighbor edges, but it is mathematically advantageous to allow for more general graphs. We restrict ourselves to spin-$\frac12$ systems, mainly because the stochastic representations only work in this case.

To each site $x \in \caV$ is associated a 2-dimensional Hilbert space $\caH_x = \bbC^2$. It is convenient to use Dirac's ``bra", $\langle\cdot|$, and ``ket", $|\cdot\rangle$, notation, in which we identify
\be
\label{bra ket}
|\tfrac12 \rangle = {1 \choose 0}, \qquad |-\tfrac12 \rangle = {0 \choose 1}.
\ee
The notation $\langle f | g \rangle$ means the inner product; we use the convention that it is linear in the second variable (and antilinear in the first). Occasionally, we also write $\langle f | A | g \rangle$ for $\langle f | Ag \rangle$.
The Hilbert space of a quantum spin system on $\Lambda$ is the tensor product
\be
\caH^{(\caV)} = \bigotimes_{x\in\Lambda} \caH_x,
\ee
which is the $2^{|\caV|}$ dimensional space spanned by elements of the form $\otimes_{x\in\caV} f_x$ with $f_x \in \caH_x$. 
The inner product between two such vectors is defined by
\be
\Bigl\langle \otimes_{x\in\caV} f_x \Big| \otimes_{x\in\caV} g_x \Bigr\rangle_{\caH^{(\caV)}} = \prod_{x\in\Lambda} \langle f_x | g_x \rangle_{\caH_x}.
\ee
The inner product above extends by (anti)linearity to the other vectors, 
which are all linear combinations of vectors of the form $\otimes_{x\in\caV} f_x$.

The basis \eqref{bra ket} of $\bbC^2$ has a natural extension in $\caH^{(\caV)}$; namely, given $s^{(\caV)} = (s_x)_{x\in\caV}$ with $s_x = \pm\frac12$, let
\be
|s^{(\caV)}\rangle = \bigotimes_{x\in\Lambda} |s_x\rangle.
\ee
These elements are orthonormal, i.e.
\be
\langle s^{(\caV)} | \tilde s^{(\caV)} \rangle = \prod_{x\in\caV} \langle s_x | \tilde s_x \rangle = \prod_{x\in\caV} \delta_{s_x, \tilde s_x},
\ee
where $\delta$ is Kronecker's symbol, $\delta_{ab} = 1$ if $a=b$, 0 otherwise.

\subsection{Spin operators}
\label{sec spins}

In the quantum world, physically relevant quantities are called \magicword{observables} and they 
are represented by self-adjoint operators.  The operators for the observable 
properties of our spin-$\frac{1}{2}$ particles are called the Pauli matrices, defined by
\be
\label{Pauli}
S^{(1)} = \tfrac12 \left( \begin{matrix} 0 & 1 \\ 1 & 0 \end{matrix} \right), \qquad S^{(2)} = \tfrac12 \left( \begin{matrix} 0 & -\ii \\ \ii & 0 \end{matrix} \right), \qquad S^{(3)} = \tfrac12 \left( \begin{matrix} 1 & 0 \\ 0 & -1 \end{matrix} \right).
\ee
We interpret $S^{(i)}$ as the spin component in the $i^{th}$ direction.  The matrices are clearly Hermitian and satisfy the relations
\be
\label{Pauli commutations}
[S^{(1)}, S^{(2)}] = \ii S^{(3)}, \quad [S^{(2)}, S^{(3)}] = \ii S^{(1)}, \quad [S^{(3)}, S^{(1)}] = \ii S^{(2)}.
\ee

These operators have natural extensions as spin operators in $\caH^{(\caV)}$.  Let $x \in \caV$, and write $\caH^{(\caV)} = \caH_x \otimes \caH^{(\caV \setminus \{x\})}$.  We define operators $S_x^{(i)}$ indexed by $x\in\caV$ by
\be
S_x^{(i)} = S^{(i)} \otimes \IdentityOperator_{\caV \setminus \{x\}}.
\ee
The commutation relations \eqref{Pauli commutations} extend to the operators $S_x^{(i)}$, namely
\be
\label{space Pauli commutations}
[S_x^{(1)}, S_y^{(2)}] = \ii \delta_{xy} S_x^{(3)},
\ee
and all other relations obtained by cyclic permutations of $(123)$. Indeed, it is not hard to check that the matrix elements $\langle s^{(\caV)} | \cdot | \tilde s^{(\caV)} \rangle$ of both sides are identical for all $s^{(\caV)} \in \{ -\frac12, \frac12 \}^\caV$. It is customary to introduce the notation $\vec S_x = (S_x^{(1)}, S_x^{(2)}, S_x^{(3)})$, and
\be
\vec S_x \cdot \vec S_y = S_x^{(1)} S_y^{(1)} + S_x^{(2)} S_y^{(2)} + S_x^{(3)} S_y^{(3)}.
\ee
Note that operators of the form $S_x^{(i)} S_y^{(j)}$, with $x\neq y$, act in $\caH^{(\caV)} = \caH_x \otimes \caH_y \otimes \caH^{(\caV \setminus \{x,y\})}$ as follows
\be
S_x^{(i)} S_y^{(j)} = S^{(i)} \otimes S^{(j)} \otimes \IdentityOperator_{\caV \setminus \{x,y\}}.
\ee
In the case $x=y$, and using $(S_x^{(i)})^2 = \frac14 \IdentityOperator_\caV$, we get
\be
\label{Pauli squares}
\vec S_x^2 = (S_x^{(1)})^2 + (S_x^{(2)})^2 + (S_x^{(3)})^2 = \tfrac34 \IdentityOperator_\caV.
\ee

\begin{lemma}
\label{lem coupling}
Consider $\vec S_x \cdot \vec S_y$ in $\caH_x \otimes \caH_y$. It is self-adjoint, and its eigenvalues and eigenvectors are as follows:
\begin{itemize}
\item $-\frac34$ is an eigenvalue with multiplicity 1; the eigenvector is $\frac1{\sqrt2} ( |\frac12, -\frac12 \rangle - |-\frac12, \frac12 \rangle)$.
\item $\frac14$ is an eigenvalue with multiplicity 3; three orthonormal eigenvectors are
\[
|\tfrac12, \tfrac12 \rangle, \qquad |-\tfrac12, -\tfrac12 \rangle, \qquad \tfrac1{\sqrt2} \bigl( |\tfrac12, -\tfrac12 \rangle + |-\tfrac12, \tfrac12 \rangle \bigr).
\]
\end{itemize}
\end{lemma}

The eigenvector corresponding to $-\frac34$ is called a ``singlet state" by physicists, while the eigenvectors for $\frac14$ are called ``triplet states".

\begin{proof}
We have for all $a,b = \pm \frac12$,
\be
\begin{split}
& S_x^{(1)} S_y^{(1)} |a,b\rangle = \tfrac14 |-a,-b\rangle, \\
& S_x^{(2)} S_y^{(2)} |a,b\rangle = -ab |-a,-b\rangle, \\
& S_x^{(3)} S_y^{(3)} |a,b\rangle = ab |a,b\rangle.
\end{split}
\ee
The lemma then follows from straightforward linear algebra.
\end{proof}

\subsection{Hamiltonians and magnetization}

We can now introduce the Heisenberg Hamiltonians, which are self-adjoint operators in $\caH^{(\caV)}$.
\be
\label{def Heisenberg}
\begin{split}
H_{\Lambda,h}^{\rm ferro} &= -\sum_{\{x,y\} \in \caE} \vec S_x \cdot \vec S_y - h \sum_{x\in\caV} S_{x}^{(3)}, \\
H_{\Lambda,h}^{\rm anti} &= +\sum_{\{x,y\} \in \caE} \vec S_x \cdot \vec S_y - h \sum_{x\in\caV} S_{x}^{(3)}.
\end{split}
\ee
Let us briefly discuss the physical motivation behind these operators. One is interested in describing a condensed matter system where atoms are arranged on a regular lattice. Each atom hosts exactly one relevant electron. Each electron stays on its atom and its spin is described by a vector in the Hilbert space $\bbC^2$. A system of two spins is described by a vector in $\bbC^{2} \otimes \bbC^{2}$. The singlet and triplet states of Lemma \ref{lem coupling} are invariant under rotation of the spins and they form a basis. In absence of external magnetic field, the energy operator should be diagonal with respect to these states, and there should be one eigenvalue for the singlet, and one other eigenvalue for the triplets. Up to constants, it should then be $\pm \vec S_{x} \cdot \vec S_{y}$. It is natural to define the total energy as the sum of nearest-neighbor interactions. Taking into account the contribution of the external magnetic field, which can be justified along similar lines, we get the Hamiltonians of \eqref{def Heisenberg}.

Next, let $M_\Lambda$ be the operator that represents the magnetization in the 3rd direction.
\be
M_\Lambda^{(3)} = \sum_{x\in\caV} S_x^{(3)}.
\ee

\begin{lemma}
The Hamiltonian and magnetization operators commute, i.e.,
\[
[H_{\Lambda,h}, M_\Lambda] = 0.
\]
\end{lemma}

\begin{proof}
This follows from the commutation relations \eqref{space Pauli commutations}.  Namely,
using the fact that $S_{x}^{(i)}$ and $S_{y}^{(3)}$ commute for $x \neq y$,
\be
\begin{split}
&[H_{\Lambda,h}, M_\Lambda] = \sum_{\{x,y\}\in\caE, z \in \caV} [\vec S_x \cdot \vec S_y, S_z^{(3)}] \\
&= \sum_{\{x,y\}\in\caE} \Bigl( [S_x^{(1)} S_y^{(1)}, S_x^{(3)}] + [S_x^{(1)} S_y^{(1)}, S_y^{(3)}] + [S_x^{(2)} S_y^{(2)}, S_x^{(3)}] + [S_x^{(2)} S_y^{(2)}, S_y^{(3)}] \Bigr).
\end{split}
\ee
The first commutator is
\be
[S_x^{(1)} S_y^{(1)}, S_x^{(3)}] = [S_x^{(1)}, S_x^{(3)}] S_y^{(1)} = -\ii S_x^{(2)} S_y^{(1)},
\ee
and the others are similar.  We get
\be
[H_{\Lambda,h}, M_\Lambda] = \ii \sum_{\{x,y\}\in\caE} \Bigl( -S_x^{(2)} S_y^{(1)} - S_x^{(1)} S_y^{(2)} + S_x^{(1)} S_y^{(2)} + S_x^{(2)} S_y^{(1)} \Bigr) = 0.
\ee
\end{proof}

\subsection{Gibbs states and free energy}
\label{sec Gibbs states}


The equilibrium states of quantum statistical mechanics are given by \emph{Gibbs states} $\langle \cdot \rangle_{\Lambda,\beta,h}$. These are nonnegative linear functionals on the space of operators in $\caH^{(\caV)}$ of the form
\be
\langle A \rangle_{\Lambda,\beta,h} = \frac1{Z_\Lambda(\beta,h)} \Tr A \e{-\beta H_{\Lambda,h}},
\ee
where the normalization 
\be
Z_\Lambda(\beta,h) = \Tr \e{-\beta H_{\Lambda,h}}
\ee
is called the \emph{partition function}.  Here, $\Tr$ denotes the usual matrix trace.


There are deep reasons why the Gibbs states describe equilibrium states but we will not dwell on them here. We now introduce the \emph{free energy} $F_\Lambda(\beta,h)$. Its physical motivation is that it provides a connection to thermodynamics. It is a kind of generating function and it is therefore mathematically useful. The definition of the free energy in our case is
\be
F_\Lambda(\beta,h) = -\frac1\beta \log Z_\Lambda(\beta,h).
\ee

\begin{lemma}
The function $\beta F_{\Lambda}(\beta,h)$ is concave in $(\beta,\beta h)$.
\end{lemma}

\begin{proof}
We will rather check that $-\beta F_\Lambda$ is convex, which is the case if the matrix
\[
\left( \begin{matrix} \frac{\partial^2 \beta F_\Lambda}{\partial \beta^2} & \frac{\partial^2 \beta F_\Lambda}{\partial \beta \partial (\beta h)} \\ \frac{\partial^2 \beta F_\Lambda}{\partial \beta \partial (\beta h)} & \frac{\partial^2 \beta F_\Lambda}{\partial (\beta h)^2} \end{matrix} \right)
\]
is positive definite. Let us write $\langle \cdot \rangle$ instead of $\langle \cdot \rangle_{\Lambda,\beta,h}$. We have
\be
\label{c'est concave}
\begin{split}
&\frac{\partial^2}{\partial \beta^2} \beta F_\Lambda(\beta,h) = - \bigl\langle (H_{\Lambda,0} - \langle H_{\Lambda,0} \rangle)^2 \bigr\rangle, \\
&\frac{\partial^2}{\partial (\beta h)^2} \beta F_\Lambda(\beta,h) = - \bigl\langle (M_\Lambda - \langle M_\Lambda \rangle)^2 \bigr\rangle, \\
&\frac{\partial^2}{\partial \beta \partial (\beta h)} \beta F_\Lambda(\beta,h) = \bigl\langle (H_{\Lambda,0} - \langle H_{\Lambda,0} \rangle) (M_\Lambda - \langle M_\Lambda \rangle) \bigr\rangle.
\end{split}
\ee
Then $F_\Lambda$ is convex if
\be
\label{CS ineq}
\bigl\langle (H_{\Lambda,0} - \langle H_{\Lambda,0} \rangle) (M_\Lambda - \langle M_\Lambda \rangle) \bigr\rangle^2 \leq \bigl\langle (H_{\Lambda,0} - \langle H_{\Lambda,0} \rangle)^2 \bigr\rangle \bigl\langle (M_\Lambda - \langle M_\Lambda \rangle)^2 \bigr\rangle.
\ee
It is not hard to check that the map $(A,B) \mapsto \langle A^* B \rangle$ is an inner product on the space of operators that commute with $H_{\Lambda,h}$. Then
\be
\label{gen CS ineq}
|\langle A^* B \rangle|^2 \leq \langle A^* A \rangle \langle B^* B \rangle
\ee
by the Cauchy-Schwarz inequality and, in particular, this implies \eqref{CS ineq}.
\end{proof}

Concave functions are necessarily continuous. But it is useful to establish that $F_\Lambda(\beta,h)$ is uniformly continuous on compact domains. This property will be used in Section \ref{sec thermolim} which discusses the existence of infinite volume limits.

\begin{lemma}
\label{lem F unif cont}
\[
\bigl| \beta F_\Lambda(\beta,h) - \beta' F_\Lambda(\beta',h') \bigr| \leq |\beta-\beta'| (\tfrac34 |\caE| + \tfrac{|h|}2 |\caV|) + \tfrac12 \beta |h-h'| |\caV|.
\]
\end{lemma}

\begin{proof}
We have
\be
\beta F_\Lambda(\beta,h) - \beta' F_\Lambda(\beta',h) = \int_{\beta'}^{\beta} \frac{\dd}{\dd s} s F_\Lambda(s,h) \dd s = \int_{\beta'}^{\beta} \langle H_{\Lambda,h} \rangle_{\Lambda, s, h} \dd s.
\ee
We can also check that $\beta F_\Lambda(\beta,h) - \beta F_\Lambda(\beta,h') = \int_{h'}^{h} \langle M_{\Lambda} \rangle_{\Lambda,\beta,s} \dd s$. 
The result follows from $|\langle A \rangle_{\Lambda,\beta,h}| \leq \|A\|$ for any operator $A$, and from $\| \vec S_x \cdot \vec S_y \| = \frac34$ (cf Lemma \ref{lem coupling}) and $\| S_x^{(3)} \| = \frac12$.
\end{proof}

\subsection{Symmetries}

In quantum statistical mechanics, a symmetry is represented by a unitary transformation which leaves the Hamiltonian invariant. It follows that (finite volume) Gibbs states also possess the symmetry.  However, infinite volume states may lose it. 
This is called \emph{symmetry breaking} and is a manifestation of a phase transition.  We only mention the ``spin flip" symmetry here, corresponding to the unitary operator 
\be
U |s^{(\caV)} \rangle = |-s^{(\caV)} \rangle.
\ee
One can check that $U^{-1}S^{(i)}_{x}S^{(i)}_{y}U = S^{(i)}_{x}S^{(i)}_{y}$ and $U^{-1}S_{x}^{(3)}U = -S_{x}^{(3)}$.
It follows that
\be
U^{-1} H_{\Lambda,h} U = H_{\Lambda,-h}.
\ee
This applies to both the ferromagnetic and antiferromagnetic Hamiltonians. It follows that $F_\Lambda(\beta,-h) = F_\Lambda(\beta,h)$, and so the free energy is symmetric as a function of $h$.


\section{Stochastic representations}
\label{sec stoch rep}

Stochastic representations of quantum lattice models go back to Ginibre, who used a Peierls contour argument to prove the occurrence of phase transitions in anisotropic models \cite{MR1552567}. Conlon and Solovej introduced a random walk representation for the ferromagnetic model and used it to get an upper bound on the free energy \cite{MR1148752}. A different representation was introduced by T\'oth, who improved the previous bound \cite{MR1224836}.
Further work on quantum models using similar representations include the quantum Pirogov-Sinai theory \cite{MR1414838,MR1400181} and Ising models in transverse magnetic field \cite{MR2581610,MR2608122,MR2477388}.

A major advantage of T\'oth's representation is that spin correlations have natural probabilistic expressions, being given by the probability that two sites belong to the same cycle (see below for details). A similar representation was introduced by Aizenman and Nachtergaele for the antiferromagnetic model, who used it to study properties of spin chains \cite{MR1288152}. The random objects are a bit different (loops instead of cycles), but this representation shares the advantage that spin correlations are given by the probability of belonging to the same loop.

The representations due to T\'oth and Aizenman-Nachtergaele both involve a Poisson process on the edges of the graph. The measure is reweighted by a function of suitable geometric objects (``cycles" or ``loops"). We first describe the two models in Section \ref{sec Poisson conf}; we will relate them to the Heisenberg models in Sections \ref{sec Toth} and \ref{sec AN}.

\subsection{Poisson edge process, cycles and loops}
\label{sec Poisson conf}

Recall that $\Lambda = (\caV,\caE)$ is a finite undirected graph.  
We attach to each edge a 
Poisson process on $[0,\beta]$ of unit intensity (see \S\ref{s:PPPmu} for the definition of a Poisson point process).  The Poisson processes
for different edges are independent.
A realization of this ``Poisson edge process'' is a finite sequence of pairs
\be
\omega = \bigl( (e_1,t_1), \ldots, (e_k,t_k) \bigr).
\ee
Each pair is called a bridge.  The number of bridges across each 
edge, thus, has a Poisson distribution with mean $\beta$, and the total number of bridges is Poisson with mean 
$\beta |\caE|$.  Conditional on there being $k$ bridges, their times of arrival are uniformly distributed in 
$\{0 <  t_{1} < t_{2} < \ldots < t_{k} < \beta \}$
and the edges are chosen uniformly from $\caE$.
The corresponding measure is denoted $\rho_{\Lambda,\beta}(\dd\omega)$.


To each realization $\omega$ there corresponds a configuration of cycles and configuration of loops. The mathematical definitions are a bit cumbersome but the geometric ideas are simpler and more elegant.  The reader is encouraged to look at Figure \ref{fig cycles} for an illustration.

\bfig
\centerline{\includegraphics[width=120mm]{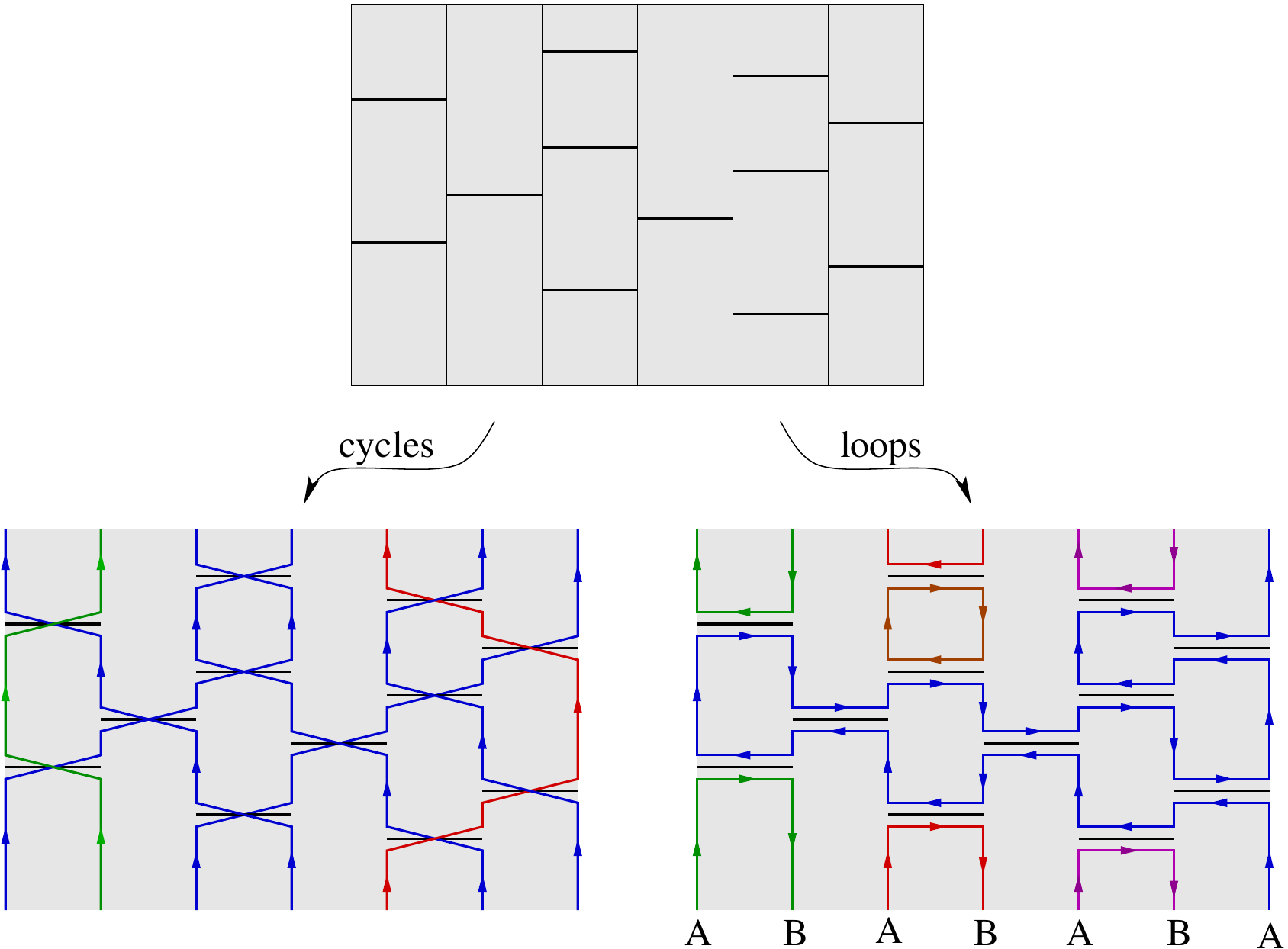}}
\caption{Top: an edge Poisson configuration $\omega$ on $\caV \times [0,\beta]_{\rm per}$. Bottom left: its associated cycle configuration. Bottom right: its associated loop configuration. We see that $|\caC(\omega)| = 3$ and $|\caL(\omega)| = 5$.}
\label{fig cycles}
\efig

We consider the cylinder $\caV \times [0,\beta]_{\rm per}$, where the subscript ``per'' indicates that we consider periodic boundary conditions. A \emph{cycle} is a closed trajectory on this space; that is, it is a function $\gamma : [0,L] \to \caV \times [0,\beta]_{\rm per}$ such that, if $\gamma(\tau) = (x(\tau), t(\tau))$, we have:
\begin{itemize}
\item $\gamma(\tau)$ is piecewise continuous; if it is continuous on the interval $I \subset [0,L]$, then $x(\tau)$ is constant and $\frac\dd{\dd \tau} t(\tau) = 1$ in $I$.
\item $\gamma(\tau)$ is discontinuous at $\tau$ iff the pair $(e,t)$ belongs to $\omega$, where $t = t(\tau)$ and $e$ is the edge $\{ x(\tau-), x(\tau+) \}$.
\end{itemize}
We choose $L$ to be the smallest positive number such that $\gamma(L) = \gamma(0)$. Then $L$ is the length of the cycle; it corresponds to the sum of the vertical legs in Figure\ \ref{fig cycles} and is necessarily a multiple of $\beta$. Let us make the cycles semi-continuous by assigning the value $\gamma(\tau) = \gamma(\tau-)$ at the points of discontinuity. We identify cycles whose support is identical. Then to each $\omega$ corresponds a configuration of cycles $\caC(\omega)$ whose supports form a partition of the cylinder $\caV \times [0,\beta]_{\rm per}$. The number of cycles is $|\caC(\omega)|$.

Loops are similar, but we now suppose that the graph is bipartite. The A sublattice possesses and orientation, which is reversed on the B sublattice. We still consider the cylinder $\caV \times [0,\beta]_{\rm per}$. A \emph{loop} is a closed trajectory on this space; that is, it is a function $\gamma : [0,L] \to \caV \times [0,\beta]_{\rm per}$ such that, with $\gamma(\tau) = (x(\tau), t(\tau))$:
\begin{itemize}
\item $\gamma(\tau)$ is piecewise continuous; if it is continuous in interval $I \subset [0,L]$, then $x(\tau)$ is constant and, in $I$,
\be
\frac\dd{\dd \tau} t(\tau) = \begin{cases} 1 & \text{ if $x(\tau)$ belongs to the A sublattice,} \\ -1 & \text{ if $x(\tau)$ belongs to the B sublattice.} \end{cases}
\ee
\item $\gamma(\tau)$ is discontinuous at $\tau$ iff the pair $(e,t)$ belongs to $\omega$, where $t = t(\tau)$ and $e$ is the edge $\{ x(\tau-), x(\tau+) \}$.
\end{itemize}
We choose $L$ to be the smallest positive number such that $\gamma(L) = \gamma(0)$. Then $L$ is the length of the loop; it corresponds to the sum of the vertical legs in Figure\ \ref{fig cycles} (as for cycles), but it is not a multiple of $\beta$ in general (contrary to cycles). We also make the loops semi-continuous by assigning the value $\gamma(\tau) = \gamma(\tau-)$ at the points of discontinuity. Identifying loops whose support is identical, to each $\omega$ corresponds a configuration of loops $\caL(\omega)$ whose supports form a partition of the cylinder $\caV \times [0,\beta]_{\rm per}$. The number of loops is $|\caL(\omega)|$.

As we shall see, the relevant probability measures for the Heisenberg models (with $h=0$) are 
proportional to $2^{|\caC(\omega)|} \rho_{\caE,\beta}(\dd\omega)$ and $2^{|\caL(\omega)|} \rho_{\caE,\beta}(\dd\omega)$.

\subsection{Duhamel expansion}
\label{sec Duhamel}

We first state and prove Duhamel's formula.  It is a variant of the Trotter product formula 
that is usually employed to derive stochastic representations.

\begin{proposition}
\label{prop Duhamel}
Let $A,B$ be $n \times n$ matrices. Then
\[
\begin{split}
\e{A+B} &= \e{A} + \int_0^1 \e{tA} B \e{(1-t)(A+B)} \dd t \\
&= \sum_{k\geq0} \int_{0<t_1<\dots<t_k<1} \dd t_1 \dots \dd t_k \e{t_1 A} B \e{(t_2-t_1) A} B \dots B \e{(1-t_k) A}.
\end{split}
\]
\end{proposition}

\begin{proof}
Let $F(s)$ be the matrix-valued function
\be
F(s) = \e{sA} + \int_0^s \e{tA} B \e{(s-t) (A+B)} \dd t.
\ee
We show that, for all $s$,
\be
\label{Duhamelissima}
\e{s(A+B)} = F(s).
\ee
The derivative of $F(s)$ is
\be
F'(s) = \e{sA} A + \e{sA} B + \int_0^s \e{tA} B \e{(s-t) (A+B)} (A+B) \dd t = F(s) (A+B).
\ee
On the other hand, the derivative of $\e{s(A+B)}$ is $\e{s(A+B)} (A+B)$. The identity \eqref{Duhamelissima} clearly holds for $s=0$ and, since both sides satisfy the same differential equation, they must be equal for all $s$.

We can iterate Duhamel's formula $N$ times so as to get
\be
\begin{split}
&\e{A+B} = \sum_{k=0}^N \int_{0<t_1<\dots<t_k<1} \dd t_1 \dots \dd t_k \e{t_1 A} B \e{(t_2-t_1) A} B \dots B \e{(1-t_k) A} \\
&+ \int_{0<t_1<\dots<t_N<1} \dd t_1 \dots \dd t_k \e{t_1 A} B \e{(t_2-t_1) A} B \dots B \Bigl[ \e{(1-t_N) (A+B)} - \e{(1-t_N) A} \Bigr].
\end{split}
\ee
Using $\| \e{tA} \| \leq \e{t\|A\|}$, the last term is less than $2 \e{\|A\|+\|B\|} \frac{\|B\|^N}{N!}$ and so 
it vanishes in the limit $N\to\infty$. The summand is less than $\e{\|A\|} \frac{\|B\|^k}{k!}$, so that the sum is absolutely convergent.
\end{proof}

Our goal is to perform Duhamel's expansion on the Gibbs operator $\e{-\beta H_{\Lambda,h}}$, where the Hamiltonian is given by a sum of terms indexed by the edges and by vertices. The following corollary applies in this case.

\begin{corollary}
\label{cor Duhamel}
Let $A$ and $(h_e)$, $e \in \caE$, be matrices in $\caH^{(\caV)}$. Then
\[
\e{\beta (A + \sum_{e\in\caE} h_e)} = \int \dd\rho_{\caE,\beta}(\omega) \e{t_1 A} h_{e_1} \e{(t_2-t_1) A} h_{e_2} \dots h_{e_k} \e{(\beta-t_k) A},
\]
where $(t_1,e_1), \dots, (t_k,e_k)$ are the bridges in $\omega$.
\end{corollary}

\begin{proof}
We can expand the right side by summing over the number $k$ of events, then integrating over $0 < t_1 < \dots < t_k < \beta$ for the times of occurrence, and then summing over edges $e_1, \dots, e_k \in \caE$. After the change of variables $t_i' = t_i/\beta$, we recognize the second line of Proposition \ref{prop Duhamel}. 
\end{proof}

\subsection{T\'oth's representation of the ferromagnet}
\label{sec Toth}

It is convenient to introduce the operator $T_{x,y}$ which transposes the spins at $x$ and $y$. In $\caH_x \otimes \caH_y$, the operator acts as follows:
\be
T_{x,y} |a,b\rangle = |b,a\rangle, \qquad a,b = \pm \tfrac12.
\ee
This rule extends to general vectors by linearity, and it extends to $\caH^{(\caV)}$ by tensoring it with $\IdentityOperator_{\caV \setminus \{x,y\}}$. Using Lemma \ref{lem coupling}, it is not hard to check that
\be
\label{coupling vs transposition}
\vec S_x \cdot \vec S_y = \tfrac12 T_{x,y} - \tfrac14 \IdentityOperator_{\{x,y\}}.
\ee

Recall that $\caC(\omega)$ is the set of cycles of $\omega$, and let $\gamma_x \in \Cycles(\omega)$ denote the cycle that intersects $(x,0) \in \caV \times [0,\beta]_{\rm per}$.  Let $L(\gamma)$ denote the (vertical) length of the cycle $\gamma$; it is always a multiple of $\beta$ in the theorem below.

\begin{theorem}[T\'oth's representation of the ferromagnet]
\label{thm Toth}
The partition function, the average magnetization, and the two-point correlation function have the following expressions.
\[
\begin{split}
&Z_\Lambda^{\rm ferro}(2\beta,h) = \e{-\frac\beta2 |\caE|} \int\dd\rho_{\caE,\beta}(\omega) \prod_{\gamma \in \caC(\omega)} \bigl( 2 \cosh(hL(\gamma)) \bigr), \\
&\Tr S_x^{(3)} \e{-2\beta H_{\Lambda,h}^{\rm ferro}} = \tfrac12 \e{-\frac\beta2 |\caE|} \int\dd\rho_{\caE,\beta}(\omega) \, \tanh(hL(\gamma_x)) \prod_{\gamma \in \caC(\omega)} \bigl( 2 \cosh(hL(\gamma)) \bigr), \\
&\Tr S_x^{(3)} S_y^{(3)} \e{-2\beta H_{\Lambda,h}^{\rm ferro}} = \tfrac14 \e{-\frac\beta2 |\caE|} \int\dd\rho_{\caE,\beta}(\omega) \prod_{\gamma \in \caC(\omega)} \bigl( 2 \cosh(hL(\gamma)) \bigr) \\
&\hspace{4cm} \times \begin{cases} 1 & \text{if } \gamma_x = \gamma_y, \\ \tanh(hL(\gamma_x)) \tanh(hL(\gamma_y)) & \text{if } \gamma_x \neq \gamma_y. \end{cases}
\end{split}
\]
\end{theorem}

\bfig
\centerline{\includegraphics[width=60mm]{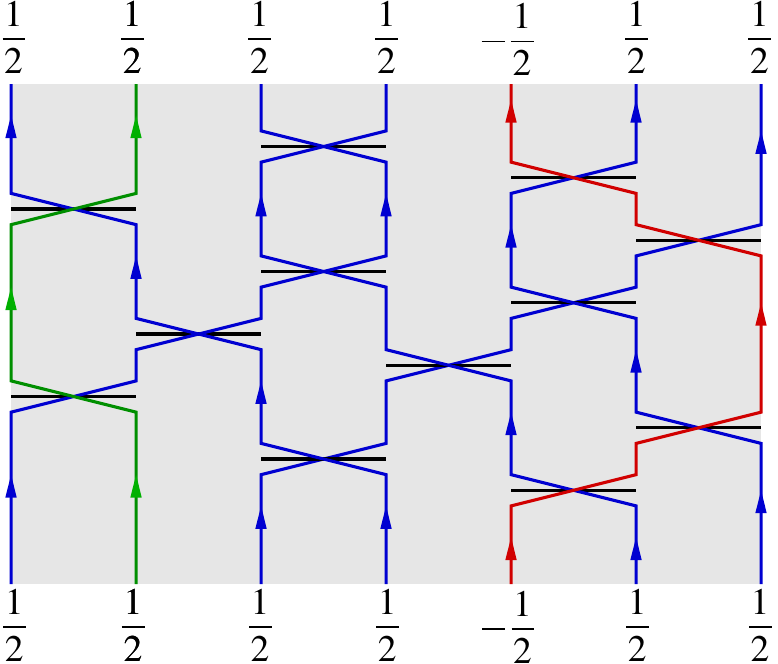}}
\caption{Each cycle is characterized by a given spin.}
\label{fig ferro}
\efig

\begin{proof}
The partition function can be expanded using Corollary \ref{cor Duhamel} so as to get
\be
\begin{split}
Z_\Lambda^{\rm ferro}(2\beta,h) &= \e{-\frac\beta2 |\caE|} \Tr \e{\beta (2h M_\Lambda + \sum_e T_e)} \\
&= \e{-\frac\beta2 |\caE|} \int\dd\rho_{\caE,\beta}(\omega) \sum_{s^{(\caV)}} \langle s^{(\caV)} | \e{2t_1 h M_\Lambda} T_{e_1} \dots T_{e_k} \e{2(\beta - t_k) h M_\Lambda} |s^{(\caV)} \rangle,
\end{split}
\ee
where $(e_1,t_1), \dots, (e_k,t_k)$ are the times and edges of $\omega$. Observe that the vectors $|s^{(\caV)} \rangle$ are eigenvectors of $\e{tM_\Lambda}$. It is not hard to see that the matrix element above is zero unless each cycle is characterized by a single spin value (see illustration in Figure\ \ref{fig ferro}). If the matrix element is not zero, then it is equal to
\be
\langle s^{(\caV)} | \e{2t_1 h M_\Lambda} T_{e_1} \dots T_{e_k} \e{2(\beta - t_k) h M_\Lambda} |s^{(\caV)} \rangle = \prod_{\gamma\in\Cycles(\omega)} \e{2h L(\gamma) s(\gamma)}
\ee
with $s(\gamma)$ the spin of the cycle $\gamma$. After summing over $s(\gamma) = \pm\frac12$, each cycle contributes $\e{hL(\gamma)} + \e{-hL(\gamma)} = 2 \cosh(h L(\gamma))$, and we obtain the expression for the partition function.

The expression which involves $S_x^{(3)}$ is similar, except that the cycle $\gamma_x$ that contains $x \times \{0\}$ contributes $\frac12 \e{hL(\gamma_x)} - \tfrac12 \e{-hL(\gamma_x)} = \sinh(h L(\gamma_x))$. Since the factor $2 \cosh(h L(\gamma_x))$ appears in the expression, it must be corrected by the hyperbolic tangent.

Finally, the expression that involves $S_x^{(3)} S_y^{(3)}$ has two terms, corresponding to whether $(x ,0)$ and $(y,0)$ find themselves in the same cycle or not. In the first case, we get $\frac12 \cosh(hL(\gamma_{xy}))$, but in the second case we get $\sinh(h L(\gamma_x)) \sinh(h L(\gamma_y))$, which eventually gives the hyperbolic tangents.
\end{proof}

It is convenient to rewrite the cycle weights somewhat. Using $2 \cosh(hL(\gamma)) = \e{hL(\gamma)} (1 + \e{-2hL(\gamma)})$ and $\sum_{\gamma\in\Cycles(\omega)} L(\gamma) = \beta |\caV|$, the relevant probability measure 
for the cycle representation can be written
\be
\label{mu ferro}
\Prob^{\rm cycles}_{\Lambda,\beta,h}(d\omega) = Z^{\rm ferro}_\Lambda(2\beta,h)^{-1} \e{-\frac\beta2 |\caE| + \beta h |\caV|} \dd\rho_{\caE,\beta}(d\omega) \prod_{\gamma\in\Cycles(\omega)} \bigl( 1 + \e{-2hL(\gamma)} \bigr). 
\ee
This form makes it easier to see the effect of the external field $h \geq 0$. Notice that
the product over cycles simplifies to $2^{|\caC(\omega)|}$ when the external field strength vanishes (i.e. $h=0$).  
Then, in terms of the cycle model, the expectation of the spin operators and correlations are given by 
\be
\langle S_x^{(3)} \rangle_{\Lambda,2\beta,h} = \tfrac12 \bbE^{\rm cycles}_{\Lambda,\beta,h} \bigl( \tanh(hL(\gamma_x)) \bigr)
\ee
and
\be
\begin{split}
\langle S_x^{(3)} S_y^{(3)} \rangle_{\Lambda,2\beta,h} = &\tfrac14 \bbP^{\rm cycles}_{\Lambda,\beta,h}(\gamma_x = \gamma_y) \\ 
&+ \tfrac14 \bbE^{\rm cycles}_{\Lambda,\beta,h} \bigl[ 1_{\gamma_x \neq \gamma_y} \tanh(hL(\gamma_x)) \tanh(hL(\gamma_y)) \bigr].
\end{split}
\ee
In the case $h=0$, we see that $\langle S_x^{(3)} \rangle_{\Lambda,2\beta,0} = 0$, as already noted from the spin flip symmetry, and
\be
\langle S_x^{(3)} S_y^{(3)} \rangle_{\Lambda,2\beta,0} = \tfrac14 \bbP^{\rm cycles}_{\Lambda,\beta,h}(\gamma_x = \gamma_y).
\ee
That is, the spin-spin correlation of two sites $x$ and $y$ is proportional to the probability that
the sites lie in the same cycle.

\subsection{Aizenman-Nachtergaele's representation of the antiferromagnet}
\label{sec AN}

The antiferromagnetic model only differs from the ferromagnetic model by a sign, but this leads to deep changes. As the transposition operator now carries a negative sign in the Hamiltonian, one possibility would be to turn the measure
corresponding to \eqref{mu ferro} into a signed measure, with an extra factor $(-1)^k$ where $k=k(\omega)$ is the number of transpositions. That would mean descending from the heights of probability theory to... well, to measure theory. This fate can fortunately be avoided thanks to the observations of Aizenman and Nachtergaele \cite{MR1288152}.

Their representation is restricted to bipartite graphs. A graph is \emph{bipartite} if the set of vertices $\caV$ can be partitioned into two sets $\caV_{\rm A}$ and $\caV_{\rm B}$ such that edges only connect the A set to the B set:
\be
\{x,y\} \in \caE \Longrightarrow (x,y) \in \caV_{\rm A} \times \caV_{\rm B} \text{ or } (x,y) \in \caV_{\rm B} \times \caV_{\rm A}.
\ee
This class contains many relevant cases, such as finite boxes in $\bbZ^d$; periodic boundary conditions are allowed provided the side lengths are even. In the following, we use the notation
\be
(-1)^{x} = \begin{cases} 1 & \text{if } x \in \caV_{\rm A}, \\ -1 & \text{if } x \in \caV_{\rm B}. \end{cases}
\ee

Instead of the transposition operator, we consider the projection operator $P_{xy}^{(0)}$ onto the singlet state 
described in Lemma \ref{lem coupling}. Its action on the basis is
\be
\label{proj singlet}
P_{xy}^{(0)} |a,a\rangle = 0, \qquad P_{xy}^{(0)} |a,-a\rangle = \tfrac12 |a,-a \rangle - \tfrac12 |-a,a \rangle,
\ee
for all $a = \pm\tfrac12$. Further, it follows from Lemma \ref{lem coupling} that
\be
\vec S_x \cdot \vec S_y = \tfrac14 \IdentityOperator_{\{x,y\}} - P_{xy}^{(0)}.
\ee
Recall that $\caL(\omega)$ is the set of loops of $\omega$. Let $\gamma_{x}$ denote the loop that contains $(x,0)$. We do not need notation for the loops that do not intersect the $t=0$ plane.  Also, it is not the \magicword{lengths} of the loops 
which are important but their winding number $w(\gamma)$.

\begin{theorem}[Aizenman-Nachtergaele's representation of the antiferromagnet]
\label{thm AN}
Assume that $\Lambda$ is a bipartite graph. The partition function, the average magnetization and the two-point correlation function have the following expressions.
\[
\begin{split}
&Z_\Lambda^{\rm anti}(2\beta,h) = \e{-\frac\beta2 |\caE|} \int\dd\rho_{\caE,\beta}(\omega) \prod_{\gamma \in \caL(\omega)} \bigl( 2 \cosh(\beta h w(\gamma)) \bigr), \\
&\Tr S_x^{(3)} \e{-2\beta H_{\Lambda,h}^{\rm anti}} = \tfrac12 (-1)^{x} \e{-\frac\beta2 |\caE|} \int\dd\rho_{\caE,\beta}(\omega) \, \tanh(\beta h w(\gamma_x)) \\
&\hspace{7cm} \times \prod_{\gamma \in \caL(\omega)} \bigl( 2 \cosh(\beta h w(\gamma)) \bigr), \\
&\Tr S_x^{(3)} S_y^{(3)} \e{-2\beta H_{\Lambda,h}^{\rm anti}} = \tfrac14 (-1)^{x} (-1)^{y} \e{-\frac\beta2 |\caE|} \int\dd\rho_{\caE,\beta}(\omega) \\
& \times \prod_{\gamma \in \caL(\omega)} \bigl( 2 \cosh(\beta h w(\gamma)) \bigr) \times \begin{cases} 1 & \text{if } \gamma_x = \gamma_y, \\ \tanh(\beta h w(\gamma_x)) \tanh(\beta h w(\gamma_y)) & \text{if } \gamma_x \neq \gamma_y. \end{cases}
\end{split}
\]
\end{theorem}

When $h=0$, we get the simpler factor $2^{|\caL(\omega)|}$.

\bfig
\centerline{\includegraphics[width=60mm]{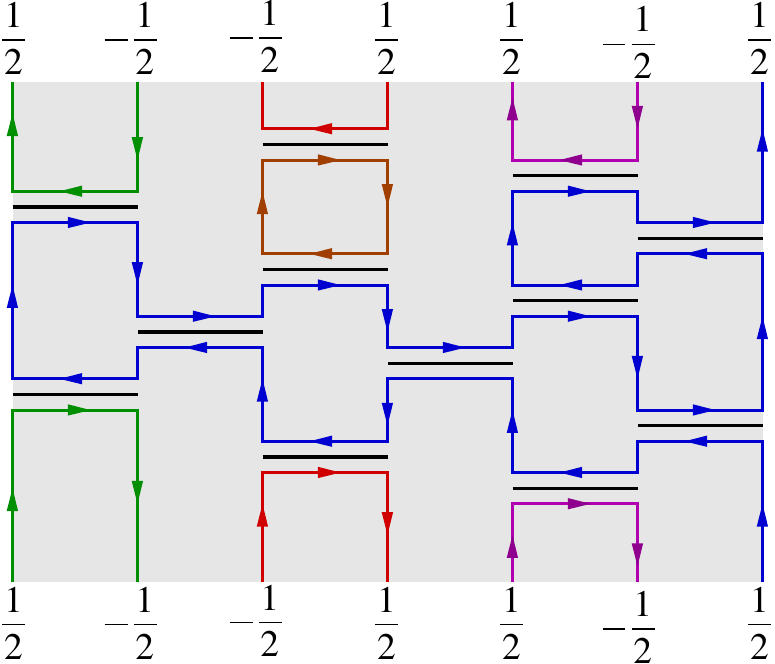}}
\caption{Each loop is characterized by a given spin, but the values alternate according to whether the site belongs to the A or B sublattice.}
\label{fig anti}
\efig

\begin{proof}
As before, we expand the partition function using Corollary \ref{cor Duhamel} and we get
\be
\begin{split}
Z_\Lambda^{\rm anti}&(2\beta,h) = \e{-\frac\beta2 |\caE|} \Tr \e{\beta (2h M_\Lambda + \sum_e 2P_e^{(0)})} \\
&= \e{-\frac\beta2 |\caE|} \int\dd\rho_{\caE,\beta}(\omega) \sum_{s^{(\caV)}} \langle s^{(\caV)} | \e{2t_1 h M_\Lambda} 2P_{e_1}^{(0)} \dots 2P_{e_k}^{(0)} \e{2(\beta - t_k) h M_\Lambda} |s^{(\caV)} \rangle,
\end{split}
\ee
where $(e_1,t_1), \dots, (e_k,t_k)$ are the times and the edges of $\omega$. Notice that
\be
\e{tM_\Lambda} |s^{(\caV)} \rangle = \e{t \langle s^{(\caV)} | M_\Lambda |s^{(\caV)} \rangle}  |s^{(\caV)} \rangle.
\ee
In Dirac's notation, the resolution of the identity is
\be
\IdentityOperator_{\caV} = \sum_{s^{(\caV)} \in \{-\frac12,\frac12\}^{\caV}} |s^{(\caV)} \rangle \langle s^{(\caV)}|.
\ee
We insert it on the right of each operator $P_{e}^{(0)}$ and we obtain
\be
\begin{split}
&Z_\Lambda^{\rm anti}(2\beta,h) = \e{-\frac\beta2 |\caE|} \int\dd\rho_{\caE,\beta}(\omega) \sum_{s_{1}^{(\caV)}, \dots, s_{k}^{(\caV)}} \e{2t_{1} h \langle s_{1}^{(\caV)} | M_\Lambda |s_{1}^{(\caV)} \rangle} \langle s_{1}^{(\caV)} | 2P_{e_1}^{(0)} |s_{2}^{(\caV)} \rangle \\
& \times \!\! \e{2(t_{2}-t_{1}) h \langle s_{2}^{(\caV)} | M_\Lambda |s_{2}^{(\caV)} \rangle} \langle s_{2}^{(\caV)} | 2P_{e_2}^{(0)} |s_{3}^{(\caV)} \rangle \dots \langle s_{k}^{(\caV)} | 2P_{e_k}^{(0)} |s_{1}^{(\caV)} \rangle \e{2(\beta-t_{k}) h \langle s_{1}^{(\caV)} | M_\Lambda |s_{1}^{(\caV)} \rangle}.
\end{split}
\ee
Let us now observe that this long expression can be conveniently expressed in the language of loops. We can interpret $\omega$ and $s_{1}^{(\caV)}, \dots, s_{k}^{(\caV)}$ as a spin configuration $\bss$ in $\caV \times [0,\beta]_{\rm per}$. It is constant in time except possibly at $(e_{i},t_{i})$. By \eqref{proj singlet}, the product
\[
\langle s_{1}^{(\caV)} | 2P_{e_1}^{(0)} |s_{2}^{(\caV)} \rangle \dots \langle s_{k}^{(\caV)} | 2P_{e_k}^{(0)} |s_{1}^{(\caV)} \rangle
\]
differs from 0 iff the value of $(-1)^{x} \bss(x,t)$ is constant on each loop (see illustration in Figure\ \ref{fig anti}). In this case, its value is $\pm1$, as each bridge contributes $+1$ if the spins are constant, and $-1$ if they flip. Let us check that, in fact, it is always $+1$. If the bridge separates two loops with spins $a$ and $b$, the factor is
\be
(-1)^{a-b} = \e{\ii \pi a} \e{-\ii \pi b}.
\ee
Looking at the loop $\gamma$ with spin $a$, there is a factor $\e{\ii \pi a}$ for each jump A$\to$B (of the form $\ulcorner\!\urcorner$) and a factor $\e{-\ii \pi a}$ for each jump B$\to$A (of the form $\llcorner\!\lrcorner$). Since there is an identical number of both types of jumps, these factors precisely cancel.

The product
\[
\e{2t_{1} h \langle s_{1}^{(\caV)} | M_\Lambda |s_{1}^{(\caV)} \rangle} \e{2(t_{2}-t_{1}) h \langle s_{2}^{(\caV)} | M_\Lambda |s_{2}^{(\caV)} \rangle} \dots \e{2(\beta-t_{k}) h \langle s_{1}^{(\caV)} | M_\Lambda |s_{1}^{(\caV)} \rangle}
\]
also factorizes according to loops. The contribution of a loop $\gamma$ with spin $a$ is $\e{2h L_{\rm A}(\gamma) a - 2h L_{\rm B}(\gamma) a}$, where $L_{\rm A}, L_{\rm B}$ are the vertical lengths of $\gamma$ on the A and B sublattices. We have
\be
L_{\rm A}(\gamma) - L_{\rm B}(\gamma) = \beta w(\gamma).
\ee
The contribution is therefore $\e{2\beta h w(\gamma) a}$. Summing over $a = \pm \frac12$, we get the hyperbolic cosine of the expression for the partition function of Theorem \ref{thm AN}.

The expression that involves $S_{x}^{(3)}$ is similar; the only difference is that the loop that contains $(x,0)$ contributes $(-1)^{x} \sinh (\beta h w(\gamma))$ instead of $2 \cosh (\beta h w(\gamma))$, hence the hyperbolic tangent. Finally, the expression that involves $S_{x}^{(3)} S_{y}^{(3)}$ is similar but we need to treat separately the cases where $(x,0)$ and $(y,0)$ belong or do not belong to the same loop.
\end{proof}

\section{Thermodynamic limit and phase transitions}
\label{sec phase transitions}

Phase transitions are cooperative phenomena where a small change of the external parameters results in drastic alterations in the properties of the system. There was some confusion in the early days of statistical mechanics as to whether the formalism contained the possibility of describing phase transitions, as all finite volume quantities are smooth. It was eventually realized that the proper formalism involves a \emph{thermodynamic limit} where the system size tends to infinity, in such a way that the local behavior remains largely unaffected. The proofs of the existence of thermodynamic limits were fundamental contributions to the mathematical theory of phase transitions, and they were pioneered by Fisher and Ruelle in the 1960's; see \cite{MR0289084} for more references.

We show that the free energy converges in the thermodynamic limit along a sequence of boxes in $\bbZ^{d}$ of increasing size (Section \ref{sec thermolim}). We discuss various characterizations of ferromagnetic phase transitions in Section \ref{sec thermo limit ferro}, and magnetic long-range order in Section \ref{sec thermo limit anti}. In Section \ref{sec prob phase transitions} we consider the relations between the magnetization in the quantum models and the lengths of the cycles and loops.

\subsection{Thermodynamic limit}
\label{sec thermolim}

Despite our professed intention to treat arbitrary graphs, we now restrict ourselves to a very specific case, namely that of a sequence of cubes in $\Z^d$ whose side lengths tends to infinity.
Since $F_\Lambda(\beta,h)$ scales like the volume of the system, we define the \emph{mean free energy} $f_\Lambda$ to be
\be
f_\Lambda(\beta,h) = \frac1{|\caV|} F_\Lambda(\beta,h).
\ee
We consider the sequence of graphs $\Lambda_n = (\caV_n,\caE_n)$ where $\caV_n = \{1,\dots,n\}^d$ and $\caE_n$ is the set of nearest-neighbors, i.e., $\{x,y\} \in \caE_n$ iff $\|x-y\| = 1$.

\begin{theorem}[Thermodynamic limit of the free energy]
\label{thm thermo limit}
The sequence of functions $(f_{\Lambda_n}(\beta,h))_{n\geq1}$ converges pointwise to a function $f(\beta,h)$, uniformly on compact sets.
\end{theorem}

\bfig
\centerline{\includegraphics[width=60mm]{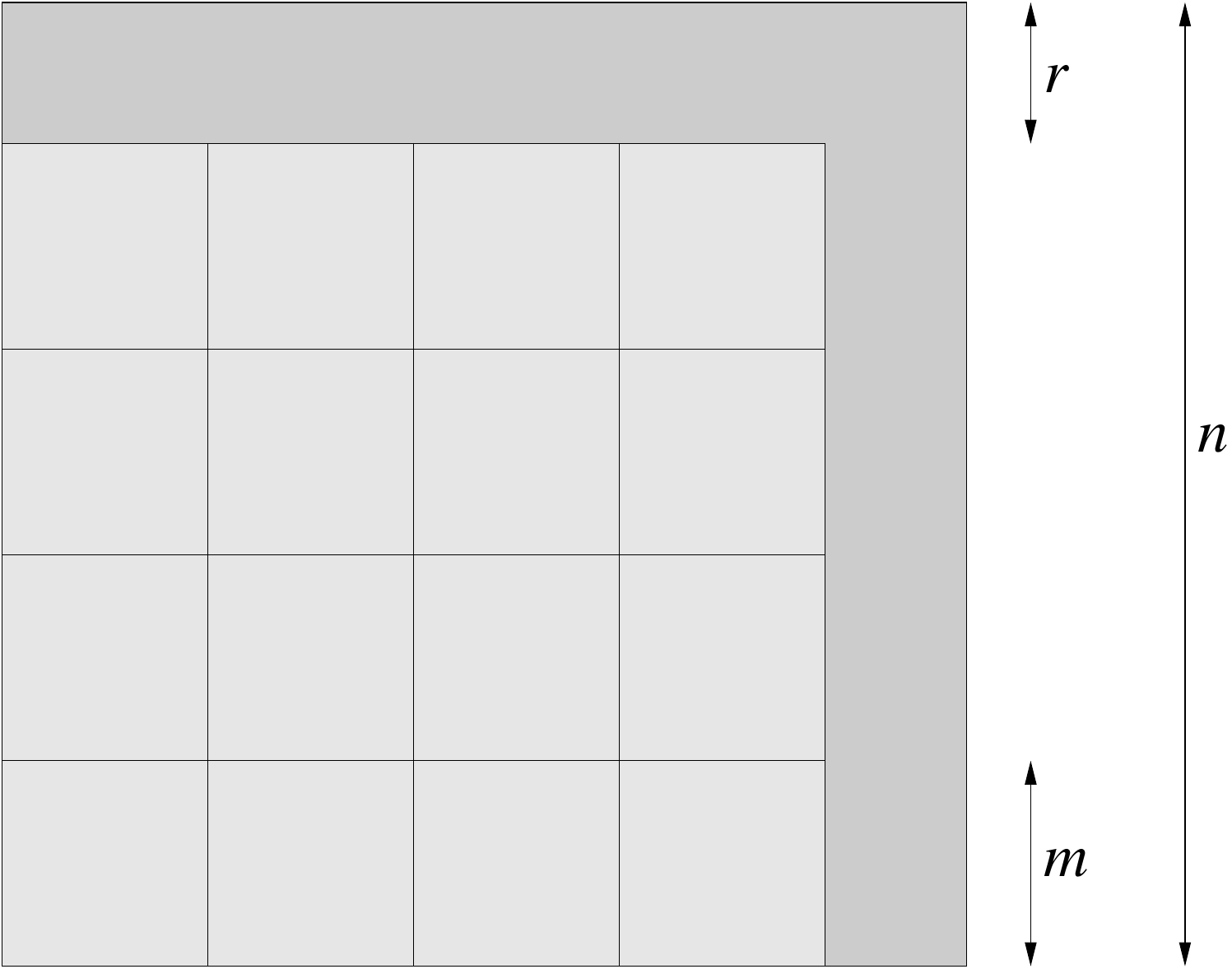}}
\caption{The large box of size $n$ is decomposed in $k^{d}$ boxes of size $m$; there are no more than $drn^{d-1}$ remaining sites in the darker area.}
\label{fig boxes}
\efig

\begin{proof}
We consider the ferromagnetic model, but the modifications for the antiferromagnetic model are straightforward. We use a subadditive argument. Notice that the inequality $\Tr \e{A+B} \geq \Tr \e{A}$ holds for all self-adjoint operators $A,B$ with $B \geq 0$. (This follows e.g.\ from the minimax principle, or from Klein's inequality.) We rewrite the Hamiltonian so as to have only positive definite terms. Namely, let
\be
h_{x,y} = \vec S_x \cdot \vec S_y + \tfrac34 \IdentityOperator.
\ee
Then
\be
Z_\Lambda(\beta,h) = \e{-\frac34 \beta |\caE|} \Tr \exp\Bigl( \beta \sum_{\{x,y\} \in \caE} h_{x,y} + \beta h \sum_{x \in \caV} S_x^{(3)} \Bigr).
\ee
Let $m,n,k,r$ be integers such that $n = km+r$ and $0 \leq r < m$. The box $\caV_n$ is the disjoint union of $k^d$ boxes of size $m$, and of some remaining sites (fewer than $dr n^{d-1}$); see Figure\ \ref{fig boxes} for an illustration. We get an inequality for the partition function in $\Lambda_n$ by dismissing all $h_{x,y}$ where $\{x,y\}$ are not inside a single box of size $m$. The boxes $\caV_{m}$ become independent, and
\be
\begin{split}
Z_{\Lambda_n}(\beta,h) &\geq \e{-\frac34 \beta |\caE_{n}|} \Bigl[ \Tr_{\caH^{(\caV_m)}} \exp \Bigl( \beta \sum_{\{x,y\} \in \caE_m} h_{x,y} + \beta h \sum_{x \in \caV_m} S_x^{(3)} \Bigr) \Bigr]^{k^{d}} \\
&= [Z_{\Lambda_m}(\beta,h)]^{k^d} \e{-\frac34 \beta |\caE_{n}|} \e{k^d \frac34 \beta |\caE_m|}.
\end{split}
\ee
We have neglected the contribution of $\e{\beta h S_{x}^{(3)}}$ for $x$ outside the boxes $\caV_{m}$, which is possible because their traces are greater than 1. It is not hard to check that
\be
|\caE_n| \leq k^d |\caE_m| + k^d d m^{d-1} + d^{2} r n^{d-1}.
\ee
We then obtain a subbaditive relation for the free energy, up to error terms that will soon disappear:
\be
f_{\Lambda_n}(\beta,h) \leq \frac{(km)^d}{n^d} f_{\Lambda_m}(\beta,h) + \frac{3k^d d m^{d-1}}{4n^d} + \frac{3d^{2}r}{4n}.
\ee
Then, since $\frac{km}n \to 1$ as $n\to\infty$,
\be
\limsup_{n\to\infty} f_{\Lambda_n}(\beta,h) \leq f_{\Lambda_m}(\beta,h) + \frac{3d}{4m}.
\ee
Taking the $\liminf$ over $m$ in the right side, we see that it is larger or equal to the $\limsup$, and so the limit necessarily exists.

Uniform convergence on compact intervals follows from Lemma \ref{lem F unif cont} (which implies that $(f_{\Lambda_n})$ is equicontinuous) and the Arzel\`a-Ascoli theorem (see e.g.\ Theorem 4.4 in Folland \cite{MR1681462}).
\end{proof}

\begin{corollary}[Thermodynamic limit with periodic boundary conditions]
\label{cor thermo limit}
Let $(\Lambda_n^{\rm per})$ be the sequence of cubes in $\bbZ^{d}$ of size $n$ with periodic boundary conditions and nearest-neighbor edges. Then $(f_{\Lambda_n^{\rm per}}(\beta,h))_{n\geq1}$ converges pointwise to the same function $f(\beta,h)$ as in Theorem \ref{thm thermo limit}, uniformly on compact sets.
\end{corollary}

This follows from $|f_{\Lambda_n^{\rm per}}(\beta,h) - f_{\Lambda_n}(\beta,h)| \leq \frac{3d}{4n}$, which is not too hard to prove, and Theorem \ref{thm thermo limit}.

\subsection{Ferromagnetic phase transition}
\label{sec thermo limit ferro}

In statistical physics, an \magicword{order parameter} is a quantity which allows one to identify a phase, typically because it vanishes in all phases except one. The relevant order parameter here is the magnetization, which is expected to be zero at high temperatures and positive at low temperatures. There are actually three natural definitions for the magnetization; we show below that the first two are equivalent, and that the last one is smaller than the first two.

\begin{itemize}
\item {\bf Thermodynamic magnetization}. This is equal to (the negative of) the right-derivative of $f(\beta,h)$ with respect to $h$.  We are looking for a jump in the derivative, which is referred to as a \magicword{first-order} phase transition.
\be
m^*_{\rm th}(\beta) = -\lim_{h\to0+} \frac{f(\beta,h) - f(\beta,0)}h.
\ee
(The limit exists because $f$ is concave.)
\item {\bf Residual magnetization}. Imagine placing the ferromagnet in an external magnetic field, so that it becomes
magnetized.  Now gradually turn off the external field.  Does the system still display global magnetization?  
Mathematically, the relevant order parameter is
\be
m^*_{\rm res}(\beta) = \lim_{h\to0+} \liminf_{n\to\infty} \frac1{n^d} \langle M_{\Lambda_n} \rangle_{\Lambda_n,\beta,h}.
\ee
(We see below that the $\liminf$ can be replaced by the $\limsup$ without affecting $m^*_{\rm res}$. The limit over $h$ exists because $\langle M_{\Lambda_n} \rangle$ is the derivative of a concave function, and it is therefore monotone.)
\item {\bf Spontaneous magnetization} at $h=0$. Since $\langle M_\Lambda \rangle = 0$ (because of the spin flip symmetry), we rather consider
\be
m^*_{\rm sp}(\beta) = \liminf_{n\to\infty} \frac1{n^d} \langle |M_{\Lambda_n}| \rangle_{\Lambda_n,\beta,0}.
\ee
Here, $|M_{\Lambda}|$ denotes the absolute value of the matrix $M_{\Lambda}$.
\end{itemize}
A handier quantity, however, is the expectation of $M_{\Lambda}^{2}$, which can be expressed in terms of the two-point correlation function, see below. It is equivalent to $m^*_{\rm sp}$ in the sense that both are zero or both differ from zero:

\begin{lemma}
\label{lem magn vs corr}
\[
\bigl\langle \tfrac{|M_{\Lambda}|}{|\caV|} \bigr\rangle^{2}_{\Lambda, \beta, 0} \leq \bigl\langle \bigl( \tfrac{M_{\Lambda}}{|\caV|} \bigr)^{2} \bigr\rangle_{\Lambda,\beta,0} \leq \tfrac12 \bigl\langle \tfrac{|M_{\Lambda}|}{|\caV|} \bigr\rangle_{\Lambda,\beta,0}.
\]
\end{lemma}

\begin{proof}
For the first inequality, use $|M_{\Lambda}| = |M_{\Lambda}| \IdentityOperator$ and then the Cauchy-Schwarz inequality \eqref{gen CS ineq}. For the second inequality, observe that $|M_{\Lambda}| \leq \frac12 |\caV| \IdentityOperator$ implies that $M^{2}_{\Lambda} \leq \frac12 |\caV| |M_{\Lambda}|$, and use the fact that the Gibbs state is a positive linear functional.
\end{proof}

\bfig
\centerline{\includegraphics[width=50mm]{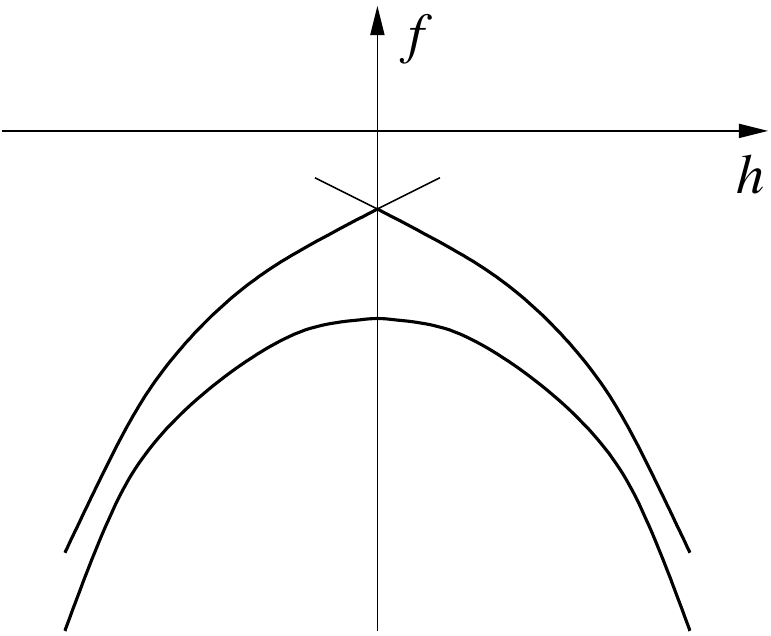}}
\caption{Qualitative graphs of the free energy $f(\beta,h)$ as a function of $h$, for $\beta$ large (top) and $\beta$ small (bottom).}
\label{fig fen}
\efig

\begin{proposition}
\label{prop ineq res th} The three order parameters above are related as follows:
\[
m^*_{\rm th}(\beta) = m^*_{\rm res}(\beta) \geq \tfrac12 m^*_{\rm sp}(\beta).
\]
\end{proposition}

\begin{proof}[Proof of $m^*_{\rm th} = m^*_{\rm res}$]
We prove that whenever $f_n$ is a sequence of differentiable concave functions that converge pointwise to the (necessarily concave) function $f$, then
\be
\label{a very nice identity!}
D_+ f(0) = \lim_{h\to0+} \limsup_{n\to\infty} f_n'(h) = \lim_{h\to0+} \liminf_{n\to\infty} f_n'(h).
\ee
Up to the signs, the left side is equal to $m^*_{\rm th}$ and the right side to $m^*_{\rm res}$ and we obtain the identity in Proposition \ref{prop ineq res th}.
The proof of \eqref{a very nice identity!} follows from the general properties
\be
\begin{split}
&\limsup_i \bigl( \inf_j a_{ij} \bigr) \leq \inf_j \bigl( \limsup_i a_{ij} \bigr), \\
&\liminf_i \bigl( \sup_j a_{ij} \bigr) \geq \sup_j \bigl( \liminf_i a_{ij} \bigr),
\end{split}
\ee
and from the following expressions for left- and right-derivatives of concave functions:
\be
D_- f(h) = \inf_{s>0} \frac{f(h) - f(h-s)}s, \qquad D_+ f(h) = \sup_{s>0} \frac{f(h+s) - f(h)}s.
\ee
With these observations, the proof is straightforward. For $h>0$,
\be
\begin{split}
D_+ f(0) &\geq D_- f(h) = \inf_{s>0} \limsup_{n\to\infty} \frac{f_n(h) - f_n(h-s)}s \geq \limsup_{n\to\infty} f_n'(h) \\
&\geq \liminf_{n\to\infty} f_n'(h) \geq \sup_{s>0} \liminf_{n\to\infty} \frac{f_n(h+s) - f_n(h)}s = D_+ f(h).
\end{split}
\ee
Since right-derivatives of concave functions are right-continuous, the last term converges to $D_+ f(0)$ as $h \to 0+$. This proves Eq.\ \eqref{a very nice identity!}.
\end{proof}

\begin{proof}[Proof of $m^*_{\rm res} \geq \frac12 m^*_{\rm sp}$]
Let $h>0$, and let $\{\varphi_{j}\}$ be an orthonormal set of eigenvectors of $H_{\Lambda_{n},0}$ and $M_{\Lambda_{n}}$ with eigenvalues $e_{j}$ and $m_{j}$, respectively. Because of the spin flip symmetry, we have
\be
\label{joli developpement}
\begin{split}
\langle M_{\Lambda_{n}} \rangle_{\Lambda_{n},\beta,h} &= \frac{\sum_{j : m_{j} > 0} m_{j} \e{-\beta e_{j}} \bigl( \e{\beta hm_{j}} - \e{-\beta h m_{j}} \bigr)}{\sum_{j : m_{j} > 0} \e{-\beta e_{j}} \bigl( \e{\beta hm_{j}} + \e{-\beta h m_{j}} \bigr) + \sum_{j: m_{j}=0} \e{-\beta e_{j}}} \\
&\geq \frac{\sum_{j : m_{j} > 0} m_{j} \e{-\beta e_{j} + \beta hm_{j}} \bigl( 1 - \e{-2\beta h m_{j}} \bigr)}{2\sum_{j : m_{j} > 0} \e{-\beta e_{j} + \beta hm_{j}} + \sum_{j: m_{j}=0} \e{-\beta e_{j}}}.
\end{split}
\ee
After division by $n^{d}$, we only need to consider those $j$ with $m_{j} \sim n^{d}$, in which case $\e{-2\beta h m_{j}} \approx 0$. We can therefore replace the parenthesis by 1 in the limit $n\to\infty$.
On the other hand, consider the function $G_{n}(h) = \frac1\beta \log \Tr \e{-\beta H_{\Lambda_{n},0} + \beta h |M_{\Lambda_{n}}|}$. One can check that it is convex in $h$, see \eqref{c'est concave}, so $G_{n}'(h) \geq G_{n}'(0)$. Its derivative can be expanded as above, so that
\be
G_{n}'(h) = \frac{\sum_{j} |m_{j}| \e{-\beta e_{j} + \beta h |m_{j}|}}{\sum_{j} \e{-\beta e_{j} + \beta h |m_{j}|}}.
\ee
This is equal to twice the second line of \eqref{joli developpement} (without the parenthesis). Then
\be
m^{*}_{\rm res}(\beta) \geq \tfrac12  \lim_{h\to0+} \liminf_{n\to\infty} \frac1{n^{d}} G_{n}'(h) \geq \tfrac12 \liminf_{n\to\infty} \frac1{n^{d}} G_{n}'(0) = \tfrac12 m^{*}_{\rm sp}(\beta).
\ee
\end{proof}

\subsection{Antiferromagnetic phase transition}
\label{sec thermo limit anti}

While ferromagnets favor alignment of the spins, antiferromagnets favor staggered phases, where spins are aligned on one sublattice and aligned in the opposite direction on the other sublattice. The external magnetic field does not play much of a r\^ole. One could mirror the ferromagnetic situation by introducing a non-physical staggered magnetic field of the kind $h \sum_{x \in \caV} (-1)^{x} S_{x}^{(3)}$, which would lead to counterparts of the order parameters $m^{*}_{\rm th}$ and $m^{*}_{\rm res}$. We content ourselves with turning off the external magnetic field, i.e.\ setting $h=0$, and with looking at magnetic long-range order. For $x,y \in \caV$, we introduce the correlation function
\be
\sigma_{\Lambda,\beta}(x,y) = (-1)^{x} (-1)^{y} \langle S_{x}^{(3)} S_{y}^{(3)} \rangle_{\Lambda,\beta,0}.
\ee
One question is whether
\be
\label{sigma bar}
\overline\sigma^{2}(\beta) = \liminf_{n\to\infty} \frac1{|\caV_{n}|^{2}} \sum_{x,y\in\caV_{n}} \sigma_{\Lambda,\beta}(x,y)
\ee
differs from 0. A related question is whether the correlation function does not decay to 0 as the distance between $x$ and $y$ tends to infinity. One says that the system exhibits \emph{long-range order} if this happens.

In $\bbZ^{d}$ and for $\beta$ large enough, it is well-known that there is no long-range order and that the correlation function decays exponentially in $\|x-y\|$. Long-range order is expected in dimension $d\geq3$ but not in $d=1,2$. This is discussed in more detail in Section \ref{sec rigorous}.

\subsection{Phase transitions in cycle and loop models}
\label{sec prob phase transitions}

In this section, we clarify the relations between the order parameters of the quantum systems and the nature of cycles and loops. This yields probabilistic interpretations for the quantum results. We introduce two quantities, which apply simultaneously to cycles and loops. Recall that $\gamma_{x}$ denotes either the cycle that contains $(x,0)$ in the cycle model, or the loop that contains $(x,0)$ in the loop model. We write $\bbP$ and $\bbE$ for $\bbP^{\rm cycles}$ and $\bbE^{\rm loop}$ when equations hold in both cases.
\begin{itemize}
\item The fraction of vertices in \emph{infinite} cycles/loops:
\be
\eta_{\infty}(\beta,h) = \lim_{K \to \infty} \liminf_{n\to\infty} \frac1{n^{d}} \bbE_{\Lambda_{n}, \beta, h} \bigl( \#\{ x \in \caV_{n} : L(\gamma_{x}) > K \} \bigr).
\ee
\item The fraction of vertices in \emph{macroscopic} cycles/loops:
\be
\eta_{\rm macro}(\beta,h) = \lim_{\varepsilon \to 0+} \liminf_{n\to\infty} \frac1{n^{d}} \bbE_{\Lambda_{n}, \beta, h} \bigl( \#\{ x \in \caV_{n} : L(\gamma_{x}) > \varepsilon n^{d} \} \bigr).
\ee
\end{itemize}
It is clear that $\eta_{\infty}(\beta,h) \geq \eta_{\rm macro}(\beta,h)$. These two quantities relate to magnetization and long-range order as follows. The first two statements deal with cycles and the third statement deals with loops.

\begin{proposition}
\label{prop quantum vs prob}
\hfill
\begin{itemize}
\item[(a)] $\displaystyle m^{*}_{\rm res}(2\beta) \geq \tfrac12 \lim_{h \to 0+} \eta_{\infty}(\beta,h)$.
\item[(b)] $\displaystyle m^{*}_{\rm sp}(2\beta) > 0 \; \Longleftrightarrow \; \eta_{\rm macro}(\beta,0) > 0$.
\item[(c)] $\displaystyle \overline\sigma(2\beta) > 0 \; \Longleftrightarrow \; \eta_{\rm macro}(\beta,0) > 0$.
\end{itemize}
\end{proposition}

\begin{proof}
Let
\be
m(2\beta,h) = \liminf_{n\to\infty} \langle S_{0}^{(3)} \rangle_{\Lambda_{n},2\beta,h}.
\ee
We use $\tanh x \geq \tanh K \cdot 1_{x>K}$, which holds for any $K$, and Theorem \ref{thm Toth}, so as to get
\be
m(2\beta,h) \geq \tfrac12 \tanh(h K) \liminf_{n\to\infty} \bbP^{\rm cycles}_{\Lambda_{n},\beta,h}(L(\gamma_{0}) > K).
\ee
Taking $K\to\infty$, we get $m(2\beta,h) \geq \frac12 \eta_{\infty}(\beta,h)$. We now take $h\to0+$ to obtain (a).

For (b), we observe that, since the vertices of $\Lambda_{n}$ are exchangeable,
\be
\frac1{n^{2d}} \langle M^{2}_{\Lambda_{n}} \rangle_{\Lambda_{n},2\beta,0} = \frac1{2\beta} \bbE^{\rm cycles}_{\Lambda_{n},\beta,0} \Bigl( \frac{L(\gamma_{0})}{n^{d}} \Bigr).
\ee
It follows from Lemma \ref{lem magn vs corr} that
\be
m^{*}_{\rm sp}(2\beta) > 0 \; \Longleftrightarrow \; \liminf_{n\to\infty} \bbE^{\rm cycles}_{\Lambda_{n},\beta,0} \Bigl( \frac{L(\gamma_{0})}{n^{d}} \Bigr) > 0.
\ee
On the other hand, we have
\be
\eta_{\rm macro}(\beta,0) = \lim_{\varepsilon\to0+} \liminf_{n\to\infty} \bbP^{\rm cycles}_{\Lambda_{n},\beta,0} \Bigl( \frac{L(\gamma_{0})}{n^{d}} > \varepsilon \Bigr).
\ee
The result is then clear.

The claim (c) is identical to (b), with loops instead of cycles.
\end{proof}

It should be possible to extend Proposition \ref{prop quantum vs prob} (a) so that $m_{\rm res}^{*}(\beta) > 0 \Leftrightarrow \eta_{\infty}(\beta,0) > 0$. This suggests that $m^{*}_{\rm th}$ and $m^{*}_{\rm res}$ are related to the existence of infinite cycles, while $m^{*}_{\rm sp}$ is related to the occurrence of macroscopic cycles. The question is then whether there exists a phase in which a positive fraction of vertices belongs to mesoscopic cycles or loops. Such a phase could have something to do with the Berezinski\u\i-Kosterlitz-Thouless transition \cite{Berez,0022-3719-6-7-010}, which has been rigorously established in the classical XY model \cite{MR634447}. It is not expected in the Heisenberg model, though. The Mermin-Wagner theorem (Section \ref{sec MW}) rules out any kind of infinite cycles or loops in one and two dimensions.

\section{Rigorous results for the quantum models}
\label{sec rigorous}

Quantum lattice systems have seen a considerable amount of study in the past decades, and the effort is not abating. Physicists are interested in properties of the ground state (i.e., the eigenvector of the Hamiltonian with lowest eigenvalue), in dynamical behavior, and in the existence and nature of phase transitions. Out of many results, we only discuss two in this section, which have been chosen because of their direct relevance to the understanding of the cycle and loop models: the Mermin-Wagner theorem concerning the absence of spontaneous magnetization in one and two dimensions, and the theorem of Dyson, Lieb, and Simon concerning the existence of long-range order in the antiferromagnetic model.

\subsection{Mermin-Wagner theorem}
\label{sec MW}

This fundamental result of condensed matter physics states that a continuous symmetry cannot be broken in one and two dimensions \cite{Mermin:1966lr}. In particular, there is no spontaneous magnetization or long-range order in Heisenberg models.

\begin{theorem}
\label{thm MW}
Let $(\Lambda_{n}^{\rm per})_{n\geq1}$ be the sequence of cubic boxes in $\bbZ^{d}$ with periodic boundary conditions. For $d=1$ or 2, and for any $\beta \in [0,\infty)$,
\[
m^{*}_{\rm res}(\beta) = 0.
\]
\end{theorem}

By Proposition \ref{prop ineq res th}, all three ferromagnetic order parameters are zero, and there are no infinite cycles by Proposition \ref{prop quantum vs prob} in the cycle model that corresponds to the Heisenberg ferromagnet. The theorem can also be stated for the staggered magnetic field discussed in Section \ref{sec thermo limit anti}. One could establish antiferromagnetic counterparts to Lemma \ref{lem magn vs corr} and Proposition \ref{prop ineq res th}, and therefore prove that $\eta_{\infty}(\beta)$ is also zero in the loop model that corresponds to the Heisenberg antiferromagnet.

An open question is whether the theorem can be extended to more general measures of the form
\[
\vartheta^{|\caC(\omega)|} \dd\rho_{\caE,\beta}(\omega) \quad \text{and} \quad \vartheta^{|\caL(\omega)|} \dd\rho_{\caE,\beta}(\omega)
\]
(up to normalization), for values of $\vartheta$ other than $\vartheta=2$. The case $3^{|\caL(\omega)|}$ can actually be viewed as the representation of a model with Hamiltonian $-\sum_{\{x,y\} \in \caE} (\vec S_{x} \cdot \vec S_{y})^{2}$ (see \cite{MR1288152}) and the Mermin-Wagner theorem certainly holds in that case.

The theorem may not apply when $\vartheta$ is too large, and the system is in a phase with many loops, similar to the one studied in \cite{MR1749392}.

We present the standard proof \cite{MR0289084} that is based on Bogolubov's inequality.

\begin{proposition}[Bogolubov's inequality]
Let $\beta>0$ and $A,B,H$ be operators on a finite-dimensional Hilbert space, with $H$ self-adjoint. Then
\[
\bigl| \Tr [A,B] \e{-\beta H} \bigr|^{2} \leq \tfrac12 \beta \Tr (AA^{*} + A^{*}A) \e{-\beta H} \; \Tr \bigl[ [B,H], B^{*} \bigr] \e{-\beta H}.
\]
\end{proposition}

\begin{proof}
We only sketch the proof; see \cite{MR0289084} for more details. Let $\{\varphi_{i}\}$ be an orthonormal set of eigenvectors of $H$ and $\{e_{i}\}$ the corresponding eigenvalues. We introduce the following inner product:
\be
\label{funny inner product}
(A,C) = \sum_{i,j: e_{i} \neq e_{j}} \langle \varphi_{i}, A^{*} \varphi_{j} \rangle \langle \varphi_{j}, C \varphi_{i} \rangle \frac{\e{-\beta e_{j}} - \e{-\beta e_{i}}}{e_{i}-e_{j}}.
\ee
One can check that
\be
\label{une borne sup}
(A,A) \leq \tfrac12 \beta \Tr (AA^{*} + A^{*}A) \e{-\beta H}.
\ee
We choose $C = [B^{*},H]$, and we check that
\be
\label{numero un}
\Tr [A,B] \e{-\beta H} = \overline{(A,C)}
\ee
and
\be
\label{numero deux}
\Tr \bigl[ [B,H], B^{*} \bigr] \e{-\beta H} = (C,C).
\ee
Inserting \eqref{numero un} and \eqref{numero deux} in the Cauchy-Schwarz inequality of the inner product \eqref{funny inner product}, and using \eqref{une borne sup}, we get Bogolubov's inequality.
\end{proof}

\begin{proof}[Proof of Theorem \ref{thm MW}]
Let $m_{n}(\beta,h) = n^{-d} \langle M_{\Lambda_{n}} \rangle_{\Lambda_{n},\beta,h}$. Let
\be
S_{x}^{(\pm)} = \tfrac1{\sqrt2} ( S_{x}^{(1)} \pm \ii S_{x}^{(2)} ).
\ee
One easily checks that
\be
\label{ca commute pas}
[S_{x}^{(+)}, S_{y}^{(-)}] = S_{x}^{(3)} \delta_{x,y}.
\ee
It is convenient to label the sites of $\Lambda_{n}^{\rm per}$ as follows
\be
\caV_{n} = \{ x \in \bbZ^{d} : -\tfrac n2 < x_{i} \leq \tfrac n2, i=1,\dots,d \}.
\ee
$\caE_{n}$ is again the set of nearest-neighbors in $\caV_{n}$ with periodic boundary conditions. For $k \in \frac{2\pi}n \caV_{n}$, we introduce
\be
\label{def Sk}
S^{(\cdot)}(k) = \frac1{n^{d/2}} \sum_{x \in \caV_{n}} \e{-\ii k x} S_{x}^{(\cdot)},
\ee
where $kx$ denotes the inner product in $\bbR^{d}$. Then, using \eqref{ca commute pas},
\be
\begin{split}
\langle [S^{(+)}(k), S^{(-)}(-k)] \rangle_{\Lambda_{n},\beta,h} &= \frac1{n^{d}} \sum_{x,y \in \caV_{n}} \e{-\ii k x} \e{\ii k y} \langle [S_{x}^{(+)}, S_{y}^{(-)}] \rangle_{\Lambda_{n},\beta,h} \\
&= m_{n}(\beta,h).
\end{split}
\ee
This will be the left side of Bogolubov's inequality. For the right side, tedious but straightforward calculations (expansions, commutation relations) give
\bm
\bigl\langle \bigl[ [S^{(+)}(k), H_{\Lambda_{n}}], S^{(-)}(-k) \bigr] \bigr\rangle_{\Lambda_{n},\beta,h} \\
= \frac2{n^{d}} \sum_{x,y : \{x,y\} \in \caE_{n}} (1 - \e{\ii k (x-y)}) \bigl\langle S_{x}^{(-)} S_{y}^{(+)} + S_{x}^{(3)} S_{y}^{(3)} \bigr\rangle_{\Lambda_{n},\beta,h} + h m_{n}(\beta,h).
\end{multline}
Despite appearances, this expression is real and positive for any $k$, as can be seen from \eqref{numero deux}. We get an upper bound by adding the same quantity, but with $-k$. This yields
\[
\frac4{n^{d}} \sum_{x,y : \{x,y\} \in \caE_{n}} (1 - \cos k(x-y)) \bigl\langle S_{x}^{(-)} S_{y}^{(+)} + S_{x}^{(3)} S_{y}^{(3)} \bigr\rangle_{\Lambda_{n},\beta,h} + 2h m_{n}(\beta,h).
\]
From Lemma \ref{lem coupling}, we have
\be
\bigl| \bigl\langle S_{x}^{(-)} S_{y}^{(+)} + S_{x}^{(3)} S_{y}^{(3)} \bigr\rangle_{\Lambda_{n},\beta,h} \bigr| = \bigl| \langle \vec S_{x} \cdot \vec S_{y} \rangle_{\Lambda_{n},\beta,h} \bigr| \leq \tfrac34.
\ee
Let us now introduce the ``dispersion relation'' of the lattice:
\be
\label{def eps k}
\varepsilon(k) = \sum_{i=1}^{d} (1 - \cos k_{i}).
\ee
Inserting all of this into Bogolubov's inequality, we get
\be
\frac{m_{n}(\beta,h)^{2}}{3 \varepsilon(k) + 2|h m_{n}(\beta,h)|} \leq \beta \bigl\langle S^{(+)}(k) S^{(-)}(-k) + S^{(-)}(-k) S^{(+)}(k) \bigr\rangle_{\Lambda_{n},\beta,h}.
\ee
Summing over all $k \in \frac{2\pi}n \caV_{n}$, and using $\sum_{k} \e{-\ii k (x-y)} = \delta_{x,y}$, we have
\bm
\sum_{k} \bigl\langle S^{(+)}(k) S^{(-)}(-k) + S^{(-)}(-k) S^{(+)}(k) \bigr\rangle_{\Lambda_{n},\beta,h} \\
= \sum_{x \in \caV_{n}} \bigl\langle S_{x}^{(+)} S_{x}^{(-)} + S_{x}^{(-)} S_{x}^{(+)} \bigr\rangle_{\Lambda_{n},\beta,h} = n^{d}.
\end{multline}
Then
\be
m_{n}(\beta,h)^{2} \frac1{n^{d}} \sum_{k \in \frac{2\pi}n \caV_{n}} \frac1{3 \varepsilon(k) + 2|h m_{n}(\beta,h)|} \leq \beta.
\ee
As $n\to\infty$, we get a Riemann integral,
\be
m(\beta,h)^{2} \frac1{(2\pi)^{d}} \int_{[-\pi,\pi]^{d}} \frac{\dd k}{3 \varepsilon(k) + 2|h m_{n}(\beta,h)|} \leq \beta.
\ee
Since $\varepsilon(k) \approx k^{2}$ around $k=0$, the integral diverges when $h\to0$, and so $m(\beta,h)$ must go to 0.
\end{proof}

Notice that the integral remains finite for $d\geq3$; the argument only applies to $d=1,2$.

\subsection{Dyson-Lieb-Simon theorem of existence of long-range order}
\label{sec DLS}

Following the proof of Fr\"ohlich, Simon and Spencer of a phase transition in the classical Heisenberg model \cite{MR0421531}, Dyson, Lieb and Simon proved the existence of long-range order in several quantum lattice models, including the antiferromagnetic quantum Heisenberg model in dimensions $d\geq5$ \cite{MR0496246}. Further observations of Neves and Perez \cite{Neves1986331}, and of Kennedy, Lieb and Shastry \cite{MR980167}, imply that long-range order is present for all $d \geq 3$.\footnote{We are indebted to the anonymous referee for pointing this out and for clarifying this to us. The following explanation is essentially taken from the referee's report.} These articles use the ``reflection positivity'' method, which was systematized and extended in \cite{MR506363,MR570370}. We recommend the Prague notes of T\'oth \cite{Toth} and  Biskup \cite{MR2581604} for excellent introductions to the topic. See also the notes of Nachtergaele \cite{MR2310198}.

Recall the definition of $\overline\sigma$ in Eq.\ \eqref{sigma bar}.

\begin{theorem}[Dyson-Lieb-Simon]
\label{thm DLS}
Let $(\Lambda_{n}^{\rm per})$ be the sequence of cubic boxes in $\bbZ^{d}$, $d\geq3$, with even side lengths and  periodic boundary conditions. There exists $\beta_{0}<\infty$ such that, for all $\beta>\beta_{0}$, the Heisenberg antiferromagnet has long-range order,
\[
\overline\sigma(\beta) > 0.
\]
\end{theorem}

Clearly, this theorem has remarkable consequences for the loop model with weights $2^{|\caL(\omega)|}$. Indeed, there are macroscopic loops, $\eta_{\rm macro}(\beta,0) > 0$, provided that $\beta$ is large enough.

Despite many efforts and false hopes, there is no corresponding result for the Heisenberg ferromagnet, and hence for the cycle model.

The proof of Theorem \ref{thm DLS} for $d\geq5$ can be found in \cite{MR0496246} (see also \cite{MR506363,Toth} for useful clarifications). In the remainder of this section we explain how to use the observations of \cite{Neves1986331} and \cite{MR980167} in order to extend the result to dimensions $d=3$ and $d=4$. As these articles deal with ground state properties rather than positive temperatures, some modifications are needed. We warn the readers that this part of the notes is not really self-contained.

Recall the definitions of the operators $S^{(\cdot)}(k)$ in Eq.\ \eqref{def Sk}. We need the Duhamel two-point function, which is reminiscent of the Duhamel formula of Proposition \ref{prop Duhamel}.
\be
(S^{(j)}(k),S^{(j)}(-k))_{\Lambda,\beta,0} = \frac1{Z_{\Lambda}(\beta,0)} \int_{0}^{1} \! \Tr \e{-s\beta H_{\Lambda,0}} S^{(j)}(k) \e{-(1-s) \beta H_{\Lambda,0}} S^{(j)}(-k) \dd s.
\ee
Recall also the definition of $\varepsilon(k)$ in \eqref{def eps k}, and let $\vec\pi = (\pi,\dots,\pi) \in \bbR^{d}$. We have
\be
\varepsilon(k-\vec\pi) =  \sum_{i=1}^{d} (1 + \cos k_{i}).
\ee
Let $e_{n}(\beta)$ denote the negative of the mean energy per site, i.e.,
\be
e_{n}(\beta) = - \Bigl\langle \frac{H_{\Lambda,0}}{n^{d}} \Bigr\rangle_{\Lambda,\beta,0}.
\ee
One can show that $e_{n}(\beta)$ is nonnegative, increasing with respect to $\beta$, and that it converges pointwise as $n\to\infty$.

The main result of reflection positivity is the following ``Gaussian domination''.

\begin{proposition}
\label{prop Gaussian dom}
If $k \in \frac{2\pi}n \caV_{n}$ and $k \neq \vec\pi$, we have
\begin{itemize}
\item[(a)] $\displaystyle (S^{(j)}(k),S^{(j)}(-k))_{\Lambda_{n},\beta,0} \leq \frac1{2 \varepsilon(k-\vec\pi)}$,
\item[(b)] $\displaystyle \langle S^{(j)}(k) S^{(j)}(-k) \rangle_{\Lambda_{n},\beta,0} \leq \Bigl( \frac{e_{n}(\beta)}{6d} \Bigr)^{1/2} \Bigl( \frac{\varepsilon(k)}{\varepsilon(k-\vec\pi)} \Bigr)^{1/2} + \frac3{2\beta \varepsilon(k-\vec\pi)}$.
\end{itemize}
\end{proposition}

\begin{proof}[Sketch proof]
The claim (a) can be found in \cite{MR0496246}, Theorem 6.1. The claim (b) follows from Eqs (3), (5), and (6) of \cite{Neves1986331}, and from the relation
\be
\sum_{j=1}^{3} \bigl\langle \bigl[ S^{(j)}(k), [H_{\Lambda,0}, S^{(j)}(-k)] \bigr] \bigr\rangle_{\Lambda,\beta,0} = \tfrac4d \varepsilon(k) e_{n}(\beta).
\ee
This is Eq.\ (55) in \cite{MR0496246}.
\end{proof}

Next, let
\be
\overline\sigma_{n}(\beta) = \frac1{n^{2d}} \sum_{x,y\in\caV_{n}} (-1)^{x} (-1)^{y} \langle S_{x}^{(3)} S_{y}^{(3)} \rangle_{\Lambda,\beta,0}.
\ee
Then $\overline\sigma(\beta) = \liminf_{n} \overline\sigma_{n}(\beta)$, and the goal is to show that it differs from zero. For $k = \vec\pi$, we have
\be
\langle S^{(3)}(\vec\pi) S^{(3)}(-\vec\pi) \rangle_{\Lambda,\beta,0} = n^{d} \overline\sigma_{n}(\beta).
\ee
Kennedy, Lieb and Shastry \cite{MR980167} have proposed the following sum rule, which improves on the original one used in \cite{MR0421531,MR0496246}:
\be
\frac1{n^{d}} \sum_{k \in \frac{2\pi}n \caV_{n}} \langle S^{(3)}(k) S^{(3)}(-k) \rangle_{\Lambda,\beta,0} \cos k_{i} = \langle S_{0}^{(3)} S_{e_{i}}^{(3)} \rangle_{\Lambda,\beta,0},
\ee
where $e_{i}$ denotes the neighbor of the origin in the $i$th direction. Because of the symmetries of $\caV_{n}$ (translations and lattice rotations), we have
\be
\langle S_{0}^{(3)} S_{e_{i}}^{(3)} \rangle_{\Lambda,\beta,0} = -\frac{e_{n}(\beta)}{3d}.
\ee
The sum rule can be rewritten as
\be
\frac{e_{n}(\beta)}{3d} = \overline\sigma_{n}(\beta) + \frac1{d n^{d}} \sumtwo{k \in \frac{2\pi}n \caV_{n}}{k \neq \vec\pi} \langle S^{(3)}(k) S^{(3)}(-k) \rangle_{\Lambda,\beta,0} \Bigl( -\sum_{i=1}^{d} \cos k_{i} \Bigr).
\ee
By Proposition \ref{prop Gaussian dom} (b), we have
\be
\begin{split}
\frac{e_{n}(\beta)}{3d} \leq \overline\sigma_{n}(\beta) &+ \frac{e_{n}(\beta)^{1/2}}{(6d)^{1/2} d} \frac1{n^{d}} \sumtwo{k \in \frac{2\pi}n \caV_{n}}{k \neq \vec\pi} \Bigl( \frac{\varepsilon(k)}{\varepsilon(k-\vec\pi)} \Bigr)^{1/2} \Bigl( -\sum_{i=1}^{d} \cos k_{i} \Bigr)_{+} \\
&+ \frac3{2d\beta} \frac1{n^{d}} \sumtwo{k \in \frac{2\pi}n \caV_{n}}{k \neq \vec\pi} \frac1{\varepsilon(k-\vec\pi)} \Bigl( -\sum_{i=1}^{d} \cos k_{i} \Bigr)_{+}.
\end{split}
\ee
As $n\to\infty$, with $e(\beta) = \lim_{n} e_{n}(\beta)$, we get
\be
\label{I am ugly}
\frac{e(\beta)}{3d} \leq \overline\sigma(\beta) + \frac{e(\beta)^{1/2}}{(6d)^{1/2} d} \, I(d) + \frac3{2d\beta} \frac1{(2\pi)^{d}} \int_{[-\pi,\pi]^{d}} \frac1{\varepsilon(k-\vec\pi)} \Bigl( -\sum_{i=1}^{d} \cos k_{i} \Bigr)_{+} \dd k,
\ee
where
\be
I(d) = \frac1{(2\pi)^{d}} \int_{[-\pi,\pi]^{d}} \Bigl( \frac{\varepsilon(k)}{\varepsilon(k-\vec\pi)} \Bigr)^{1/2} \Bigl( -\sum_{i=1}^{d} \cos k_{i} \Bigr)_{+} \dd k.
\ee
The last integral in \eqref{I am ugly} is finite when $d\geq3$, and this term may be made arbitrarily small by choosing $\beta$ large enough. It follows that a sufficient condition for $\overline\sigma(\beta) > 0$ for large enough $\beta$, is that
\be
\label{that is the question}
\lim_{\beta\to\infty} \frac{e(\beta)^{1/2}}{3d} > \frac1{(6d)^{1/2} d} \, I(d).
\ee
The integral $I(d)$ can be calculated numerically: $I(3) = 1.04968...$ and $I(4) = 1.01754...$ It is then enough to show that $\lim_{\beta\to\infty} e(\beta) > 0.5509...$ in $d=3$ and $\lim_{\beta\to\infty} e(\beta) > 0.3883...$ in $d=4$. The following lemma allows us to conclude that long-range order indeed takes place in $d=3$ and $d=4$.

\begin{lemma}
\[
\lim_{\beta\to\infty} e(\beta) \geq \frac d4.
\]
\end{lemma}

\begin{proof}
The Gibbs variational principle states that
\be
\label{Gibbs principle}
F_{\Lambda}(\beta,h) \leq \Tr \rho H_{\Lambda,h} - \frac1\beta S_{\Lambda}(\rho)
\ee
for any operator $\rho$ in $\caH^{(\caV)}$ such that $\rho\geq0$ and $\Tr \rho = 1$. Here, $S_{\Lambda}$ is the Boltzmann entropy,
\be
S_{\Lambda}(\rho) = -\Tr \rho \log \rho.
\ee
See e.g.\ Proposition IV.2.5 in \cite{MR1239893} (the setting in \cite{MR1239893} involves a normalized trace, hence there are a few discrepancies between our formul\ae \  and those in the book).
It is known that the Gibbs state $\rho = Z_{\Lambda}(\beta,h)^{-1} \e{-\beta H_{\Lambda,h}}$ saturates the inequality, and that the entropy satisfies the bounds
\be
0 \leq S_{\Lambda}(\rho) \leq |\caV| \log2.
\ee
It follows that
\be
e(\beta) \geq -f(\beta,0) - \frac{\log2}\beta.
\ee
In order to get a bound for the free energy, we use \eqref{Gibbs principle} with the N\'eel state $\Psi_{\text{N\'eel}}$ as a trial state,
\be
\Psi_{\text{N\'eel}} = \bigotimes_{x\in\Lambda_{n}} \bigl| (-1)^{x} \tfrac12 \bigr\rangle.
\ee
With $\rho$ the projector onto $\Psi_{\text{N\'eel}}$, we have $S_{\Lambda_{n}}(\rho) = 0$, and
\be
\begin{split}
F_{\Lambda_{n}}(\beta,0) &\leq \langle \Psi_{\text{N\'eel}}, H_{\Lambda_{n},0} \Psi_{\text{N\'eel}} \rangle \\
&= d n^{d} \langle \tfrac12, -\tfrac12| \vec S_{x} \cdot \vec S_{y} |\tfrac12, -\tfrac12 \rangle.
\end{split}
\ee
The last inner product is in $\caH_{x} \otimes \caH_{y}$.
Using \eqref{coupling vs transposition}, we find that it is equal to $-\frac14$.
\end{proof}

These results do not apply to dimension $d=2$ because the last integral in \eqref{I am ugly} is divergent. We already know that the magnetization is zero for all finite values of $\beta$ by the Mermin-Wagner theorem. An important question, which remains open to this day, is whether long-range order occurs in the ground state of the two-dimensional antiferromagnet. The last integral in \eqref{I am ugly} disappears if the limit $\beta\to\infty$ is taken before the infinite volume limit, and the question is whether \eqref{that is the question} is true. Since $I(2) = 1.29361...$ one needs $\lim_{\beta\to\infty} e(\beta) > 1.255...$. But the limit is expected to be around 0.67 \cite{MR980167} and so the method does not apply.

In contrast to the antiferromagnet, the ground state of the ferromagnet is trivial with full magnetization. If $\beta$ is taken to infinity in the cycle model for a fixed graph, the spatial structure is lost and the resulting random permutation has Ewens distribution (that is, it is weighted by $2^{|\caC|}$). Almost all vertices belong to macroscopic cycles and the cycle lengths are distributed according to the Poisson-Dirichlet distribution PD$_{2}$.

\section{Rigorous results for cycle and loop models}

The cycle and loop representations in Theorems \ref{thm Toth} and \ref{thm AN} are interesting in 
their own right and can be studied using purely probabilistic techniques.  Without
the physical motivation, the external magnetic field is less relevant and 
more of an annoyance.  We prefer to switch it off.  The models in this simpler
situation are defined below, with the small generalization that the geometric 
weight on the number of cycles or loops is arbitrary.  This is analogous to how,
for example, one obtains the random cluster or Fortuin-Kasteleyn representation from the Ising model.  

\subsection{Cycle and loop models}\label{s:cycleloopmodel}

As usual we suppose that $\Lambda = (\caV,\caE)$ is a finite undirected graph.  Recall that the 
Poisson edge measure $\rho_{\caE,\beta}$ is obtained by attaching independent Poisson point processes on $[0,\beta]$ to each edge of $\caE$.  

For each realization $\omega$ of the Poisson edge process, we define cycles $\Cycles(\omega)$ and loops $\Loops(\omega)$ 
as in  \S\ref{sec Poisson conf}.  
The random cycle and loop models are obtained via a change of measure 
in which the number of cycles or loops receives a geometric weight $\vartheta > 0$.  
That is, the probability measures of interest are
\be
\label{nice measures}
\begin{split}
&\bbP_{\Lambda,\beta}^{\rm cycles}(\dd\omega) = Z_{\Lambda}^{\rm cycles}(\beta)^{-1} \vartheta^{|\caC(\omega)|} \rho_{\caE,\beta}(\dd\omega), \\
&\bbP_{\Lambda,\beta}^{\rm loops}(\dd\omega) = Z_{\Lambda}^{\rm loops}(\beta)^{-1} \vartheta^{|\caL(\omega)|} \rho_{\caE,\beta}(\dd\omega),
\end{split}
\ee
where $Z_{\Lambda}^{\cdots}(\beta)$ are the appropriate normalizations.
As remarked above, $\vartheta = 2$ is the physically relevant choice in both these measures.

The main question deals with the possible occurrence of cycles or loops of diverging lengths. Recall the definitions of the fraction of vertices in infinite cycles/loops, $\eta_{\infty}(\beta)$, and the fraction of vertices in macroscopic cycles/loops, \linebreak $\eta_{\rm macro}(\beta)$, which were defined in Section \ref{sec prob phase transitions}. (We drop the dependence in $h$, since $h=0$ here.) In the case where the graph is a cubic box in $\bbZ^{d}$ with periodic boundary conditions, and $\vartheta=2$, the Mermin-Wagner theorem rules out infinite cycles in one and two dimensions, and the theorem of Dyson-Lieb-Simon shows that macroscopic loops are present in $d\geq3$, provided that the parameter $\beta$ is sufficiently large.

It is intuitively clear that there cannot be infinite cycles or loops when $\beta$ is small.  In Section \ref{sec high T} we prove this is indeed the case and give an explicit lower bound on the critical value of $\beta$.

The model for $\vartheta=1$ is known as random stirring or the interchange process.  The question of the existence of infinite cycles in this setting has been considered by several authors.  Angel considered the model on regular trees, and proved the existence of infinite cycles (for $\beta$ lying in an appropriate interval) when the degree of the tree is larger than 5 \cite{MR2042369}.  Schramm considered the model on the complete graph and obtained a fairly precise description of the asymptotic cycle length distribution \cite{MR2166362}.  We review this important result in Section \ref{s:complete}.  Recently, Alon and Kozma found a surprising formula for the probability that the permutation is cyclic, using representation theory \cite{2010arXiv1009.3723A}.

\subsection{No infinite cycles at high temperatures}
\label{sec high T}

We consider general graphs $\Lambda = (\caV,\caE)$. We let $\kappa$ denote the maximal degree of the graph, i.e., $\kappa = \sup_{x \in \caV} |\{ y : \{x,y\} \in \caE \}|$. Recall that $L(\gamma_{x})$ denotes the length of the cycle or loop that contains $x \times \{0\}$.  Let $a$ be the small parameter
\be
a = \begin{cases} \vartheta^{-1} (1-\e{-\beta}) & \text{if } \vartheta \leq 1, \\ 1-\e{-\beta} & \text{if } \vartheta \geq 1. \end{cases}
\ee
in the case of cycles and 
\be
a = \begin{cases} \vartheta^{-1} (1-\e{-\beta}) & \text{if } \vartheta \leq 1, \\ \e{-\beta} (\e{\beta\vartheta} - 1) & \text{if } \vartheta \geq 1. \end{cases}
\ee
in the case of loops.

\begin{theorem}\label{t:hightempexpdecay}
For either the cycle or the loop model, i.e., for either measure in \eqref{nice measures}, we have
\[
\bbP_{\Lambda,\beta}(L(\gamma_x) > \beta k) \leq (a (\kappa-1))^{-1} [a \kappa (1 - \tfrac1\kappa)^{-\kappa+1}]^k.
\]
for every $x \in \caV$.
\end{theorem}

Of course, the theorem is useful only if the right-hand side is less than 1, in which case large cycles have exponentially small probability. This result is pretty reasonable on the square lattice with $\vartheta\leq1$. When $\vartheta>1$, configurations with many cycles are favored, and the domain should allow for larger $\beta$. Our condition does not show it. The case $\vartheta\gg1$ is close to the situation treated in \cite{MR1749392} with phases of closely packed loops. In the case of the complete graph on $N$ vertices and $\vartheta=1$, the maximal degree is $\kappa = N-1$ and the optimal condition is $\beta < 1/N$ (Erd\H{o}s-R\'enyi, \cite{MR0125031}). Using $a\kappa \leq \beta N$ and $(1 - \tfrac1\kappa)^{-\kappa+1} \leq \e{}$, we see that our condition is off by a factor of $\e{}$.

As a consequence of the theorem, we have $\eta_{\infty}(\beta)=0$ for small enough $\beta$. This implies that $m^{*}_{\rm sp}(\beta) = \overline\sigma(\beta) = 0$ in the corresponding Heisenberg ferromagnet and antiferromagnet. One could extend the claim so that $m^{*}_{\rm th}(\beta) = 0$ as well.

\begin{proof}
Given $\omega$, let $G(\omega) = (V,E)$ denote the subgraph of $\Lambda$ with edges
\be
E = \{ e_{i}: (e_{i},t_{i}) \in \omega \},
\ee
and $V = \cup_{i} e_{i}$ the set of vertices that belong to at least one edge. $G(\omega)$ can be viewed as the percolation graph of $\omega$, where an edge $e$ is open if at least one bridge of the form $(e,t)$ occurs in $\omega$. Then we denote $C_{x}(\omega) = (V_{x},E_{x})$ the connected component of $G(\omega)$ that contains $x$. It is clear that $L(\gamma_{x}) \leq \beta |V_{x}|$ for both cycles and loops. Then, using Markov's inequality,
\be
\bbP_{\Lambda,\beta}(L(\gamma_x) > \beta k) \leq \bbP_{\Lambda,\beta}(|V_x| > k) \leq \alpha^{-k} \, \bbE_{\Lambda,\beta}(\alpha^{|V_x|}),
\ee
for any $\alpha \geq 1$. 

We consider first the case of cycles.  Given a subgraph $G' = (V',E')$ of $\Lambda$, let
\be
\label{def phi}
\phi(G') = \vartheta^{-|V'|} \int \Indi{[G(\omega)=G']} \vartheta^{|\caC(\omega)|} \dd\rho_{E',\beta}(\omega).
\ee

By partitioning $\Omega$ according to the connected components of $G(\omega)$, then using the fact 
that $\rho_{\caE,\beta}$ is a product measure over edges and that cycles are contained entirely 
within connected components, we have
\be \label{e:expectationintermsofphi}
 \bbE^{\rm cycles}_{\Lambda,\beta}(\alpha^{|V_x'|}) = \sum_{C_x'} \phi(C_x') \alpha^{|V_x'|} 
 \frac{\sum_{G' \cap C_x' = \emptyset} \phi(G')}{\sum_{G''} \phi(G'')}
\ee

The first sum is over connected subgraphs $C_x' = (V_x',E_x')$ of $\Lambda$ that contain $x$.  The second sum is over
subgraphs $G' = (V',E')$ that are compatible with $C_x'$, in the sense that $V' \cap V_x' = \emptyset$ and $V' \cup C_x' = \caV$.  The sum in the denominator is over all subgraphs $G'' = (V'',E'')$ with $V'' = \caV$.  

Notice that for any $C_x'$, the corresponding compatible graph $G' = (V', E')$ can be enlarged to $G'' = (\caV, E')$ 
by adding the vertices from $V_x'$.  The new vertices from $V_x'$ are all disconnected in $G''$.  Thus, if $G(\omega) = G''$, 
each vertex in $V_x'$ necessarily forms a single cycle of length 1.  It follows that $\phi(G'') = \phi(G')$.  Furthermore,
different $G'$ give rise to different $G''$.  So, the ratio in \eqref{e:expectationintermsofphi} 
is less than 1.

Now we claim that
\be \label{e:phiboundaE}
\phi(G') \leq a^{|E'|}
\ee 
for any connected $G'$.  First consider $\vartheta \leq 1$.
Since $G'$ is connected we have $|E'| \geq |V'| - 1$.  So, 
$\vartheta^{-|V'| + |\Cycles(\omega)|} \leq \vartheta^{-|V'| + 1} \leq \vartheta^{-|E'|}$ for any $\omega$.
When $\vartheta > 1$, 
use $|\Cycles(\omega)| \leq |V'|$ to see $\vartheta^{-|V'| + |\Cycles(\omega)|} \leq 1$.

On the other hand, $G(\omega)=G'$ holds if and only if the Poisson process for each edge of $G'$ contains at least one point.  So,
\be 
\int \Indi{[G(\omega)=G']} \dd\rho_{E',\beta}(\omega) = (1-\e{-\beta})^{|E'|},
\ee
and \eqref{e:phiboundaE} follows in the case of cycles.


The same bound also holds for the loop model when $\vartheta^{|\caC(\omega)|}$ is replaced by $\vartheta^{|\caL(\omega)|}$ in \eqref{def phi}.  For $\vartheta \leq 1$ the argument is the same as before.  
For $\vartheta > 1$, we use the inequality $|\Loops(\omega)| \leq |V'| + |\omega|$ that holds for any $\omega$, where $|\omega|$ is the number of bridges in $\omega$.  
This follows from the fact that each bridge in $\omega$ either splits a loop into two or merges two loops (see Lemma \ref{l:addremswap}), and that $|\Loops(\omega)| = |V'|$ when $\omega = \emptyset$.  Hence,
\be
\phi(G') \leq \int \vartheta^{|\omega|} \Indi{[G(\omega)=G']} \dd\rho_{E',\beta}(\omega) = \Bigl( \e{-\beta} \sum_{n=1}^{\infty} \frac{(\vartheta\beta)^{n}}{n!} \Bigr)^{|E'|},
\ee
which gives the bound \eqref{e:phiboundaE} for loops.


Combining \eqref{e:expectationintermsofphi} and \eqref{e:phiboundaE} shows that for either loops or cycles,
\be
\bbE_{\Lambda,\beta}(\alpha^{|V_x|}) = \sum_{C_x'} \phi(C_x') \alpha^{|V_x'|} \frac{\sum_{G' \cap C_x' = \emptyset} \phi(G')}{\sum_{G'} \phi(G')} \leq \sum_{C_x'} \alpha^{|V_{x}'|} a^{|E_{x}'|}.
\ee
Let $\delta(C_{x}')$ denote the ``depth'' of the connected graph $C_{x}'$, i.e., the minimal number of edges of $E_{x}'$ that must be crossed in order to reach any point of $V_{x}'$. Let
\be
B(\ell) = \sum_{C_x', \delta(C_{x}') \leq \ell} \alpha^{|V_{x}'|} a^{|E_x'|}.
\ee
We want an upper bound for $B(\ell)$ for any $\ell$. We show by induction that $B(\ell) \leq b$ for a number $b$ to be determined shortly. We proceed by induction on $\ell$. The case $\ell=0$ is $\alpha \leq b$. For $\ell+1$, we write the sum over graphs with depth less than $\ell+1$, attached at $x$, as a sum over graphs of depth less than $\ell$, attached at neighbors of $x$. Neglecting overlaps gives the following upper bound:
\be
\begin{split}
B(\ell+1) &\leq \alpha \prod_{y : \{x,y\} \in \caE} \Bigl( 1 + a \sum_{C_y', \delta(C_{y}') \leq \ell} \alpha^{|V_{y}'|} a^{|E_y'|} \Bigr) \\
&\leq \alpha (1 + ab)^\kappa.
\end{split}
\ee
This needs to be less than $b$; this condition can be written $a \leq b^{-1} ((b/\alpha)^{1/\kappa} - 1)$. The optimal choice that maximizes the possible values of $a$ is $b = \alpha (1-\frac1\kappa)^{-\kappa}$. A sufficient condition is then
\be
a \leq \tfrac1{\alpha\kappa} (1 - \tfrac1\kappa)^{\kappa-1}
\ee
We have obtained that
\be
\bbP_{\Lambda,\beta}(L(\gamma_{x}) > \beta k) \leq \alpha^{-k+1} (1-\tfrac1\kappa)^{-\kappa},
\ee
and this holds for all $1 \leq \alpha \leq \frac1{a\kappa} (1-\tfrac1\kappa)^{\kappa-1}$. Choosing the maximal value for $\alpha$, we get the bound of the theorem.
\end{proof}

\subsection{Rigorous results for the complete graph}
\label{s:complete}

 

Suppose $T_{1}, T_{2}, T_{3}, \ldots$ are independent random transpositions of pairs of elements of $\{1,2, \ldots, n\}$ 
and $\pi_{k} = T_{1} \circ T_{2} \circ \ldots \circ T_{k}$.  Write $\lambda(\pi_{k})$ 
for the vector of cycle lengths in $\pi_{k}$, sorted into decreasing order.  So, $\lambda_{i}(\pi_{k})$
is the size of the $i^{th}$ largest cycle and if there are fewer than $i$ cycles in $\pi_{k}$, 
we take $\lambda_{i}(\pi_{k}) = 0$. 

Note the simple connection between cycles here and the cycles in our model; if $N$ is a Poisson random 
variable with mean $\beta n(n-1)/2$, independent of the $T_{i}$, then
$\lambda(\pi_{N})$ has exactly the distribution of the ordered cycle lengths in $\Cycles$ under $\rho_{K_{n},\beta}$, 
where $K_{n}$ is the complete graph with $n$ vertices.

Schramm proved that for $c > 1/2$, an asymptotic fraction $\eta_{\infty} = \eta_{\infty}(2c)$ of elements from $\{1,2,\ldots,n\}$ lie in 
infinite cycles of $\pi_{\lfloor cn \rfloor}$ as $n \to \infty$.  The (non-random) fraction $\eta_{\infty}(2c)$ 
turns out to be the asymptotic fraction of vertices lying in the giant component of the Erd\H{o}s-R\'enyi 
random graph with edge probability $c/n$.  Equivalently, $\eta_{\infty}(s)$ is the survival probability for a Galton-Watson branching
process with Poisson offspring distribution with mean $s$.  Berestycki \cite{Ber10} proved a similar result.

Furthermore, Schramm also showed that the normalised cycle lengths converge to the Poisson-Dirichlet(1) distribution.  


\begin{theorem}[Schramm \cite{MR2166362}]
Let $c > 1/2$.  The law of $\lambda(\pi_{\lfloor cn \rfloor})/(n\eta_{\infty}(2c))$ converges weakly to $\PDlaw_{1}$ as $n \to \infty$.
\end{theorem}

\section{Uniform split-merge and its invariant measures}

We now take a break from spin systems and consider a random evolution on partitions of $[0,1]$
in which blocks successively split or merge.  Stochastic processes incorporating the phenomena of coalescence and fragmentation have been much studied in the recent probability literature (see, for example, \cite{AldousSurvey,BertoinCoagFrag} or Chapter 5 of \cite{PitmanStFl}, and their bibliographies).  The space of partitions of $[0,1]$ provides a natural setting for such processes.  The particular model we will discuss here has the property that the splitting and merging can be seen to balance each other out in the long run, so that there exists a stationary (or invariant) distribution.  Our aim is to summarise what is known 
about this invariant distribution.  Only a basic familiarity with probability theory 
is assumed and we will recall the essentials as we go.  This section is self-contained and can be read independently of the first.  As is the way among probabilists,
we assume there is a phantom probability space $(\Omega, \cF, \Prob)$ that hosts all our 
random variables.  It is summoned only when needed.

\subsection{Introduction}

Let $\PSs$ denote the space of (decreasing, countable) partitions of $[0,1]$.  Formally
\begin{equation}
\PSs \eqdef \Bigl\{ p \in [0,1]^{\N}: \; p_{1} \geq p_{2} \geq \ldots, \; \sum_{i} p_{i} = 1  \Bigr\},
\end{equation}
where the size of the $i^{th}$ part (or block) of $p \in \PSs$ is $p_{i}$.  
We define split and merge operators $S^{u}_{i}, M_{ij}:\PSs \to \PSs$, $u \in (0,1)$ as follows:
\begin{itemize}
\item $S^{u}_{i}p$ is the non-increasing sequence obtained by splitting $p_{i}$ into 
two new parts of size $up_{i}$ and $(1-u)p_{i}$, and
\item $M_{ij}p$ is the non-increasing sequence obtained by merging $p_{i}$ and $p_{j}$
into a part of size $p_{i} + p_{j}$.
\end{itemize}

\begin{figure}[htbp]
\begin{center}
\begin{picture}(0,0)%
\includegraphics{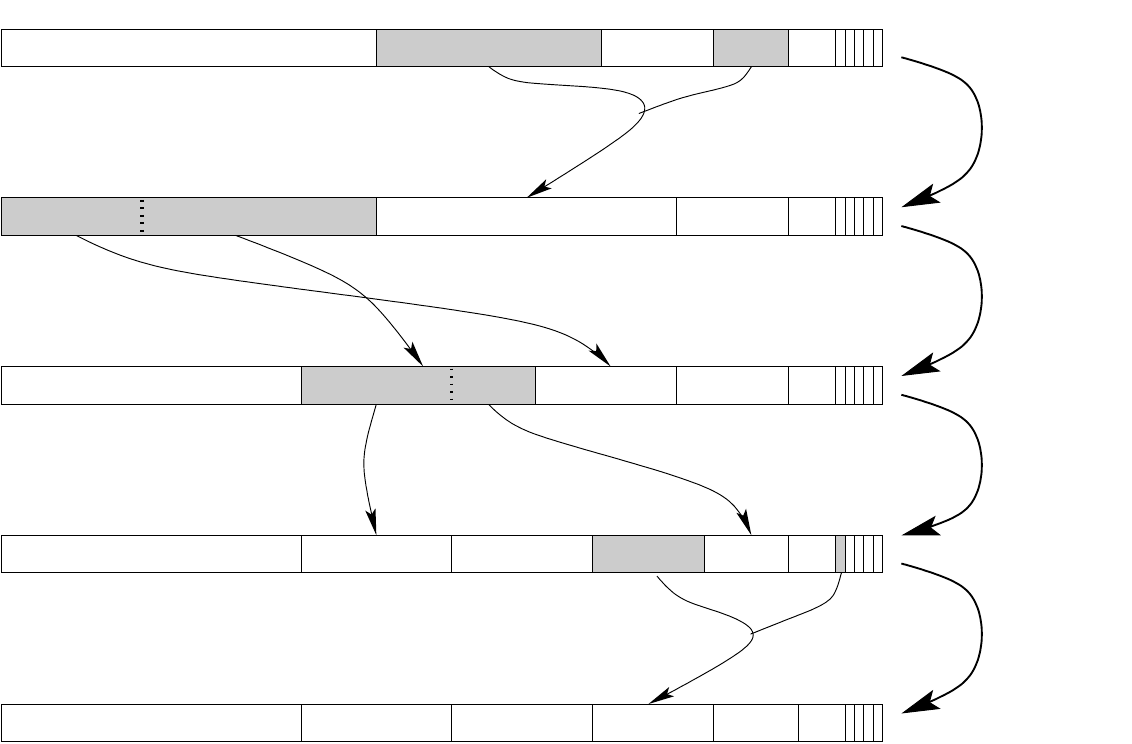}%
\end{picture}%
\setlength{\unitlength}{2368sp}%
\begingroup\makeatletter\ifx\SetFigFont\undefined%
\gdef\SetFigFont#1#2#3#4#5{%
  \reset@font\fontsize{#1}{#2pt}%
  \fontfamily{#3}\fontseries{#4}\fontshape{#5}%
  \selectfont}%
\fi\endgroup%
\begin{picture}(9054,5919)(1189,-6073)
\put(4726,-3586){\makebox(0,0)[lb]{\smash{{\SetFigFont{7}{8.4}{\rmdefault}{\mddefault}{\updefault}{\color[rgb]{0,0,0}$u$}%
}}}}
\put(4876,-286){\makebox(0,0)[lb]{\smash{{\SetFigFont{7}{8.4}{\rmdefault}{\mddefault}{\updefault}{\color[rgb]{0,0,0}{$2$}}%
}}}}
\put(7051,-286){\makebox(0,0)[lb]{\smash{{\SetFigFont{7}{8.4}{\rmdefault}{\mddefault}{\updefault}{\color[rgb]{0,0,0}$4$}%
}}}}
\put(6226,-4336){\makebox(0,0)[lb]{\smash{{\SetFigFont{7}{8.4}{\rmdefault}{\mddefault}{\updefault}{\color[rgb]{0,0,0}$4$}%
}}}}
\put(4726,-2986){\makebox(0,0)[lb]{\smash{{\SetFigFont{7}{8.4}{\rmdefault}{\mddefault}{\updefault}{\color[rgb]{0,0,0}$2$}%
}}}}
\put(7801,-4336){\makebox(0,0)[lb]{\smash{{\SetFigFont{7}{8.4}{\rmdefault}{\mddefault}{\updefault}{\color[rgb]{0,0,0}$7$}%
}}}}
\put(2551,-1636){\makebox(0,0)[lb]{\smash{{\SetFigFont{7}{8.4}{\rmdefault}{\mddefault}{\updefault}{\color[rgb]{0,0,0}$1$}%
}}}}
\put(2296,-2191){\makebox(0,0)[lb]{\smash{{\SetFigFont{7}{8.4}{\rmdefault}{\mddefault}{\updefault}{\color[rgb]{0,0,0}$u$}%
}}}}
\put(9151,-1111){\makebox(0,0)[lb]{\smash{{\SetFigFont{9}{10.8}{\rmdefault}{\mddefault}{\updefault}{\color[rgb]{0,0,0}$M_{2,4}$}%
}}}}
\put(9151,-2536){\makebox(0,0)[lb]{\smash{{\SetFigFont{9}{10.8}{\rmdefault}{\mddefault}{\updefault}{\color[rgb]{0,0,0}$S_1^u$}%
}}}}
\put(9151,-3886){\makebox(0,0)[lb]{\smash{{\SetFigFont{9}{10.8}{\rmdefault}{\mddefault}{\updefault}{\color[rgb]{0,0,0}$S_2^u$}%
}}}}
\put(9151,-5236){\makebox(0,0)[lb]{\smash{{\SetFigFont{9}{10.8}{\rmdefault}{\mddefault}{\updefault}{\color[rgb]{0,0,0}$M_{4,7}$}%
}}}}
\end{picture}%
\end{center}
\caption{Illustration for the split-merge process. The partition undergoes a merge followed by two splits and another merge.}
\end{figure}

The basic uniform split-merge transformation of a partition $p$ is defined as follows.  
First we choose two parts of $p$ at random, with the $i^{th}$ part being chosen 
with probability $p_{i}$ (this is called size-biased sampling). 
The two parts, which we call $p_{I}$ and $p_{J}$, are chosen independently 
 and we allow repetitions.
 If the 
same part is chosen twice, i.e. $I = J$, sample a uniform random variable $U$ on $[0,1]$ and 
split $p_{I}$ into two new parts of size $Up_{I}$ and $(1-U)p_{J}$ (i.e. apply $S^{U}_{I}$).  
If different parts are chosen, i.e. $I \neq J$, then merge them by applying $M_{IJ}$.
This transformation gives a new (random) element of $\PSs$.  Conditional on plugging a state $p \in \PSs$ into the transformation, the distribution of the new element of $\PSs$ obtained is given by the so-called transition kernel
\begin{equation}
K(p,\cdot) \eqdef \sum_{i}p^{2}_{i} \int_{0}^{1} \delta_{S^{u}_{i}p}(\cdot) du +  \sum_{i \neq j} p_{i}p_{j}\delta_{M_{ij}p}(\cdot).
\end{equation}

Repeatedly applying the transformation gives a sequence $P = ( P^{k})_{k = 0,1,2,\ldots}$
of random partitions evolving in discrete time.   We assume that the updates at each step are independent.  
So, given $P^{k}$, the distribution of $P^{k+1}$ is independent of $P^{k-1},\ldots,P^{0}$.  
In other words, $P$ is a discrete time Markov process on $\PSs$ with transition kernel $K$.  We call
it the basic split-merge chain.

Several authors have studied the large time behaviour of $P$, and the related
issue of invariant probability measures, i.e. $\mu$ such that $\mu K = \mu$ (if the initial value $P^{0}$ is distributed according to $\mu$, then $P^{k}$ also has distribution given by $\mu$ at all subsequent 
times $k = 1,2,\ldots$).  

Recent activity began with Tsilevich \cite{MR1751188}.  In that paper the author showed that the
Poisson-Dirichlet($\theta$) distribution (defined in \S \ref{s:PoissDiricDefn} below 
and henceforth denoted $\PDlaw_{\theta}$) with parameter $\theta = 1$ is invariant.  
The paper contains the conjecture (of  Vershik) that $\PDlaw_{1}$ is the only invariant measure.

Uniqueness within a certain class of analytic measures was established by Mayer-Wolf, Zerner and Zeitouni in \cite{MR1902841}.  In fact they extended the basic split-merge transform described above
to allow proposed splits and merges to be rejected with a certain probability.  In particular, splits and merges are proposed as above but only accepted with probability $\beta_{\rm s} \in (0,1]$ and $\beta_{\rm m} \in (0,1]$ respectively, independently at different times.  The corresponding kernel is
\be
\begin{split}
K_{\beta_{\rm s},\beta_{\rm m}}(p,\cdot) \eqdef &\beta_{\rm s}\sum_{i}p^{2}_{i} \int_{0}^{1} \delta_{S^{u}_{i}p}(\cdot) du +  
\beta_{\rm m}\sum_{i \neq j} p_{i}p_{j}\delta_{M_{ij}p}(\cdot) \\ 
&+ \Bigl(1 - \beta_{\rm s}\sum_{i}p^{2}_{i} - \beta_{\rm m}\sum_{i \neq j} p_{i}p_{j} \Bigr)\delta_{p}(\cdot).
\end{split}
\ee

We call this $(\beta_{\rm s}, \beta_{\rm m})$ split-merge (the basic chain, of course, corresponds to $\beta_{\rm s} = \beta_{\rm m} = 1$).   The Poisson-Dirichlet distribution is still invariant, but the parameter is now $\theta = \beta_{\rm s}/\beta_{\rm m}$ (note that, in fact, any invariant distribution for the chain can depend on $\beta_{\rm s}$ and $\beta_{\rm m}$ only 
through $\theta$ since multiplying both acceptance probabilities by the same positive constant only affects the speed of the chain). 

Tsilevich \cite{tsilevich2001simplest} provided another insight into the large time behaviour of the 
the basic split-merge process ($\beta_{\rm s} = \beta_{\rm m} = 1$).   The main theorem is that 
if $P^{0} = (1,0,0,\ldots) \in \PSs$, then the law of $P$, sampled at a random Binomial($n,1/2$)-distributed time, converges to Poisson-Dirichlet(1) as $n \to \infty$.  

Pitman \cite{MR1930355} studied a related split-merge transformation, and by developing
results of Gnedin and Kerov, reproved Poisson-Dirichlet invariance and refined the uniqueness
result of \cite{MR1902841}.  In particular, the Poisson-Dirichlet distribution is the only invariant measure 
under which Pitman's split-merge transformation composed with `size-biased permutation' 
is invariant. 

Uniqueness for the basic chain's invariant measure was finally established by 
Diaconis, Mayer-Wolf, Zerner and Zeitouni in \cite{MR2044670}.  They coupled
the split-merge process to a discrete analogue on integer partitions of $\{1,2,\ldots,n\}$ and
then used representation theory to show the discrete chain is close to equilibrium
before decoupling occurs.

Schramm \cite{MR2166362} used a different coupling to give another uniqueness proof 
for the basic chain.  His arguments readily extend to allow $\beta_{\rm s} / \beta_{\rm m} \in (0,1]$. In summary,

\begin{theorem} \label{thm:invariance} \hfill
\begin{itemize}
\item[(a)] Poisson-Dirichlet($\beta_{\rm s}/\beta_{\rm m}$) is invariant for the 
uniform split-merge chain with $\beta_{\rm s}, \beta_{\rm m} \in (0,1]$.

\item[(b)] If $\beta_{\rm s} / \beta_{\rm m} \leq 1$, it is the unique invariant measure.
\end{itemize}
\end{theorem}

We give a short proof of part (a) in Section~\ref{ss:splitmergeinvariance} below.

\subsection{The Poisson-Dirichlet distribution}\label{s:PoissDiricDefn}


Write $\Measures_{1}(\PSs)$ for the set of probability measures on $\PSs$.  The Poisson-Dirichlet distribution $\PDlaw_{\theta} \in \Measures_{1}(\PSs)$, $\theta > 0$, 
is a one parameter family of laws introduced by Kingman in \cite{MR0368264}.  
It has cropped up in combinatorics, population genetics, number theory, Bayesian statistics
and probability theory.  
The interested reader may consult \cite{MR2663265,MR591166,MR2032426,MR1434129} for details of  
applications and extensions.  We will simply define it and give some basic properties.

There are two important characterizations of $\PDlaw_{\theta}$.  
We will introduce both, since one will serve to provide intuition and the other will be useful for calculations.  We start with the so-called `stick-breaking' construction.  Let 
$T_{1}, T_{2}, \ldots$ be independent Beta($1,\theta$) random variables
(that is, $\Prob(T_{i} > s) = (1-s)^{\theta}$; if $U$ is uniform on $[0,1]$, one can check that $1 - U^{1/\theta}$ is Beta($1,\theta$) distributed).
Form a random partition from the $T_{i}$ by letting 
the $k^{th}$ block take fraction $T_{k}$ of the unallocated mass.  
That is, the first block has size $P_{1} = T_{1}$, the second $P_{2} = T_{2}(1-P_{1})$ and 
$P_{k+1} = T_{k+1}(1 - P_{1} - \ldots - P_{k})$.  One imagines taking 
a stick of unit length and breaking off a fraction $T_{k+1}$ of what remains 
after $k$ pieces have already been taken.   A one-line induction argument shows that
$1 - P_{1} - \ldots - P_{k} = (1-T_{1})(1-T_{2})\ldots(1-T_{k})$, giving
\begin{equation}
P_{k+1} = T_{k+1}(1-T_{1})(1-T_{2})\ldots(1-T_{k}).
\end{equation}
In case it is unclear that $\sum_{i=1}^{\infty} P_{i} = 1$ almost surely, note that
\begin{equation}
\E\Bigl[1 - \sum_{i=1}^{k} P_{i}\Bigr] = \E\Bigl[\prod_{i=1}^{k} (1-T_{i}) \Bigr] = \Bigl(\int_{0}^{1} \theta t (1-t)^{\theta - 1} \Bigr)^{k} = (\theta + 1)^{-k} \to 0
\end{equation}
as $k \to \infty$.  So, the vector $(P_{[1]}, P_{[2]}, \ldots)$ of the $P_{i}$ sorted into decreasing order 
is an element of $\PSs$.   
It determines a unique measure $\PDlaw_{\theta} \in \Measures_{1}(\PSs)$.  It is interesting to note that the original vector $(P_1, P_2, \ldots)$ is obtained from $(P_{[1]}, P_{[2]}, \ldots)$ by size-biased re-ordering; its distribution is called the GEM (Griffiths-Engen-McCloskey) distribution.  In other words, consider the interval $[0,1]$ partitioned into lengths $(P_{[1]}, P_{[2]}, \ldots)$. Take a sequence $U_1, U_2, \ldots$ of i.i.d.\ uniform random variables on $[0,1]$.  Now list the blocks ``discovered'' by the uniforms in the order that they are found.  The resulting sequence has the same distribution as $(P_1, P_2, \ldots)$.

\subsubsection{Poisson Point processes}\label{s:PPPmu} 

Kingman's original characterization of $\PDlaw_{\theta}$ was made in terms of a suitable random 
point process on $\R_+$, which is a generalization of the usual Poisson counting process.  We now provide a crash course in the theory of such processes on a measurable space $(X,\BorelSets)$.  (The standard reference is \cite{Kingman}.)  Although we will only need this theory for $X = \R_+$, there is no extra cost for introducing it in general.   Let $\Measures(X)$ denote the set of $\sigma$-finite measures on $X$.  

Suppose that $\mu \in \Measures(X)$ and consider the special case $\mu(X) < \infty$.  Thus, $\mu(\cdot)/\mu(X)$ is a probability measure 
and we can sample, independently, points $Y_{1}, Y_{2}, \ldots$ according to 
this distribution.  Let $N_{0}$ be Poisson($\mu(X)$) distributed, so that $\Prob(N_{0} = n) = \frac{\mu(X)^{n}}{n!} \e{-\mu(X)}$.  
Conceptually, the Poisson point process with intensity measure $\mu$ 
is simply the random collection $\{Y_{1}, \ldots, Y_{N_{0}} \}$.  

Formally, the point process is defined in terms of a random counting measure $N$ which counts the number of random points 
lying in sets $A \in \BorelSets$ i.e.\ $N(A) = \sum_{i = 1}^{N_{0}} \Indi{Y_{i} \in A}$.
%
%
Thus $N(A)$ is a random variable, which has Poisson($\mu(A)$) distribution.  Indeed,
\begin{equation}
\begin{split}
\Prob(N(A) = k) &= \sum_{n=k}^{\infty}\Prob(N_{0} = n)\Prob\Bigl(\sum_{i = 1}^{N_{0}} \Indi{Y_{i} \in A} = k \Big| N_0 = n\Bigr) \\
& = \sum_{n=k}^{\infty} \frac{\mu(X)^{n}}{n!} \e{-\mu(X)} \left(\frac{n!}{k!(n-k)!}\right) \left( \frac{\mu(A)}{\mu(X)} \right)^{k}
\left( 1- \frac{\mu(A)}{\mu(X)} \right)^{n-k} \\
& = \e{-\mu(X)} \frac{\mu(A)^{k}}{k!} \sum_{n=k}^{\infty} \frac{1}{(n-k)!} (\mu(X) - \mu(A))^{n-k}  \\
& = \frac{\mu(A)^{k}}{k!} \e{-\mu(A)}.
\end{split}
\end{equation}
Similar calculations show that if $A_{1}, \ldots, A_{k} \in \BorelSets$ 
are disjoint then $N(A_{1}), \ldots, N(A_{k})$ are independent.    These properties turn out to be sufficient to completely specify the distribution of the random measure $N$.




\begin{definition}[Poisson point process]\label{def:PPP}
A Poisson point process on $X$ with intensity $\mu \in \Measures(X)$ 
(or PPP($\mu$) for short) is a random counting measure $N:\BorelSets(X) \to \N \cup \{0\} \cup \{\infty\}$ such that 
\begin{itemize}
\item for any $A \in \BorelSets(X)$, $N(A)$ has Poisson($\mu(A)$) distribution.
By convention, $N(A) = \infty$ a.s.\ if $\mu(A) = \infty$.
  
\item If $A_{1}, A_{2}, \ldots, A_{k} \in \BorelSets$ are disjoint,
the random variables $N(A_{1}), \ldots$, \linebreak $N(A_{k})$ are independent.
\end{itemize}
\end{definition}


For general $\sigma$-finite intensity measures, we can construct $N$ by superposition.  
Suppose that $X = \bigcup_{i} X_{i}$ where the $X_{i}$ are disjoint and $\mu(X_{i}) < \infty$.  
Use the recipe given at the start of this section to construct,
independently,  a PPP($\mu|_{X_{i}}$) $N_{i}$ on each subspace $X_{i}$.  Then 
$N(A) = \sum_{i=1}^{\infty}N_{i}(A)$ is the desired measure.  
It is purely atomic, and the atoms $Y_1, Y_2, \ldots$ are called
the points of the process.  In applications it is useful to know moments and 
Laplace transforms of functionals of the process.  

\begin{lemma}\label{l:PPPmomlaplacepalm} \hfill
\begin{enumerate}
\item First moment: If $f \geq 0$ or $f \in L^{1}(\mu)$ then 
\[
\E \Bigl[ \sum_i f(Y_i) \Bigr] = \int_{X} f(y) \mu(\mathrm{d} y)
\]
(we agree that both sides  can be $\infty$).
\item Campbell's formula: If $f \geq 0$ or $1 - e^{-f} \in L^{1}(\mu)$ then
\begin{equation*}
 \E\Bigl[\exp\Bigl(-\sum_i f(Y_i)\Bigr)\Bigr] = \exp \Bigl(-\int_X (1-e^{-f(y)}) \mu(\mathrm{d} y)\Bigr)
\end{equation*}
(we agree that $\exp(-\infty)=0$).
\item Palm's formula: Let $\tilde{\Measures}(X) \subset \Measures(X)$ denote the space of point measures on $X$; let $G:X \times \tilde{\Measures} \to \R_{+}$ 
be a measurable functional of the points; and suppose $f$ is as in (2).  Then
\begin{equation*}
\E\Bigl[ \sum_i f(Y_i) G(Y_i, N)\Bigr] = \int_{X}\E[G(y, \delta_{y} + N)]f(y)\mu(\mathrm{d} y).
\end{equation*}
\end{enumerate}
\end{lemma}

The formulation here is that of Lemma 2.3 of \cite{BertoinCoagFrag}.  We include sketch proofs to give a flavor of the calculations involved.
  
\begin{proof}

Let $f = \sum_{k=1}^{n}c_{k}\Indi{A_{k}}$, be a simple function with $\mu(A_{k}) < \infty$.

(1) We have
\begin{equation}
\E\Bigl[ \sum_i f(Y_i) \Bigr] = \E\Bigl[ \sum_{k=1}^{n} c_{k} N(A_{k}) \Bigr] = \sum_{k=1}^{n} c_{k} \mu(A_{k}) 
= \int_X f(y) \mu(\mathrm{d} y).
\end{equation}

(2) We have
\be
\begin{split}
\E&\Bigl[\exp\Bigl(-\sum_i f(Y_i)\Bigr)\Bigr] = \E\Bigl[e^{-\sum_{k} c_{k} N(A_{k})}\Bigr] = \prod_{k=1}^{n} \E\left[e^{-\sum_{k} c_{k} N(A_{k})}\right] \\
& = \prod_{k=1}^{n} \exp(-\mu(A_{k}) (1 - e^{-c_{k}})) = \exp \Bigl(-\int_X (1-e^{-f(y)}) \mu(\mathrm{d} y)\Bigr).
\end{split}
\ee

Both (1) and (2) extend to measurable $f \geq 0$ using standard arguments, which we omit. 
Part (1) for $f \in L^{1}(\mu)$ follows immediately.  
Part (2) for $1 - e^{-f} \in L^{1}(\mu)$ is also omitted.

(3) First suppose $G$ is of the form $G(N) = \exp(-\sum_i g(Y_i))$ for some non-negative measurable $g$. 
Campbell's formula gives, for $q \geq 0$,
\begin{equation}
\E\Bigl[\exp\Bigl(-q\sum_{i}f(Y_{i})\Bigr) G(N)\Bigr]
 = \exp\Bigl( - \int_{X} (1 - e^{-qf(y) - g(y)}) \mu(\mathrm{d} y) \Bigr).
\end{equation}
Differentiating this identity in $q$ at 0 gives 
\begin{equation}
\begin{split}
\E\Bigl[\sum_{i}f(Y_{i})G(N)\Bigr]
 & = \int_{X} f(y) e^{-g(y)} \mu(\mathrm{d} y) \exp\Bigl( - \int_{X} (1 - e^{-g(y)}) \mu(\mathrm{d} y) \Bigr) \\
 & = \int_{X} f(y) e^{-g(y)} \mu(\mathrm{d} y) \E\Bigl[\exp\Bigl(-\sum_i g(Y_i) \Bigr)\Bigr]\\
 & = \int_{X} f(y)  \E\Bigl[\exp\Bigl(-\sum_i g(Y_i) - g(y) \Bigr)\Bigr] \mu(\mathrm{d} y) \\
 & = \int_{X} f(y)  \E[G(N + \delta_{y})] \mu(\mathrm{d} y),
\end{split}
\end{equation}
where Campbell's formula is used to get the second and last lines.

Now, suppose $G(y,N) = \sum_{k=1}^{n} c_{k }\Indi{y \in A_{k}} \exp(-\sum_ig_{k}(Y_i))$
for $A_{1}, \ldots, A_{n} \in \BorelSets$ and measurable $g_{k}:X \to [0,\infty)$.   By linearity, the preceding calculations 
give

\begin{equation}
\begin{split}
\E\Bigl[ \sum_i f(Y_i) G(Y_i, N)\Bigr] & = \int_{X} \sum_{k=1}^{n} c_{k} \Indi{y \in A_{k}} f(y)  \E\Bigl[\exp \Bigl(-\sum_i g_k(Y_i) - g_k(y) \Bigr)\Bigr] \mu(\mathrm{d} y) \\
 & = \int_{X}f(y)\E[G(y,N + \delta_{y})] \mu(\mathrm{d} y).
\end{split}
\end{equation}

From here it is a standard monotone class argument.  
\end{proof}

\subsubsection{The Poisson-Dirichlet distribution via a PPP} 

Consider the PPP with intensity measure given by $\eta(\mathrm{d} x) = \theta x^{-1} \exp(-x)\mathrm{d} x$ on $[0,\infty)$. (Note that $\eta$ is an infinite measure, but is $\sigma$-finite since $\eta(2^{-k-1},2^{-k}] \leq \theta$.)
A practical way to construct this process is given in Tavar\'e \cite{MR896431}.  Let $T_{1} < T_{2} <  \ldots$ be the points 
of a Poisson counting process of rate $\theta$ (that is, the differences $T_{i+1} - T_{i}$
are independent exponential variables of rate $\theta$) and $E_{1}, E_{2}, \ldots$
be exponentially distributed with rate 1.  Then, the points in our PPP($\eta$) 
can be expressed as $\xi_{i} = \exp(-T_{i}) E_{i}$, $i \geq 1$.

\bfig
\centerline{
\includegraphics[width=100mm]{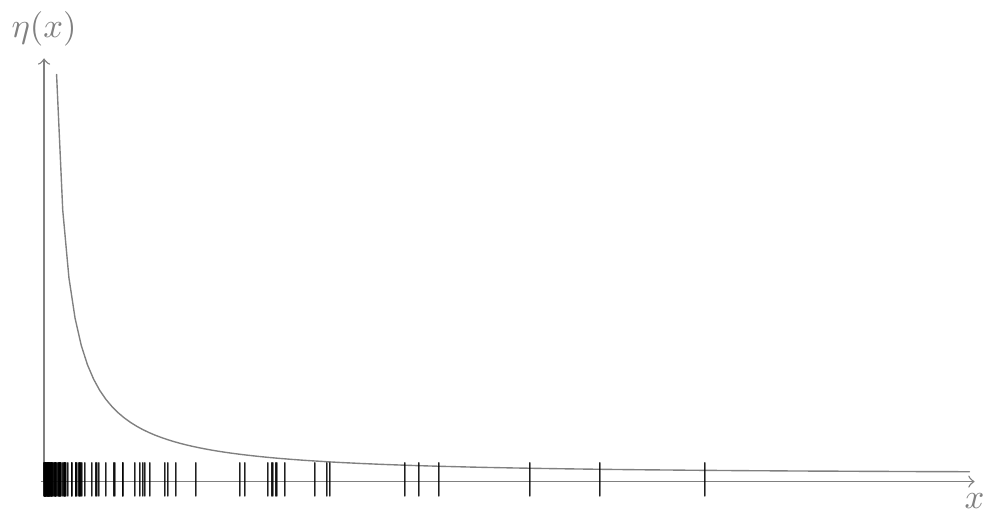}
}
\caption{A sample of the Poisson Point Process (points marked by bars) with intensity measure $\eta$ (overlaid in gray).  Note that the points are dense around the origin.}
\efig

The probability that all points are less than $K > 0$ is 
\begin{equation}
\Prob(N(K, \infty) = 0) = \exp\Bigl(-\int_{K}^{\infty} \theta x^{-1} \exp(-x)\mathrm{d} x \Bigr) \to 1
\end{equation}
as $K \to \infty$.  Thus, there is a largest point and we can order the points in decreasing order so that
$\xi_1 \ge \xi_2 \ge \ldots \ge 0$.  
The sum $\sum_{i=1}^{\infty} \xi_{i}$ is finite 
almost surely.  Indeed, we can say much more.  Recall that the Gamma($\gamma, \lambda$) distribution has density
\[
\frac{1}{\Gamma(\gamma)} \lambda^{\gamma} x^{\gamma-1} \exp(-\lambda x).
\]
\begin{lemma} We have
\[
\sum_{i=1}^{\infty} \xi_{i} \sim \mathrm{Gamma}(\theta,1).
\]
\end{lemma}

\begin{proof}
Since $\sum_i \xi_i$ is a non-negative random variable, its distribution is determined by its Laplace transform. By Campbell's formula, this is given by
\be
\begin{split}
\E\Bigl[\exp \Bigl(-r\sum_{i} \xi_{i} \Bigr)\Bigr] &= \exp\Bigl(-\theta \int_{0}^{\infty} (1-e^{-rx}) x^{-1} \exp(-x)\mathrm{d} x \Bigr) \\
& = \exp\Bigl(-\theta \int_{0}^{r} \int_{0}^{\infty} \exp(-x(1+r))\mathrm{d} x \mathrm{d} r \Bigr) \\
& = (1 + r)^{-\theta},
\end{split}
\ee
for $|r| < 1$, implying that $\sum_{i} \xi_{i}$ is Gamma($\theta,1$) distributed.  
\end{proof}

The Poisson-Dirichlet($\theta$) distribution, $\PDlaw_{\theta} \in \Measures(\PSs)$, is the law of the ordered points, normalised by their sum, i.e.
\begin{equation}\label{e:orderednormedvect}
\frac{1}{\sum_{i} \xi_{i}}\left(\xi_{1}, \xi_{2}, \xi_{3}, \ldots \right).
\end{equation}

In the next section, we will wish to appeal to various properties of Beta and Gamma random variables which are often known collectively as the ``Beta-Gamma algebra''.  Recall that the Beta($a,b$) distribution has density $\frac{\Gamma(a+b)}{\Gamma(a) \Gamma(b)} t^a (1-t)^b$ on $[0,1]$.

\begin{lemma} \label{lem:betagamma}
Suppose that $\Gamma^{\lambda}_{\alpha} \sim \mathrm{Gamma}(\alpha,\lambda)$ and $\Gamma_{\beta}^{\lambda} \sim \mathrm{Gamma}(\beta,\lambda)$ are independent.  Then 
\begin{itemize}
\item $\Gamma_{\alpha}^{\lambda} + \Gamma_{\beta}^{\lambda} \sim \mathrm{Gamma}(\alpha+\beta,\lambda)$, 
\item $\Gamma_{\alpha}^{\lambda}/(\Gamma_{\alpha}^{\lambda} + \Gamma_{\beta}^{\lambda}) \sim \mathrm{Beta}(\alpha,\beta)$,
\item The two random variables above are independent.
\end{itemize}
\end{lemma}

Note that the converse also follows: if $B \sim \mathrm{Beta}(\alpha,\beta)$ is independent of $\Gamma_{\alpha+\beta}^{\lambda} \sim \mathrm{Gamma}(\alpha+\beta,\lambda)$ then $B \Gamma_{\alpha+\beta}^{\lambda} \sim \mathrm{Gamma}(\alpha,\lambda)$, $(1-B) \Gamma_{\alpha+\beta}^{\lambda} \sim \mathrm{Gamma}(\beta, \lambda)$ and these last two random variables are independent.

\begin{proof}
In order to simplify the notation, let $X= \Gamma_{\alpha}^{\lambda}$ and $Y = \Gamma_{\beta}^{\lambda}$.  We will find the joint density of $S = X+Y$ and $R = X/(X+Y)$.  We first find the Jacobian corresponding to this change of variables: we have
\begin{alignat*}{3}
\frac{\partial x}{\partial s} &= r, & \frac{\partial x}{\partial r} &= s  \\
\frac{\partial y}{\partial s} &= 1-r, \quad&  \frac{\partial y}{\partial r} &= -s 
\end{alignat*}
and so the Jacobian is $|-rs - (1-r)s| = s$.  Noting that $X = RS$ and $Y = (1-R)S$, we see that $S$ and $R$ have joint density
\begin{equation}
\begin{split}
&s \frac{1}{\Gamma(\alpha)} \lambda^{\alpha} (rs)^{\alpha-1} e^{-\lambda rs} \frac{1}{\Gamma(\beta)}((1-r)s)^{\beta-1} e^{-\lambda (1-r)s} \\
&= \frac{1}{\Gamma(\alpha + \beta)} \lambda^{\alpha+\beta} s^{\alpha + \beta -1} e^{-\lambda s} \cdot
\frac{\Gamma(\alpha + \beta)}{\Gamma(\alpha) \Gamma(\beta)} r^{\alpha-1} (1-r)^{\beta-1}.
\end{split}
\end{equation}
Since this factorizes with the factors being the correct Gamma and Beta densities, the result follows.
\end{proof}

In the next lemma, we will see the power of the Beta-Gamma algebra.  We use it to make a connection between our two different representations of the Poisson-Dirichlet distribution.  This will serve as a warm up for the calculations in the next section.

\begin{lemma} \label{lem:sbpick}
Suppose that $P = (P_1,P_2,\ldots) \sim \mathrm{PD}_{\theta}$.  Let $P_{*}$ be a size-biased pick from amongst $P_1, P_2, \ldots$.  Then $P_* \sim \mathrm{Beta}(1,\theta)$.
\end{lemma}

So $P_*$ has the same distribution as the length of the first stick in the stick-breaking construction. 

\begin{proof}
Note that, conditional on $P_1, P_2, \ldots$, we have that
\begin{equation}
P_* = P_i \quad \text{ with probability $P_i$, $i \ge 1$}.
\end{equation}
In order to determine the distribution of $P_*$, it suffices to find $\E[f(P_*)]$ for all bounded measurable test functions $f: [0,1] \to \R_+$.  (Indeed, it would suffice to find $\E[f(P_*)]$ for all functions of the form $f(x) = \exp(-qx)$ i.e.\ the Laplace transform. However, our slightly unusual formulation will generalize better when we consider random variables on $\PSs$ in the next section.) Conditioning on $P_1, P_2, \ldots$ and using the Tower Law we see that
\begin{equation}
\E [f(P_*)] = \E[ \E[ f(P_*) | P_1, P_2, \ldots]] = \E \Bigl[ \sum_{i =1}^{\infty} P_i f(P_i) \Bigr].
\end{equation}
Now use the representation (\ref{e:orderednormedvect}) to see that this is equal to 
\begin{equation}
\E \Bigl[ \sum_{i=1}^{\infty} \frac{\xi_i}{\sum_{j=1}^{\infty} \xi_j} f \Bigl(  \frac{\xi_i}{\sum_{k=1}^{\infty} \xi_k} \Bigr) \Bigr].
\end{equation}
This is in a form to which we can apply the Palm formula; we obtain
\begin{equation}
\E \Bigl[ \int_0^{\infty} \frac{y}{y+ \sum_{i=1}^{\infty} \xi_i} f \Bigl(  \frac{y}{y+ \sum_{j=1}^{\infty} \xi_j} \Bigr) \theta y^{-1} e^{-y} \mathrm{d} y\Bigr].
\end{equation}
After cancelling $y$ and $y^{-1}$, we recognise the density of the Exp(1) (= Gamma(1,1)) distribution and so we can write
\begin{equation}
\E \Bigl[ \frac{\theta}{\Gamma+ \sum_{i=1}^{\infty} \xi_i} f \Bigl(  \frac{\Gamma}{\Gamma+ \sum_{j=1}^{\infty} \xi_j} \Bigr) \Bigr],
\end{equation}
where $\Gamma \sim \mathrm{Exp}(1)$ is independent of $\xi_1, \xi_2, \ldots$.  Recall that $\sum_{i=1}^{\infty} \xi_i \sim \mathrm{Gamma}(\theta,1)$.  Then by Lemma~\ref{lem:betagamma}, $\Gamma + \sum_{i=1}^{\infty} \xi_i$ has a $\mathrm{Gamma}(\theta + 1,1)$ distribution and is independent of $\Gamma/( \Gamma + \sum_{i=1}^{\infty} \xi_i)$, which has a Beta($1,\theta$) distribution.  Hence, we get
\begin{equation}
\E \Bigl[ \frac{\theta}{\Gamma+ \sum_{i=1}^{\infty} \xi_i} \Bigr] \E \bigl[ f(B) \bigr],
\end{equation}
where $B \sim \mathrm{Beta}(1,\theta)$.  We conclude by observing that 
\be
\E\Bigl[ \frac{\theta}{\Gamma+ \sum_{i=1}^{\infty} \xi_i} \Bigr] = 1. \qedhere
\ee
\end{proof}

We close this section by noting an important property of the PPP we use to create the Poisson-Dirichlet vector.

\begin{lemma} \label{lem:independenceofsum}
The random variable $\sum_{i = 1}^{\infty} \xi_i$ is independent of
\[
\frac{1}{\sum_{i} \xi_{i}}\left(\xi_{1}, \xi_{2}, \xi_{3}, \ldots \right).
\]
\end{lemma}

This is another manifestation of the independence in the Beta-Gamma algebra; see \cite{Kingman}.

\subsection{Split-merge invariance of Poisson-Dirichlet} \label{ss:splitmergeinvariance}

We use the method that we exploited in the proof of Lemma~\ref{lem:sbpick} to prove part (a) of Theorem~\ref{thm:invariance}.

First define a random function $F: \PSs \to \PSs$ corresponding to $(\beta_{\rm s}, \beta_{\rm m})$ split-merge 
as follows.  Fix $p \in \PSs$ and let $I(p)$ and $J(p)$ be 
the indices of the two independently size-biased parts of $p$, that is
\begin{equation}
\Prob(I(p) = k) = \Prob(J(p) = k) = p_k, \quad k \ge 1.
\end{equation}
Now let $U$ and $V$ be independent $\text{U}(0,1)$ random variables, independent of $I(p)$ and $J(p)$.  Let
\begin{equation}
F(p) = 
\begin{cases}
S^{U}_{i}p & \text{if $I(p) = J(p) = i$ and $V \le \beta_{\rm s}$} \\
M_{ij}p & \text{if $I(p) =i \neq J(p) = j$ and $V \le \beta_{\rm m}$} \\
p & \text{otherwise}.
\end{cases}
\end{equation}

We wish to prove that if $P \sim \text{PD}_{\theta}$ then $F(P) \sim \text{PD}_{\theta}$ also.  Let $g: \PSs \to \R_+$ be a bounded measurable test function which is symmetric in its arguments (this just means that we can forget about ordering the elements of our sequences).  Then, conditioning on $P$, considering the different cases and using the Tower Law, we have
\begin{equation}
\begin{split}
\E[g(F(P))] = & \E \Bigl[\E \Bigl[ \Indi{V \le \beta_{\rm s}} \sum_{i=1}^{\infty} \Indi{I(P) = J(P) = i}  g(S_i^U P) \Big| P \Bigr] \Bigr] \\
&+ \E \Bigl[\E \Bigl[ \Indi{V > \beta_{\rm s}} \sum_{i=1}^{\infty} \Indi{I(P) = J(P) = i}  g(P) \Big| P \Bigr] \Bigr]  \\
&+ \E \Bigl[ \E \Bigl[ \Indi{V \le \beta_{\rm m}} \sum_{i\neq j} \Indi{I(P) = i} \Indi{J(P) = j}  g(M_{ij} P) \Big| P \Bigr] \Bigr]  \\
&+ \E \Bigl[ \E \Bigl[ \Indi{V > \beta_{\rm m}} \sum_{i\neq j} \Indi{I(P) = i} \Indi{J(P) = j}  g(P) \Big| P \Bigr] \Bigr].
\end{split}
\end{equation}
Note that, conditional on $P$, $I(P) = i, J(P) = j$ with probability $P_iP_j$, so that we get
\begin{equation}
\begin{split}
\E[g(F(P))] = & \beta_{\rm s} \E \Bigl[ \sum_{i=1}^{\infty} P_i^2 g(S_i^U P) \Bigr] + (1-\beta_{\rm s}) \E \Bigl[ \sum_{i=1}^{\infty} P_i^2 g(P) \Bigr] \\
& + \beta_{\rm m} \E \Bigl[ \sum_{i \neq j} P_i P_j g(M_{ij} P) \Bigr] + (1 - \beta_{\rm m}) \E \Bigl[ \sum_{i \neq j} P_i P_j g(P)\Bigr].
\end{split}
\end{equation}
Now use the symmetry of $g$ to write
\be
g(S_k^U P) = g \left((P_k U, P_k(1-U), (P_i)_{i \ge 1, i \neq k}) \right)
\ee
and
\be
g(M_{ij}P) = g \left((P_i + P_j, (P_k)_{k \ge 1, k \neq i,j}) \right).
\ee
Set $(P_1, P_2, \ldots) = \frac{1}{\sum_{i=1}^{\infty} \xi_i}(\xi_1, \xi_2, \ldots)$ as in (\ref{e:orderednormedvect}) to obtain
\begin{equation}
\begin{split}
\E [g(F(P))] & = \beta_{\rm s} \E \biggl[ \sum_{k=1}^{\infty} \frac{\xi_k^2}{\bigl(\sum_{i=1}^{\infty} \xi_i \bigr)^2} g \Bigl(\frac{1}{\sum_{i=1}^{\infty} \xi_i} (\xi_k U, \xi_k(1-U), (\xi_i)_{i \ge 1, i \neq k}) \Bigr) \biggr]   \\
& \quad + (1 - \beta_{\rm s}) \E \biggl[\sum_{k=1}^{\infty} \frac{\xi_k^2}{\left(\sum_{i=1}^{\infty} \xi_i \right)^2} g\Bigl(\frac{1}{\sum_{i=1}^{\infty} \xi_i}(\xi_i)_{i \ge 1}\Bigr) \biggr]  \\
& \quad + \beta_{\rm m} \E \biggl[ \sum_{i \neq j} \frac{\xi_i \xi_j}{\left(\sum_{k=1}^{\infty} \xi_k \right)^2} g \Bigl(\frac{1}{\sum_{k=1}^{\infty} \xi_k} (\xi_i + \xi_j, (\xi_k)_{k \ge 1, k \neq i,j}) \Bigr) \biggr]  \\
& \quad + (1 - \beta_{\rm m}) \E \bigg[ \sum_{i \neq j} \frac{\xi_i \xi_j}{\left(\sum_{k=1}^{\infty} \xi_k \right)^2} g\Bigl( \frac{1}{\sum_{i=1}^{\infty} \xi_i} (\xi_i)_{i \ge 1}\Bigr)\biggr].
\end{split}
\end{equation}
%
%
The Palm formula (Lemma \ref{l:PPPmomlaplacepalm}, (3)) applied to each of the expectations above (twice for the double sums) gives
\begin{equation}
\begin{split}
&\E[g(F(P))] = \! \theta \beta_{\rm s} \! \E\biggl[ \int_0^{\infty} \!\!\!\! \frac{x^{-1} e^{-x} x^2}{\left( x + \sum_{k=1}^{\infty} \xi_k \right)^2} g \Bigl(\frac{1}{x + \sum_{k=1}^{\infty} \xi_k} (xU, x(1-U), (\xi_i)_{i \ge 1}) \Bigr) \mathrm{d} x \biggr] \\
& +  \theta(1-\beta_{\rm s}) \E\biggl[ \int_0^{\infty} \!\! \frac{x^{-1} e^{-x} x^2}{\left( x + \sum_{k=1}^{\infty} \xi_k \right)^2} g \Bigl(\frac{1}{x + \sum_{k=1}^{\infty} \xi_k} (x, (\xi_i)_{i \ge 1}) \Bigr) \mathrm{d} x \biggr] \\
& +  \theta^2\beta_{\rm m} \E\biggl[ \int_0^{\infty} \!\! \int_0^{\infty} \!\! \frac{x^{-1} e^{-x} y^{-1} e^{-y} xy}{\left( x + y + \sum_{k=1}^{\infty} \xi_k \right)^2} g \Bigl(\frac{1}{x + y + \sum_{k=1}^{\infty} \xi_k} (x+y, (\xi_i)_{i \ge 1}) \Bigr) \mathrm{d} x \mathrm{d} y \biggr] \\
& + \! \theta^2 \!(1 - \beta_{\rm m}) \! \E\biggl[ \int_0^{\infty} \!\!\!\! \int_0^{\infty} \!\!\!\! \frac{x^{-1} e^{-x} y^{-1} e^{-y} xy}{\left( x + y + \sum_{k=1}^{\infty} \xi_k \right)^2} g \Bigl(\frac{1}{x + y + \sum_{k=1}^{\infty} \xi_k} (x, y, (\xi_i)_{i \ge 1}) \Bigr) \mathrm{d} x \mathrm{d} y \biggr].
\end{split}
\end{equation}
It helps to recognise the densities we are integrating over here (after cancellation).  In the 
first two expectations, which correspond to split proposals, we have the density $x e^{-x}$
of the Gamma(2,1) distribution.  The other density to appear is $e^{-x}e^{-y}$, 
which corresponds to a pair of independent standard exponential variables.  Using Lemma~\ref{lem:betagamma}, it follows that
\begin{equation}
\begin{split}
& \E[g(F(P))] = \theta \beta_{\rm s} \E\biggl[ \frac{1}{\left(\Gamma + \sum_{k=1}^{\infty} \xi_k \right)^2} g \Bigl(\frac{1}{ \Gamma + \sum_{k=1}^{\infty} \xi_k} (\Gamma U,  \Gamma (1-U), (\xi_i)_{i \ge 1}) \Bigr) \biggr] \\
& +  \theta(1 - \beta_{\rm s}) \E \biggl[ \frac{1}{\left(\Gamma + \sum_{k=1}^{\infty} \xi_k \right)^2} g \Bigl(\frac{1}{ \Gamma + \sum_{k=1}^{\infty} \xi_k} (\Gamma, (\xi_i)_{i \ge 1}) \Bigr) \biggr] \\
& +  \theta^2\beta_{\rm m} \E \biggl[ \frac{1}{\left(\Gamma + \sum_{k=1}^{\infty} \xi_k \right)^2} g \Bigl(\frac{1}{ \Gamma + \sum_{k=1}^{\infty} \xi_k} (\Gamma, (\xi_i)_{i \ge 1}) \Bigr) \biggr] \\
& + \theta^2(1 - \beta_{\rm m})\E \biggl[ \frac{1}{\left(\Gamma + \sum_{k=1}^{\infty} \xi_k \right)^2} g \Bigl(\frac{1}{ \Gamma + \sum_{k=1}^{\infty} \xi_k} (\Gamma U, \Gamma (1-U), (\xi_i)_{i \ge 1}) \Bigr) \biggr],
\end{split}
\end{equation}
where $\Gamma \sim \mathrm{Gamma}(2,1)$, independently of $(\xi_i)_{i \ge 1}$.  By Lemmas~\ref{lem:betagamma} and \ref{lem:independenceofsum}, $\Gamma + \sum_{k} \xi_{k}$ is Gamma$(2 + \theta, 1)$ distributed and independent of the argument of $g$ in all of the above expectations.  
More calculation shows that
\begin{equation}
\E\biggl[\frac{1}{\left(\Gamma + \sum_{k=1}^{\infty} \xi_k \right)^2}\biggr]= \frac{1}{\theta(\theta+1)},
\end{equation}
and so we are left with
\begin{equation}
\begin{split}
\E[g(F(P))] = \frac{\theta \beta_{\rm s} + \theta^2(1- \beta_{\rm m})}{\theta(\theta+1)} \E\biggl[g \Bigl(\frac{1}{ \Gamma + \sum_{k=1}^{\infty} \xi_k} (\Gamma U,  \Gamma (1-U), (\xi_i)_{i \ge 1}) \Bigr)\biggr] \\
+\frac{\theta(1 - \beta_{\rm s}) + \theta^2 \beta_{\rm m}}{\theta(\theta+1)}  \E \biggl[g \Bigl(\frac{1}{ \Gamma + \sum_{k=1}^{\infty} \xi_k} (\Gamma, (\xi_i)_{i \ge 1}) \Bigr) \biggr].
\end{split}
\end{equation}
Next use $\beta_{\rm s} = \theta \beta_{\rm m}$ to get
\begin{equation}
\theta \beta_{\rm s} + \theta^2 (1 - \beta_{\rm m}) = \theta^2 \quad \text{ and } \quad \theta(1 - \beta_{\rm s}) + \theta^2 \beta_{\rm m} = \theta.
\end{equation}
So the expression for $\E[g(F(P))]$ simplifies to
\begin{equation}
\begin{split}
\frac{\theta}{(\theta+1)} \E\biggl[g \Bigl(\frac{1}{ \Gamma + \sum_{k=1}^{\infty} \xi_k} (\Gamma U,  \Gamma (1-U), (\xi_i)_{i \ge 1}) \Bigr)\biggr] \\
+\frac{1}{(\theta+1)}  \E\biggl[g \Bigl(\frac{1}{ \Gamma + \sum_{k=1}^{\infty} \xi_k} (\Gamma, (\xi_i)_{i \ge 1}) \Bigr) \biggr].
\end{split}
\end{equation}
We can re-express this as a sum of expectations as follows:
\begin{equation}
\begin{split}
\frac{1}{\theta(\theta+1)} \E \biggl[\int_0^{\infty} \int_0^{\infty} \theta^2 e^{-x} e^{-y} g \Bigl(\frac{1}{x+y + \sum_{k=1}^{\infty} \xi_k} (x, y, (\xi_i)_{i \ge 1}) \Bigr) \mathrm{d} x \mathrm{d} y \biggr] \\
+ \frac{1}{\theta(\theta+1)} \E \biggl[\int_0^{\infty}  \theta x e^{-x} g \Bigl(\frac{1}{x + \sum_{k=1}^{\infty} \xi_k} (x, (\xi_i)_{i \ge 1}) \Bigr) \mathrm{d} x \mathrm{d} y \biggr].
\end{split}
\end{equation}
Using the Palm formula in the other direction gives
\begin{equation}
\begin{split}
& \frac{1}{\theta(\theta+1)}  \E \biggl[\sum_{i\neq j} \xi_i \xi_j g \Bigl( \frac{1}{\sum_{k=1}^{\infty} \xi_k}( \xi_k)_{k \ge 1} \Bigr) +  \sum_{k=1}^{\infty} \xi_k^2 g \Bigl(\frac{1}{\sum_{k=1}^{\infty} \xi_k}( \xi_k)_{k \ge 1} \Bigr) \biggr] \\
& = \frac{1}{\theta (\theta+1)} \E \biggl[\Bigl( \sum_{k=1}^{\infty} \xi_k \Bigr)^2 g \Bigl(\frac{1}{\sum_{k=1}^{\infty} \xi_k}( \xi_k)_{k \ge 1} \Bigr) \biggr].
\end{split}
\end{equation}

Once again, $\sum_{k=1}^{\infty} \xi_k$ is independent of the argument of $g$.  Moreover, it is easily shown that
\begin{equation}
\E \biggl[ \Bigl( \sum_{k=1}^{\infty} \xi_k \Bigr)^2\biggr] = \theta(\theta+1),
\end{equation}
since it is simply the second moment of a Gamma$(\theta, 1)$ random variable.  Thus, 
\begin{equation}
\E[g(F(P))] = \E[g(P)],
\end{equation}
from which the result follows.

\subsection{Split-merge in continuous time}\label{s:ctstimesplitmerge}

The dynamics in the next section will be in continuous time, so we close this section by describing a continuous time version of the 
split-merge process.  First, consider the standard Poisson counting 
process $(N_{t}, \; t \geq 0)$, perhaps the simplest continuous 
time Markov chain.   Its trajectories take values in $\{0,1,2,\ldots\}$, are piecewise constant, increasing
and right continuous.  At each integer $k$, it is held for an exponentially  distributed 
random time before jumping to $k+1$.  Consequently, only finitely many jumps 
are made during each finite time interval.  We say $N_{t}$ increments at rate 1.  

Continuous time split-merge is the process $(P^{N_{t}}, \; t \geq 0)$ obtained by composing $(P^{k}, \; k = 0,1,2,3,\ldots)$ with an independent Poisson counting process.  It is a Markov process in $\PSs$ 
with the following dynamics.  Suppose the present state is $p \in \PSs$. Attach to each part $p_{i}$ an exponential alarm clock of rate $\beta_{\rm s} p^{2}_{i}$
and to each pair $(p_{i},p_{j})$ of distinct parts a clock of rate $2 \beta_{\rm m} p_{i} p_{j}$.  Wait for the first clock to ring.  If $p_{i}$'s clock rings first then split $p_{i}$ uniformly
(i.e. apply $S^{U}_{i}$ with $U$ uniform).  If the alarm 
for $(p_{i},p_{j})$ rings first then apply $M_{ij}$.  In other words, 
part $p_{i}$ splits uniformly at rate $\beta_{\rm s}p_{i}$ and distinct parts 
$p_{i}$ and $p_{j}$ merge at rate $2 \beta_{\rm m} p_{i} p_{j}$.  Due to the memoryless property of the exponential distribution, once an alarm clock has rung, all of the alarm clocks are effectively reset, and the process starts over from the new state.

More formally, define the rate kernel $Q:\PSs \times \BorelSets(\PSs) \to [0,\infty)$ by
\begin{equation}
 Q(p, \cdot) \eqdef \beta_{\rm s} \sum_{i} p_{i}^{2} \int_{0}^{1} \delta_{S^{u}_{i}p}(\cdot) du + \beta_{\rm m} \sum_{i \neq j} p_{i}p_{j}\delta_{M_{ij}p}(\cdot)
\end{equation}
and the (uniformly bounded) `rate of leaving' $q:\PSs \to [0,\infty)$
\begin{equation}
q(p) \eqdef Q(p,\PSs) = \beta_{\rm s} \sum_{i} p_{i}^{2} + \beta_{\rm m} \sum_{i \neq j} p_{i} p_{j}.
\end{equation}
Using standard theory (e.g. Proposition 12.20, \cite{MR1876169}), there exists a 
Markov process on $\PSs$ that waits for an Exponential($q(p)$) amount of time in state 
$p$ before jumping to a new state chosen according to $Q(p,\cdot)/q(p)$.  
Furthermore, since 
\begin{equation}
K_{\beta_{\rm s},\beta_{\rm m}}(p,\cdot) = Q(p, \cdot) + (1 - q(p))\delta_{p}(\cdot),
\end{equation}
this process is constructed explicitly as $(P^{N_{t}}, \; t \geq 0)$. 
The coincidence of the invariant measures in discrete and continuous time is immediate.

\begin{lemma}
 A measure $\nu \in \Measures(\PSs)$ is invariant for the continuous time process 
 $(P^{N_{t}}, \; t \geq 0)$ if, and only if, it is invariant for $(P^{k}, \; k = 0,1,2,3,\ldots)$.
\end{lemma}

\section{Effective split-merge process of cycles and loops}

This section contains an heuristic argument that connects the loop and cycle models
of section \ref{s:cycleloopmodel} and the split-merge process in section \ref{s:ctstimesplitmerge}.  The heuristic
leads to the conjecture that the asymptotic normalized lengths of the cycles and loops have
Poisson-Dirichlet distribution. By looking at the rates of the effective split-merge process, we can identify the parameter of the distribution.


Consider the cycle or loop model on the cubic lattice $\Lambda_{n} = \{1,\dots,n\}^{d}$ in $\bbZ^{d}$.  
As hinted at in section \ref{s:cycleloopmodel}, we expect that macroscopic cycles emerge for inverse temperatures $\beta$ large enough as $n \to \infty$.  Of course, we believe this also holds for any sequence of sufficiently connected graphs ($\Lambda_{n}$) with diverging number of vertices, but for simplicity we restrict attention to cubic lattices.  
Furthermore, since the same arguments apply to both the cycle and loop models,  we focus on cycles and only mention the modifications for loops when necessary.


Denote by $\lambda^{(i)}$ the length of the $i^{th}$ longest cycle, 
and recall that $\eta_{\mathrm{macro}}(\beta)$ is the fraction of sites lying in macroscopic cycles (see Section \ref{sec prob phase transitions}).

\begin{conjecture}\label{c:criticalbeta}
Suppose $d \geq 3$.  There exists $\beta_c > 0$ such that for $\beta > \beta_c$:
\begin{itemize}
\item[(a)] The fractions of sites in infinite and macroscopic cycles (or loops) approach  the same typical value, and
\[
\eta \eqdef \eta_{\infty}(\beta) = \eta_{\mathrm{macro}}(\beta) >0.
\]
\item[(b)] The vector of ordered normalised cycle lengths
\[
\left( \frac{\lambda^{(1)}}{\eta \, n^{d}}, \frac{\lambda^{(2)}}{\eta \, n^{d} }, \ldots\right)
\]
converges weakly to a random variable $\xi$ in $\PSs$ as $n \to \infty$.
\end{itemize}
\end{conjecture}

Assuming the conjectured result is true, what is the distribution of $\xi$?  
In some related models (the random-cluster model), $\xi$ has been found to be the 
trivial (and non-random!) partition $(1,0,0,\ldots)$.  However, we conjecture that there are \emph{many} macroscopic 
cycles in our model (rather than a unique giant cycle) and that their relative lengths can be 
described explicitly by the Poisson-Dirichlet distribution.  

\begin{conjecture}\label{c:pdcyclelengths}
The distribution of $\xi$ in Conjecture \ref{c:criticalbeta} (b) is $\PDlaw_{\theta}$
for an appropriate choice of $\theta$.
\end{conjecture}

The rest of this section is concerned with justifying this conjecture.  
The reader may guess what the parameter $\theta$ should be.  
We will tease it out below and identify it in 
section \ref{s:punchline}. 

See Section \ref{s:complete} for a summary of rigorous results by Schramm to support 
this conjecture on the complete graph.

\subsection{Burning and building bridges}

Recall that $\bbP_{\Lambda_{n}, \beta, \vartheta}$ denotes the probability measure for either 
the loop or cycle model.  We define an ergodic Markov process on $\SwapConfigs$ with $\Prob_{\Lambda,\beta,\vartheta}$ as invariant measure.
The process evolves by adding or removing bridges to the current configuration.  Conveniently, the effect of such an operation is to either split a cycle or merge two cycles.

\begin{lemma}\label{l:addremswap}
Suppose $\omega \in \Omega$ and $\omega^{\prime}$ is $\omega$ with either a bridge added (i.e. 
$\omega^{\prime} = \omega \cup \{(e,t)\}$ for some $(e,t) \in \mathcal{E} \times [0,\beta]$)
or a bridge removed (i.e. $\omega^{\prime} = \omega - \{(e,t)\}$ for some $(e,t) \in \omega$).
Then $\Cycles(\omega^{\prime})$ is obtained by splitting a cycle or merging two cycles in $\Cycles(\omega)$.
Similarly, $\Loops(\omega^{\prime})$ is obtained by a split or merge in $\Loops(\omega)$.

\end{lemma}

The point is that adding or removing a bridge never causes, for example, several cycles to join, a cycle to split 
into many pieces or the cycle structure to remain unchanged.

\bfig
\centerline{\includegraphics[width=125mm]{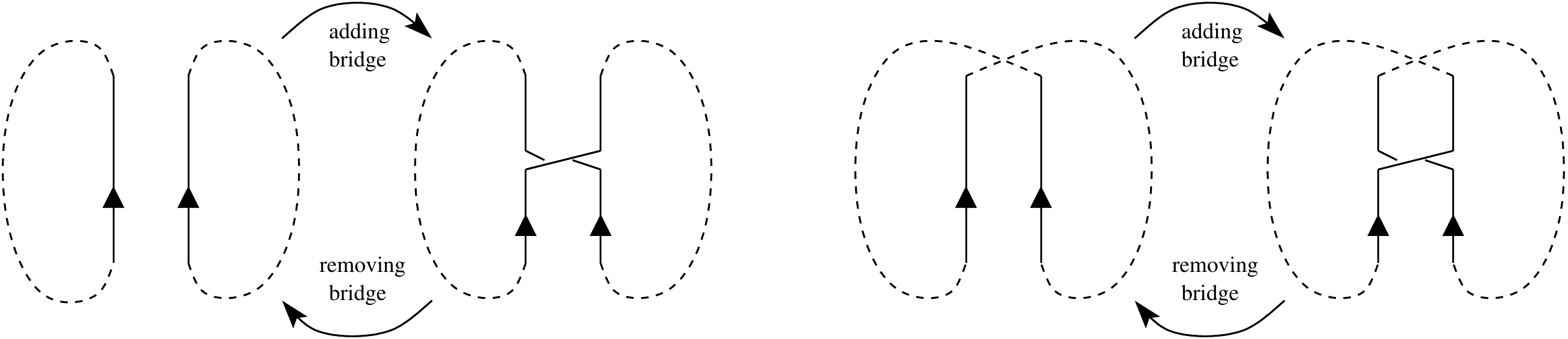}}
\caption{Adding or removing bridges always split or merge cycles. Up to topological equivalence, this figure lists all possibilities.}
\label{fig split merge}
\efig

The Lemma is most easily justified by drawing pictures for the different cases.  Suppose that we add a new bridge.  Either both endpoints of the new bridge belong to the same cycle or two different cycles.  In the former case, the cycle is split and 
we say the bridge is a self-contact.  In the latter case, the two cycles are joined and the bridge is called a contact between 
the two cycles. This is illustrated in Figure\ \ref{fig split merge} for cycles and Figure\ \ref{fig split merge 2} for loops.

Suppose that we remove an existing bridge.  Again,  either both of the bridge's endpoints belong to the same cycle (self-contact) or they are in different cycles (contact between the two cycles).  In the former case, removal splits the cycle and in the latter, the two cycles are joined.

As this argument hints, it is helpful to formally define the `contacts' between cycles.  Suppose 
that $\gamma \in \Cycles(\omega)$ is a cycle.  Recall from Section \ref{sec Poisson conf} that this means
$\gamma(\tau) = (x(\tau), t(\tau))$, $\tau \geq 0$ is a closed trajectory in $\caV \times [0,\beta]_{\rm per}$,
where $x$ is piecewise constant and has a jump discontinuity across the edge 
$e = (x(\tau-), x(\tau)) \in \caE$ at time $\tau$ if, and only if,
the bridge $( e,\; t(\tau))$ is present in $\omega$.  Such bridges are called
self contact bridges,  the set of which is denoted $\contactswaps_{\gamma}$.  Removing a bridge from
$\contactswaps_{\gamma} \subset \omega$ causes $\gamma$ to split.

\bfig
\centerline{\includegraphics[width=120mm]{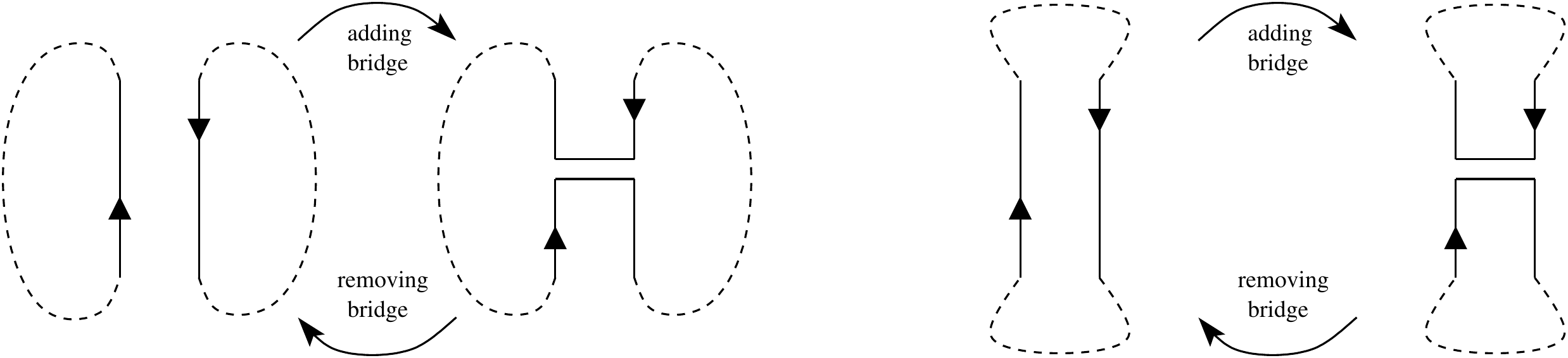}}
\caption{Same as Figure\ \ref{fig split merge}, but for loops instead of cycles.}
\label{fig split merge 2}
\efig

The self contact zone $\contactzone_{\gamma}$ of $\gamma$ is the set of $(e,\tau) \in \caE \times [0, \beta]$ for which
$e = (x(\tau), x(\tau + j \beta))$ for some integer $j$, i.e. the $(e,t)$ 
bridge touches different legs of $\gamma$'s trajectory and so adding a bridge from $\contactzone_{\gamma}$
 splits $\gamma$.

The contact bridges $\contactswaps_{\gamma,\gamma^{\prime}}$ and zones $\contactzone_{\gamma,\gamma^\prime}$ between distinct cycles  $\gamma, \gamma^\prime \in \Cycles(\omega)$ are defined similarly.  
Specifically, $\contactswaps_{\gamma,\gamma^{\prime}} \subset \omega$ is comprised of bridges 
in $\omega$ that are traversed by $\gamma = (x,t)$ and $\gamma^{\prime} = (x^{\prime},t^{\prime})$, i.e. 
$(e,t) \in \omega$ such that $e = (x(t + j_{1}\beta), x^{\prime}(t + j_{2}\beta))$ for some integers $j_{1}, j_{2}$.
Removal of a bridge in $\contactswaps_{\gamma,\gamma^{\prime}}$ causes $\gamma$ and $\gamma^{\prime}$ to merge.

$\contactzone_{\gamma,\gamma^\prime}$ is the set of $(e,t) \in \caE \times [0,\beta]$ such that $e = (x(t + j_{1}\beta), x(t + j_{2}\beta))$ for some $j_{1}, j_{2}$, i.e. those bridges that would merge $\gamma$ and $\gamma^{\prime}$.
Note that the contact (and self contact) zones partition $\caE \times [0, \beta]$ while the contact bridges partition 
$\omega$.

\subsection{Dynamics}\label{s:bridgedynamics}

The promised $\Prob_{\Lambda,\beta,\vartheta}$-invariant Markov process, denoted \linebreak $(\addremswapproc_{t})_{t \geq 0}$
is defined as follows.   Suppose that $\alpha > 0$.  
\begin{itemize}
 \item A new bridge appears in $(e,dt)$ at rate $\vartheta^{\alpha} dt$ if its appearance causes a cycle to split 
and at rate $\vartheta^{-\alpha} dt$ if it causes two cycles to join.
 \item An existing bridge is removed at rate $\vartheta^{1-\alpha}$ if its removal causes a cycle to split 
and at rate $\vartheta^{-(1-\alpha)}$ if its removal causes two cycles to join.
 \item No other transitions occur.
\end{itemize}

The rates are not uniformly bounded, so a little effort is required to check 
$\addremswapproc$ is well behaved (does not `explode').  Accepting this, 
we can show $\addremswapproc$ is actually reversible with respect to our 
cycle model.
   

\begin{lemma}\label{l:muinvariance}
The unique invariant measure of $\addremswapproc$ is $\Prob_{\Lambda,\beta,\vartheta}$.
\end{lemma}

The proof is straightforward and so we omit it.

In the sequel we take $\alpha = 1/2$, so that adding and removing bridges occur at the same rates.

\subsection{Heuristic for rates of splitting and merging of cycles}

As we know, adding or removing bridges causes cycles to split or merge so
the dynamics $(\Cycles(\addremswapproc_t),\;t \geq 0)$ that $\addremswapproc$ induces on cycles is a kind of coagulation-fragmentation process.  However, these dynamics are not Markovian and depend
 on the underlying process in a complicated manner.  Ideally we would like a simpler, more transparent 
description for the dynamics.  The first step towards this is to rewrite 
the transition rates for $\addremswapproc$ in terms of the contact zones and bridges.

Suppose that $\addremswapproc$ is currently in state $\omega \in \SwapConfigs$.  A
 cycle $\gamma \in \Cycles(\omega)$ splits if either a bridge from $\contactzone_{\gamma}$
 is added, or a bridge from  $\contactswaps_{\gamma} \subset \omega$ is removed.  The total
 rate at which these transitions occur is
\begin{equation}\label{e:contactsplitrate}
 \sqrt{\vartheta} \left( |\contactswaps_{\gamma}| + \left|\contactzone_{\gamma}\right|\right),
\end{equation}
where 
$\left|\contactzone_{\gamma}\right| = \sum_{e \in \caE} \Leb(\{t \in [0,\beta]: (e,t) \in \contactzone_{\gamma} \})$
is the (one-dimensional) Lebesgue measure of the self contact zone.  Two distinct cycles $\gamma$ and $\gamma^\prime$ merge if a bridge from $\contactzone_{\gamma,\gamma^\prime}$ is added or one from 
$\contactswaps_{\gamma,\gamma^\prime}$ removed.  The combined rate is
\begin{equation}\label{e:contactmergerate}
 \sqrt{\vartheta}^{-1} \left( |\contactswaps_{\gamma,\gamma^\prime}| +  \left|\contactzone_{\gamma,\gamma^\prime}\right| \right),
\end{equation}
where $\left|\contactzone_{\gamma,\gamma^{\prime}}\right| = \sum_{e \in \caE} \Leb(\{t \in [0,\beta]: (e,t) \in \contactzone_{\gamma,\gamma^{\prime}} \})$.

\subsubsection{Heuristics}
We believe that, for suitably connected graphs and large enough $\beta$, cycles should be macroscopic.  
The trajectories of these cycles should spread evenly over all edges and vertices in the graph.  In particular, macroscopic cycles should come into contact with each other many times and we expect some averaging 
phenomenon to come into play.  The longer a cycle is, on average, the more intersections with other cycles it should have.
In particular, we believe the contact zone between two macroscopic cycles should have size proportional to the cycles' length.


That is, if $\gamma$ and $\gamma^\prime$ are cycles with lengths $\lambda$ and $\lambda^\prime$ respectively then
there is a `law of large numbers'
\begin{equation}\label{e:heursplitrate}
|\contactzone_{\gamma} | \sim \frac{1}{2} c_2 \lambda^2,\; \; |\contactswaps_{\gamma}| \sim \frac{1}{2} c_1 \lambda^2
\end{equation}
and
\begin{equation}\label{e:heurmergerate}
 |\contactzone_{\gamma,\gamma^\prime} | \sim c_2 \lambda \lambda^\prime, \; |\contactswaps_{\gamma,\gamma^\prime}| \sim c_1 \lambda \lambda^\prime,
\end{equation}
for constants $c_1$ and $c_2$ (the notation $X \sim Y$ means that the ratio of the random variables 
converges to $1$ in probability as $\Lambda_{n}$ grows).  

The constants may depend on $\vartheta$ and $\beta$ and the graph geometry.  We believe they are linear in $\beta$
but do not depend on $\vartheta$.  
Note that the size of the contact zones can be calculated easily for the complete graph.  We get
\begin{equation}
|\contactzone_{\gamma,\gamma^\prime} | = \beta \lambda \lambda^\prime , \; |\contactzone_{\gamma}| = \tfrac{\beta}{2} \lambda (\lambda - 1).
\end{equation}
In the case $\vartheta = 1$, we also have numerical support for 
\begin{equation}
  |\contactswaps_{\gamma,\gamma^\prime}| \sim \beta \lambda \lambda^\prime, \; |\contactswaps_{\gamma}| \sim \tfrac{\beta}{2} \lambda^2.
\end{equation}

\subsection{Connection to uniform split-merge}\label{s:punchline}
Continuing with the heuristic,  $\Cycles(\addremswapproc)$ is `nearly' a Markov process in which cycles split and merge.  Substituting \eqref{e:heursplitrate}  into \eqref{e:contactsplitrate} and \eqref{e:heurmergerate} into \eqref{e:contactmergerate},
and multiplying by $2\sqrt{\vartheta} (c_{1} + c_{2})$ (which just changes the speed of the process, not its invariant measure) we see that a cycle of length $\lambda$ splits at rate $\vartheta \lambda^{2}$, while two cycles with lengths $\lambda$ and $\lambda^{\prime}$ merge at
rate $2 \lambda \lambda^{\prime}$.   There seems no reason to suppose that splits are not uniform. 

Suddenly there are many similarities between $\Cycles(\addremswapproc)$ and 
the continuous time split-merge process of section \ref{s:ctstimesplitmerge}.  This 
suggests that Poisson-Dirichlet $\PDlaw_{\theta}$ is lurking somewhere in 
the normalised cycle length distribution.  What is the right choice of the parameter $\theta$?  

Write $\vartheta = \beta_{\rm s}/\beta_{\rm m}$, $\beta_{\rm s}, \beta_{\rm m} \in (0,1]$ and multiply the rates by $\beta_{\rm m}$
to see that a cycle of length $\lambda$ splits uniformly at rate $\beta_{\rm s} \lambda^{2}$, while two 
cycles with lengths $\lambda$ and $\lambda^{\prime}$ merge at
rate $2 \beta_{\rm m} \lambda \lambda^{\prime}$.  Up to the normalising factor (which is close to the constant $\eta_{\rm macro} |\Lambda_{n}|$), these are exactly the rates in section \ref{s:ctstimesplitmerge}. Thus, the parameter $\theta$ should be equal to
$\vartheta$. This fact was initially not obvious.


\bibliography{AZ10-Goldschmidt-Ueltschi-Windridge}

\begin{thebibliography}{10}

\bibitem{MR1288152}
M.~Aizenman and B.~Nachtergaele.
\newblock Geometric aspects of quantum spin states.
\newblock {\em Comm. Math. Phys.}, 164(1):17--63, 1994.

\bibitem{AldousSurvey}
D.~Aldous.
\newblock Deterministic and stochastic models for coalescence (aggregation and
  coagulation): a review of the mean-field theory for probabilists.
\newblock {\em Bernoulli}, 5(1):3--48, 1999.

\bibitem{2010arXiv1009.3723A}
G.~{Alon} and G.~{Kozma}.
\newblock {The probability of long cycles in interchange processes}.
\newblock http://arxiv.org/abs/1009.3723, 2010.

\bibitem{MR2042369}
O.~Angel.
\newblock Random infinite permutations and the cyclic time random walk.
\newblock In {\em Discrete random walks ({P}aris, 2003)}, Discrete Math. Theor.
  Comput. Sci. Proc., AC, pages 9--16 (electronic). Assoc. Discrete Math.
  Theor. Comput. Sci., Nancy, 2003.

\bibitem{MR2032426}
R.~Arratia, A.~D. Barbour, and S.~Tavar{{\'e}}.
\newblock {\em Logarithmic combinatorial structures: a probabilistic approach}.
\newblock EMS Monographs in Mathematics. European Mathematical Society (EMS),
  Z{\"u}rich, 2003.

\bibitem{Ber10}
N.~Berestycki.
\newblock Emergence of giant cycles and slowdown transition in random
  transpositions and $k$-cycles.
\newblock {\em Electr. J. Probab.}, 16:152--173, 2011.

\bibitem{Berez}
V.~L. Berezinski{\v \i}.
\newblock Destruction of long-range order in one-dimensional and
  two-dimensional systems having a continuous symmetry group i. classical
  systems.
\newblock {\em Soviet J. Exper. Theor. Phys.}, 32:493--500, 1971.

\bibitem{BertoinCoagFrag}
J.~Bertoin.
\newblock {\em Random fragmentation and coagulation processes}, volume 102 of
  {\em Cambridge Studies in Advanced Mathematics}.
\newblock Cambridge University Press, Cambridge, 2006.

\bibitem{BU}
V.~Betz and D.~Ueltschi.
\newblock Spatial random permutations and {P}oisson-{D}irichlet law of cycle
  lengths.
\newblock {\em Electr. J. Probab.}, 16:1173--1192, 2011.

\bibitem{MR2581604}
M.~Biskup.
\newblock Reflection positivity and phase transitions in lattice spin models.
\newblock In {\em Methods of contemporary mathematical statistical physics},
  volume 1970 of {\em Lecture Notes in Math.}, pages 1--86. Springer, Berlin,
  2009.

\bibitem{MR1414838}
C.~Borgs, R.~Koteck{{\'y}}, and D.~Ueltschi.
\newblock Low temperature phase diagrams for quantum perturbations of classical
  spin systems.
\newblock {\em Comm. Math. Phys.}, 181(2):409--446, 1996.

\bibitem{MR1749392}
L.~Chayes, L.~P. Pryadko, and K.~Shtengel.
\newblock Intersecting loop models on {${\bf Z}^d$}: rigorous results.
\newblock {\em Nuclear Phys. B}, 570(3):590--614, 2000.

\bibitem{MR1148752}
J.~G. Conlon and J.~P. Solovej.
\newblock Upper bound on the free energy of the spin {$1/2$} {H}eisenberg
  ferromagnet.
\newblock {\em Lett. Math. Phys.}, 23(3):223--231, 1991.

\bibitem{MR2608122}
N.~Crawford and D.~Ioffe.
\newblock Random current representation for transverse field {I}sing model.
\newblock {\em Comm. Math. Phys.}, 296(2):447--474, 2010.

\bibitem{MR1400181}
N.~Datta, R.~Fern{{\'a}}ndez, and J.~Fr{{\"o}}hlich.
\newblock Low-temperature phase diagrams of quantum lattice systems. {I}.
  {S}tability for quantum perturbations of classical systems with finitely-many
  ground states.
\newblock {\em J. Statist. Phys.}, 84(3-4):455--534, 1996.

\bibitem{MR2044670}
P.~Diaconis, E.~Mayer-Wolf, O.~Zeitouni, and M.~P.~W. Zerner.
\newblock The {P}oisson-{D}irichlet law is the unique invariant distribution
  for uniform split-merge transformations.
\newblock {\em Ann. Probab.}, 32(1B):915--938, 2004.

\bibitem{MR0496246}
F.~J. Dyson, E.~H. Lieb, and B.~Simon.
\newblock Phase transitions in quantum spin systems with isotropic and
  nonisotropic interactions.
\newblock {\em J. Statist. Phys.}, 18(4):335--383, 1978.

\bibitem{MR0125031}
P.~Erd{\H{o}}s and A.~R{{\'e}}nyi.
\newblock On the evolution of random graphs.
\newblock {\em Magyar Tud. Akad. Mat. Kutat{\'o} Int. K{\"o}zl.}, 5:17--61,
  1960.

\bibitem{MR2681767}
W.~G. Faris.
\newblock Outline of quantum mechanics.
\newblock In {\em Entropy and the quantum}, volume 529 of {\em Contemp. Math.},
  pages 1--52. Amer. Math. Soc., Providence, RI, 2010.

\bibitem{MR2663265}
S.~Feng.
\newblock {\em The {P}oisson-{D}irichlet distribution and related topics}.
\newblock Probability and its Applications (New York). Springer, Heidelberg,
  2010.
\newblock Models and asymptotic behaviors.

\bibitem{MR1681462}
G.~B. Folland.
\newblock {\em Real analysis}.
\newblock Pure and Applied Mathematics (New York). John Wiley \& Sons Inc., New
  York, second edition, 1999.
\newblock Modern techniques and their applications, A Wiley-Interscience
  Publication.

\bibitem{MR506363}
J.~Fr{{\"o}}hlich, R.~Israel, E.~H. Lieb, and B.~Simon.
\newblock Phase transitions and reflection positivity. {I}. {G}eneral theory
  and long range lattice models.
\newblock {\em Comm. Math. Phys.}, 62(1):1--34, 1978.

\bibitem{MR570370}
J.~Fr{{\"o}}hlich, R.~B. Israel, E.~H. Lieb, and B.~Simon.
\newblock Phase transitions and reflection positivity. {II}. {L}attice systems
  with short-range and {C}oulomb interactions.
\newblock {\em J. Statist. Phys.}, 22(3):297--347, 1980.

\bibitem{MR0421531}
J.~Fr{{\"o}}hlich, B.~Simon, and T.~Spencer.
\newblock Infrared bounds, phase transitions and continuous symmetry breaking.
\newblock {\em Comm. Math. Phys.}, 50(1):79--95, 1976.

\bibitem{MR634447}
J.~Fr{{\"o}}hlich and T.~Spencer.
\newblock The {K}osterlitz-{T}houless transition in two-dimensional abelian
  spin systems and the {C}oulomb gas.
\newblock {\em Comm. Math. Phys.}, 81(4):527--602, 1981.

\bibitem{MR2360227}
D.~Gandolfo, J.~Ruiz, and D.~Ueltschi.
\newblock On a model of random cycles.
\newblock {\em J. Stat. Phys.}, 129(4):663--676, 2007.

\bibitem{MR1552567}
J.~Ginibre.
\newblock Existence of phase transitions for quantum lattice systems.
\newblock {\em Comm. Math. Phys.}, 14(3):205--234, 1969.

\bibitem{MR2477388}
G.~R. Grimmett.
\newblock Space-time percolation.
\newblock In {\em In and out of equilibrium. 2}, volume~60 of {\em Progr.
  Probab.}, pages 305--320. Birkh{\"a}user, Basel, 2008.

\bibitem{MR0307392}
T.~E. Harris.
\newblock Nearest-neighbor {M}arkov interaction processes on multidimensional
  lattices.
\newblock {\em Advances in Math.}, 9:66--89, 1972.

\bibitem{MR2581610}
D.~Ioffe.
\newblock Stochastic geometry of classical and quantum {I}sing models.
\newblock In {\em Methods of contemporary mathematical statistical physics},
  volume 1970 of {\em Lecture Notes in Math.}, pages 87--127. Springer, Berlin,
  2009.

\bibitem{MR1876169}
O.~Kallenberg.
\newblock {\em Foundations of modern probability}.
\newblock Probability and its Applications (New York). Springer-Verlag, New
  York, second edition, 2002.

\bibitem{MR980167}
T.~Kennedy, E.~H. Lieb, and B.~S. Shastry.
\newblock Existence of {N}{\'e}el order in some spin-{$\frac12$} {H}eisenberg
  antiferromagnets.
\newblock {\em J. Statist. Phys.}, 53(5-6):1019--1030, 1988.

\bibitem{MR0368264}
J.~F.~C. Kingman.
\newblock Random discrete distributions.
\newblock {\em J. Roy. Statist. Soc. Ser. B}, 37:1--15, 1975.
\newblock With a discussion by S. J. Taylor, A. G. Hawkes, A. M. Walker, D. R.
  Cox, A. F. M. Smith, B. M. Hill, P. J. Burville, T. Leonard and a reply by
  the author.

\bibitem{MR591166}
J.~F.~C. Kingman.
\newblock {\em Mathematics of genetic diversity}, volume~34 of {\em CBMS-NSF
  Regional Conference Series in Applied Mathematics}.
\newblock Society for Industrial and Applied Mathematics (SIAM), Philadelphia,
  Pa., 1980.

\bibitem{Kingman}
J.~F.~C. Kingman.
\newblock {\em Poisson processes}, volume~3 of {\em Oxford Studies in
  Probability}.
\newblock The Clarendon Press Oxford University Press, New York, 1993.
\newblock Oxford Science Publications.

\bibitem{0022-3719-6-7-010}
J.~M. Kosterlitz and D.~J. Thouless.
\newblock Ordering, metastability and phase transitions in two-dimensional
  systems.
\newblock {\em Journal of Physics C: Solid State Physics}, 6(7):1181, 1973.

\bibitem{MR1902841}
E.~Mayer-Wolf, O.~Zeitouni, and M.~P.~W. Zerner.
\newblock Asymptotics of certain coagulation-fragmentation processes and
  invariant {P}oisson-{D}irichlet measures.
\newblock {\em Electron. J. Probab.}, 7:no. 8, 25 pp. (electronic), 2002.

\bibitem{Mermin:1966lr}
N.~D. Mermin and H.~Wagner.
\newblock Absence of ferromagnetism or antiferromagnetism in one- or
  two-dimensional isotropic heisenberg models.
\newblock {\em Phys. Rev. Lett.}, 17(22):1133--1136, Nov 1966.

\bibitem{MR2310198}
B.~Nachtergaele.
\newblock Quantum spin systems after {DLS} 1978.
\newblock In {\em Spectral theory and mathematical physics: a {F}estschrift in
  honor of {B}arry {S}imon's 60th birthday}, volume~76 of {\em Proc. Sympos.
  Pure Math.}, pages 47--68. Amer. Math. Soc., Providence, RI, 2007.

\bibitem{Neves1986331}
E.~J. Neves and J.~F. Perez.
\newblock Long range order in the ground state of two-dimensional
  antiferromagnets.
\newblock {\em Physics Letters A}, 114(6):331 -- 333, 1986.

\bibitem{MR1930355}
J.~Pitman.
\newblock Poisson-{D}irichlet and {GEM} invariant distributions for
  split-and-merge transformation of an interval partition.
\newblock {\em Combin. Probab. Comput.}, 11(5):501--514, 2002.

\bibitem{PitmanStFl}
J.~Pitman.
\newblock {\em Combinatorial stochastic processes}, volume 1875 of {\em Lecture
  Notes in Mathematics}.
\newblock Springer-Verlag, Berlin, 2006.
\newblock Lectures from the 32nd Summer School on Probability Theory held in
  Saint-Flour, July 7--24, 2002, With a foreword by Jean Picard.

\bibitem{MR1434129}
J.~Pitman and M.~Yor.
\newblock The two-parameter {P}oisson-{D}irichlet distribution derived from a
  stable subordinator.
\newblock {\em Ann. Probab.}, 25(2):855--900, 1997.

\bibitem{MR0289084}
D.~Ruelle.
\newblock {\em Statistical mechanics: {R}igorous results}.
\newblock W. A. Benjamin, Inc., New York-Amsterdam, 1969.

\bibitem{MR2166362}
O.~Schramm.
\newblock Compositions of random transpositions.
\newblock {\em Israel J. Math.}, 147:221--243, 2005.

\bibitem{MR1239893}
B.~Simon.
\newblock {\em The statistical mechanics of lattice gases. {V}ol. {I}}.
\newblock Princeton Series in Physics. Princeton University Press, Princeton,
  NJ, 1993.

\bibitem{MR896431}
S.~Tavar{{\'e}}.
\newblock The birth process with immigration, and the genealogical structure of
  large populations.
\newblock {\em J. Math. Biol.}, 25(2):161--168, 1987.

\bibitem{MR1224836}
B.~T{{\'o}}th.
\newblock Improved lower bound on the thermodynamic pressure of the spin
  {$1/2$} {H}eisenberg ferromagnet.
\newblock {\em Lett. Math. Phys.}, 28(1):75--84, 1993.

\bibitem{Toth}
B.~T{{\'o}}th.
\newblock Reflection positivity, infrared bounds, continuous symmetry breaking.
\newblock
  http://www.math.bme.hu/$\sim$balint/oktatas/statisztikus\_fizika/jegyzet/,
  1996.

\bibitem{tsilevich2001simplest}
N.~Tsilevich.
\newblock {On the simplest split-merge operator on the infinite-dimensional
  simplex}.
\newblock {\em Arxiv preprint math/0106005}, 2001.

\bibitem{MR1751188}
N.~V. Tsilevich.
\newblock Stationary random partitions of a natural series.
\newblock {\em Teor. Veroyatnost. i Primenen.}, 44(1):55--73, 1999.

\end{thebibliography}
\bibliographystyle{abbrv}

\end{document}